\documentclass[letterpaper,11pt]{article}
\usepackage{amsfonts,amsmath,amsthm,amssymb,dsfont,etex}

\usepackage{enumerate}

\usepackage{tikz}
\usepackage{pgffor}
\usetikzlibrary{arrows,automata,backgrounds,calc,positioning,shapes.geometric}
\usepackage{comment}
\usepackage{newtxtext}
\usepackage[margin=1in]{geometry}
\usepackage{soul,color}
\usepackage{graphicx,float,wrapfig}
\usepackage{mathrsfs}
\usepackage[usenames,dvipsnames]{pstricks}
\usepackage{epsfig}
\usepackage{pst-grad} 
\usepackage{pst-plot} 
\usepackage{mathdots}
\usepackage[linesnumbered,boxed,ruled,vlined]{algorithm2e}

\usepackage{bbm}

\usepackage{url}

\usepackage{tabu}

\usepackage{xfrac}

\usepackage[normalem]{ulem}

\newtheorem{theorem}{Theorem}[section]

\newtheorem{lemma}[theorem]{Lemma}

\newtheorem{hypot}[theorem]{Hypothesis}

\newtheorem{cor}[theorem]{Corollary}
\theoremstyle{definition}
\newtheorem{definition}[theorem]{Definition}
\theoremstyle{plain}
\newtheorem{remark}[theorem]{Remark}

\newtheorem{open}{Open Question}

\newenvironment{reminder}[1]{\bigskip
	\noindent {\bf Reminder of #1  }\em}{\smallskip}

\renewcommand{\epsilon}{\varepsilon}
\newcommand{\E}{\mathop{\mathbb{E}}}
\newcommand{\F}{\mathbb{F}}
\newcommand{\poly}{\mathop{\mathrm{poly}}}
\newcommand{\polylog}{\mathop{\mathrm{polylog}}}

\newcommand{\aprod}{p_{\textrm{alt}}}

\newcommand{\hamsim}{\mathop{\mathrm{HamSim}}}
\newcommand{\maxsim}{\mathop{\mathrm{MaxSim}}}

\newcommand{\ik}{i_{(k)}} 
\newcommand{\xa}{x^{\mathrm{a}}}
\newcommand{\xb}{x^{\mathrm{b}}}
\newcommand{\binodd}{{\mathrm{bin}}_{\mathrm{odd}}}

\newcommand{\BQP}{\textsf{BQP}\xspace}

\newcommand{\AM}{\textsf{AM}\xspace}
 
\newcommand{\MA}{\textsf{MA}\xspace}

\newcommand{\PTIME}{\textsf{P}\xspace}

\newcommand{\ETIME}{\textsf{E}\xspace}

\newcommand{\NL}{\textsf{NL}\xspace}

\newcommand{\NP}{\textsf{NP}\xspace}

\newcommand{\PSPACE}{\textsf{PSPACE}\xspace}

\newcommand{\co}{\textsf{co}}
\newcommand{\BP}{\textsf{BP}\xspace}

\newcommand{\IP}{\textsf{IP}\xspace}

\newcommand{\eps}{\epsilon}
\newcommand{\SETH}{\textsf{SETH}\xspace}
\newcommand{\PCP}{\textsf{PCP}\xspace}
\newcommand{\NC}{\textsf{NC}\xspace}

\newcommand{\DTIME}{\textsf{DTIME}}
\newcommand{\NTIME}{\textsf{NTIME}}
\newcommand{\SPACE}{\textsf{SPACE}}
\newcommand{\NSPACE}{\textsf{NSPACE}}

\newcommand{\MAXAPAIR}[1]{\textsf{Max}-#1-\textsf{Pair}\xspace}
\newcommand{\MINAPAIR}[1]{\textsf{Min}-#1-\textsf{Pair}\xspace}
\newcommand{\ASATPAIR}[1]{#1-\textsf{Satisfying-Pair}\xspace}

\newcommand{\MAXLCSPAIR}{\textsf{Closest-LCS-Pair}\xspace}
\newcommand{\MINLCSPAIR}{\textsf{Furthest-LCS-Pair}\xspace}

\newcommand{\NEXP}{\textsf{NEXP}\xspace}

\newcommand{\BPSAT}{\textsf{BP-SAT}\xspace}
\newcommand{\BPSATPAIR}{\textsf{BP-Satisfying-Pair}\xspace}
\newcommand{\ACBPSAT}{\textsf{AC-BP-SAT}\xspace}

\newcommand{\OAPT}{\textsf{OAPT}\xspace}
\newcommand{\OAPTrestri}{Restricted~\textsf{OAPT}\xspace}
\newcommand{\OAPTfull}{\textsf{Orthogonal-Alternating-Product-Tensors}\xspace}
\newcommand{\MINTT}{\textsf{Min-TropSim}\xspace}
\newcommand{\MINTTrestri}{Restricted~\textsf{Min-TropSim}\xspace}
\newcommand{\MAXTT}{\textsf{Max-TropSim}\xspace}
\newcommand{\MAXTTrestri}{Restricted~\textsf{Max-TropSim}\xspace}
\newcommand{\GAPMINTT}[1]{#1-\textsf{Gap-Min-TropSim}\xspace}
\newcommand{\GAPMINTTrestri}[1]{Restricted~#1-\textsf{Gap-Min-TropSim}\xspace}
\newcommand{\GAPMAXTT}[1]{#1-\textsf{Gap-Max-TropSim}\xspace}
\newcommand{\GAPMAXTTrestri}[1]{Restricted~#1-\textsf{Gap-Max-TropSim}\xspace}

\newcommand{\SUPERGAPMAXTT}[1]{#1-\textsf{Super-Gap-Max-TS}\xspace}

\newcommand{\LCS}{\textsf{LCS}\xspace}

\newcommand{\REGEXPSTRPAIR}{\textsf{RegExp-String-Pair}\xspace}
\newcommand{\CLOSESTREGEXPSTRPAIR}{\textsf{Closest-RegExp-String-Pair}\xspace}
\newcommand{\GAPCLOSESTREGEXPSTRPAIR}[1]{#1-\textsf{Gap-Closest-RegExp-String-Pair}\xspace}

\newcommand{\STIPAIR}{\textsf{Subtree-Isomorphism-Pair}\xspace}
\newcommand{\MINLCSTPAIR}{\textsf{Min-LCST-Pair}\xspace}
\newcommand{\MAXLCSTPAIR}{\textsf{Max-LCST-Pair}\xspace}

\newcommand{\equivclass}{\textsf{BP-Pair-Class}\xspace}
\newcommand{\equivclasshard}{\textsf{BP-Pair-Hard}\xspace}

\newcommand{\circuitclass}{\mathcal{C}}
\newcommand{\SAT}{\textsf{SAT}}
\newcommand{\OV}{\textsf{OV}}

\newcommand{\Sub}{\mathop{\textsf{STI}}}
\newcommand{\RSub}{\mathop{\textsf{RSTI}}}

\newcommand{\LCST}{\mathop{\textsf{LCST}}}
\newcommand{\RLCST}{\mathop{\textsf{RLCST}}}
\newcommand{\NNS}{\textsf{NNS}}
\newcommand{\protocol}{\Pi}
\newcommand{\dist}{\textsf{dist}}

\newcommand{\depth}{\mathop{\mathrm{depth}}}

\title{Fine-grained Complexity Meets \IP = \PSPACE}

\def\ShowAuthNotes{1}
\ifnum\ShowAuthNotes=1
\newcommand{\authnote}[2]{\ \\ \textcolor{red}{\parbox{0.9\linewidth}{[{\footnotesize {\bf #1:} { {#2}}}]}}\newline}
\else
\newcommand{\authnote}[2]{}
\fi

\date{}

\author{Lijie Chen\thanks{MIT, \texttt{lijieche@mit.edu}. Supported by an Akamai Fellowship.}\and
	\!Shafi Goldwasser\thanks{MIT and Simons Institute for the Theory of Computing, Berkeley, \texttt{shafi@theory.csail.mit.edu}. Supported by NSF grant CNS-1413920.}\and
	\!Kaifeng Lyu\thanks{Tsinghua University, \texttt{lkf15@mails.tsinghua.edu.cn}. Most of the work was done while visiting MIT.}
	\and
	\!Guy N. Rothblum\thanks{Weizmann Institute of Science, \texttt{rothblum@alum.mit.edu}. Part of this research was done while the author was visiting MIT, supported by NSF grant CNS-1413920.}\and
	\!Aviad Rubinstein\thanks{Stanford University, \texttt{aviad@cs.stanford.edu}. Supported by a Robert N. Noyce Family Faculty Fellowship. Most of the work was done while the author was at Harvard University, supported by a Rabin Postdoctoral Fellowship.}
}

\begin{document}

\clearpage\maketitle
\thispagestyle{empty}
\begin{abstract}
In this paper we study the fine-grained complexity of finding exact and approximate solutions to problems in $\PTIME$. 
Our main contribution is showing reductions from an exact to an approximate solution for a host of such problems. 

As one (notable) example, we show that the $\MAXLCSPAIR$  problem (Given two sets of strings $A$ and $B$, compute exactly the maximum $\LCS(a, b)$ with $(a, b) \in A \times B$) is equivalent to its approximation version (under near-linear time reductions, and with a constant approximation factor). 
More generally, we identify a class of problems, which we call $\equivclass$, comprising both exact and approximate solutions, and show that they are all equivalent under near-linear time reductions.

Exploring this class and its properties, we also show: 
\begin{itemize}
\item Under the $\textsf{NC-SETH}$ assumption (a significantly more relaxed assumption than $\SETH$), solving any of the problems in this class requires essentially quadratic time.
\item Modest improvements on the running time of known algorithms (shaving log factors) would imply that $\NEXP$ is not in non-uniform $\NC^1$.
\item  Finally, we leverage our techniques to show new barriers for deterministic approximation algorithms for $\LCS$. 
\end{itemize}

A very important consequence of our results is that they continue to hold in the \textbf{\emph{data structure setting}}. In particular, it shows that a data structure for \emph{approximate} Nearest Neighbor Search for $\LCS$ ($\NNS_\LCS$) implies a data structure for \emph{exact} $\NNS_\LCS$ and a data structure for answering regular expression queries with essentially the same complexity.

At the heart of these new results is a deep connection between interactive proof systems for bounded-space computations and the fine-grained complexity of exact and approximate solutions to problems in $\PTIME$. In particular, our results build on the proof techniques from the classical $\IP = \PSPACE$ result.

	
\end{abstract}
\addtocounter{page}{-1}
\newpage

\section{Introduction}
The study of the fine-grained hardness of problems in $\PTIME$ is one of the most interesting developments of the last few years in complexity theory. The study was initially aimed at the complexity of {\it exact versions} of important problem
in $\PTIME$, such as Longest Common Subsequence (LCS), Edit Distance, All Pair Shortest Path (APSP), and 3-SUM. This was the natural starting point. There are several main thrusts of the study: establishing equivalence classes of problems that are ``equivalent'' to each other in the sense that a substantial improvement in one would imply a similar improvement in the other;  showing fine-grained hardness under complexity assumptions, most notably the SETH\footnote{The Strong Exponential Time Hypothesis (SETH) states that for every $\eps > 0$ there is a $k$ such that $k$-$\SAT$ cannot be solved in $O((2-\eps)^n)$ time~\cite{IP01-SETH}.}; and showing implications of even slight algorithmic improvements, such as ``shaving-logs'' off algorithms for $\PTIME$ time problems, to circuit lower bounds.

However, for many of these problems, approximate solutions are of interest as well, as they originate in natural problems which arise in pattern matching and bioinformatics~\cite{AVW14,BI15,BackursI16,BGL17,bringman2018multivariate}, dynamic data structures~\cite{Patrascu10,AV14,AV14,henzinger2015unifying,kopelowitz2016higher,abboud2016popular,henzinger2017conditional,GoldsteinKLP17}, graph algorithms~\cite{RV13,GIKW17,Abboud2015tria,krauthgamer2017conditional}, computational geometry~\cite{Bring14,Wil18,david2016complexity,Che18} and machine learning~\cite{BackursIS17}.
Thus, studying
the hardness of the {\it approximation version} of the problems, soon became the next frontier.

There are two ways one can imagine to attack the hardness of approximation of problems.
\begin{enumerate}
\item
{\bf Show approximation hardness  under complexity assumptions. }This has been the approach by the recent breakthrough result of Abboud, Rubinstein and Williams~\cite{ARW17-proceedings} who introduced a ``Distributed $\PCP$'' framework and used it to show tight conditional lower bounds, 
under the SETH assumption,  for several fundamental approximation problems, including approximate Bichromatic Max-Inner Product, Subset Query, Bichromatic $\LCS$ Closest Pair, Regular Expression Matching and Diameter in Product Metrics. 

We remark that the challenge in showing 
tight lower bounds for the hardness on approximation problems in $\PTIME$ in contrast to exact problems, is that the traditional $\PCP$ paradigm can not be applied directly to fine-grained complexity, due to the super-linear size blow up in the constructed $\PCP$ instances~\cite{AroraS98,AroraLMSS98,Dinur07-PCP}, which becomes super-polynomial after reducing to problems in $\PTIME$ (when we  care about the \emph{exact} exponent of the running time, a super-polynomial blow-up is certainly unacceptable).
\item {\bf Show equivalence between the hardness of exact and approximate problems}. Namely, the latter is not substantially easier than the former.
This is in essence the original $\PCP$ methodology for showing the hardness of approximation of intractable problems (e.g. $\NP$-hard problems).\\
This
is the {\bf focus}
of the current paper for $\PTIME$-time problems.
\end{enumerate}

Interestingly, in terms of proof techniques, the key to establish our equivalence results is the application of the classical $\IP = \PSPACE$ result~\cite{LundFKN92,Shamir92}. In particular, our results are established via an application of the efficient $\IP$ communication protocol for low space computation ~\cite{AW09-algebrization}. This demonstrates that the techniques in interactive proofs can be ``scaled down'' to establish better results in fine-grained complexity.
Furthermore, it adds yet another example of the connection between communication complexity and fine-grained complexity (arguably, the most involved one).

\subsection{Exact to Approximate Reduction for Nearest Neighbor Search for $\LCS$}
\label{sec:first-example}

We begin with one of our most interesting results: an equivalence between \textbf{\emph{exact}} and \textbf{\emph{approximate}} Nearest Neighbor Search for $\LCS$.   

\begin{itemize}
	\item Nearest Neighbor Search for $\LCS$ ($\NNS_\LCS$): Preprocess a database $\mathcal{D}$ of $N$ strings of length $D \ll N$, and then for each query string $x$, find $y \in \mathcal{D}$ maximizing $\LCS(x,y)$. 
	
	The approximate version only requires to find $y \in \mathcal{D}$ such that $\LCS(x,y)$ is an $f(D)$-approximation of the maximum value.
\end{itemize}

Approximate data structures for the above problem that support fast preprocessing and queries would be highly relevant for bioinformatics. For the similar $\NNS_{\textsf{Edit-Distance}}$ problem, a breakthrough work of~\cite{ostrovsky2007low} used a metric embedding technique to obtain an $2^{O(\sqrt{\log D \log\log D})}$-approximate data structure with polynomial preprocessing time, $\poly(D,\log n)$ query time. 

In contrast, our first result shows that exact $\NNS_\LCS$ and approximate $\NNS_\LCS$ are essentially equivalent. That is, a similar data structure for \emph{approximate} $\NNS_\LCS$ would directly imply a data structure for \emph{exact} $\NNS_\LCS$ with essentially the same complexity!

\begin{theorem}[Informal]\label{theo:NNS-informal}
	For $D = 2^{(\log N)^{o(1)}}$, suppose there is a data structure for $\NNS_\LCS$ with approximation ratio $2^{(\log D)^{1 - \Omega(1)}}$, then there is another data structure for exact $\NNS_\LCS$ with essentially the same preprocessing time/space and query time.
\end{theorem}

In the following, we first discuss $\MAXLCSPAIR$ (a natural \emph{offline} version of $\NNS_\LCS$) to illustrate our techniques, and then discuss our other results in details.

\subsection{Techniques: Hardness of Approximation in $\PTIME$ via Communication Complexity and the Theory of Interactive Proofs} 

$\MAXLCSPAIR$ is the problem that given two sets of strings $A$ and $B$, compute the maximum $\LCS(a,b)$ with $(a,b) \in A\times B$. We show how to reduce exact $\MAXLCSPAIR$ to approximate $\MAXLCSPAIR$ as an illustration of our proof techniques. 

\begin{theorem}[Informal]\label{theo:informal-exact-to-approx-MaxLCSPAIR}
	 There is a near-linear time\footnote{Throughout this paper, we use near-linear time to denote the running time of $N^{1 + o(1)}$.} reduction from \MAXLCSPAIR\ to $2^{(\log N)^{1-\Omega(1)}}$ factor approximate $\MAXLCSPAIR$, when $A,B$ are two sets of $N$ strings of length $D = 2^{(\log N)^{o(1)}}$.
\end{theorem}

We first introduce the concept of $\mathcal{A}$-\textit{Satisfying-Pair} problem. This problem asks whether there is a pair of $(a,b)$ from two given sets $A$ and $B$ such that $(a, b)$ is a yes-instance of $\mathcal{A}$. By binary search, $\MAXLCSPAIR$ can be easily formulated as an $\mathcal{A}$-\textit{Satisfying-Pair} problem: is there a pair $(a,b) \in A \times B$ such that $\LCS(a,b) \ge k$? The key property we are going to use is that the function $A^{\ge k}_{\LCS}(a,b) := [\LCS(a,b) \ge k]$ can be computed in very small space, i.e., it is in $\textsf{NL}$ (see Lemma~\ref{thm:lcs-nl}). Indeed, we will show in this paper that for all $\mathcal{A}$-\textit{Satisfying-Pair} such that $\mathcal{A}$ can be computed in small space, $\mathcal{A}$-\textit{Satisfying-Pair} can be reduced to approximate $\MAXLCSPAIR$.

\paragraph*{The Reduction in a Nutshell.} 
First, we consider an $\IP$ communication protocol for $\LCS$. In this setting, Alice and Bob hold strings $a$ and $b$, and they want to figure out whether $\LCS(a,b) \ge k$. To do so, they seek help from an untrusted prover Merlin by engaging in a conversation with him. The protocol should satisfy that when $\LCS(a,b) \ge k$, Merlin has a strategy to convince Alice and Bob w.h.p., and when $\LCS(a,b) < k$, no matter what Merlin does, Alice and Bob will reject w.h.p. The goal is to minimize the total communication bits (between Alice and Bob, or Alice/Bob and Merlin).

Next, by a result of Aaronson and Wigderson~\cite{AW09-algebrization}, it is shown that any function $f(a,b)$ which can be computed in $\NL$ admits an $\IP$ communication protocol with $\polylog(N)$ total communication bits. Finally, using an observation from~\cite{AbboudR18}, an efficient $\IP$ communication protocol can be embedded into approximate $\LCS$, which completes the reduction. In the following we explain each step in details.

\paragraph*{$\IP$ Communication Protocols for Low Space Computation.} The key technical ingredient of our results is the application of $\IP$ communication protocols for low space computation by Aaronson and Wigderson~\cite{AW09-algebrization}. It would be instructive to explain how it works.

\newcommand{\WT}{\widetilde}

Let us use the sum-check $\IP$ protocol for the Inner Product problem as an example. Arthur gets access to a function $f : \{0,1\}^{n} \to \{0,1\}$ and its multilinear extension $\WT{f} : \mathbb{F}_q^{n} \to \mathbb{F}_q$ over a finite field $\mathbb{F}_q$.
Let $f_{i}(x) = f(i \circ x)$ for $i \in \{0,1\}$ be the restrictions of $f$ after setting the first bit of the input, Arthur wants to compute the inner product $ \sum_{x \in \{0,1\}^{n-1}} f_0(x) \cdot f_1(x)$. To do so, he engages in a conservation with an untrusted Merlin who tries to convince him. During the conversation, Merlin is doing all the ``real work'', while Arthur only has to query $\WT{f}$ once at the last step, which is the crucial observation in~\cite{AW09-algebrization}.

Now, imagine a slightly different setting where $f_0$ and $f_1$ are held by Alice and Bob respectively, which means each of them holds a string of length $N = 2^{n-1}$. They still want to compute the inner product with Merlin, while using minimum communication between each other. 

In this setting, Alice (pretending she is Arthur) can still run the previous $\IP$ protocol with Merlin. When she has to query $\WT{f}(z)$ for a point $z$ at the last step, she only needs Bob to send her the contribution of his part to $\WT{f}(z)$, which only requires $O(\log q)$ bits. In terms of the input size of Alice and Bob, this $\IP$ communication protocol runs in $\poly(n) = \polylog(N)$ time. The same also extends to any $\poly(n) = \polylog(N)$ space computation on $f$, if we use the $\IP$ protocol for $\PSPACE$~\cite{LundFKN92,Shamir92}.

In Section~\ref{sec:tensor}, we provide a parameterized $\IP$ communication protocol for Branching Program\footnote{Informally speaking, a branching program with length $T$ and width $W$ formulates a \emph{non-uniform} low-space computation with running time $T$ and space $\log W$, see Definition~\ref{defi:BP} for a formal definition.} (Theorem~\ref{thm:bp-ip}). Informally, we have:

\begin{theorem}[$\IP$ Communication Protocol for BP (Informal)]
	Let $P$ be a branching program of length $T$ and width $W$ with $n$ input bits, equally distributed among Alice and Bob. 
	For every soundness parameter $\epsilon > 0$, there is an \IP-protocol for $P$, such that:
	\begin{itemize}
		\item Merlin and Alice exchange $\WT{O}\left(\log^2 W \log^2 T \log \epsilon^{-1}\right)$ bits, and toss the same amount of public coins;
		\item Bob sends $O( \log\log(WT) \cdot \log \epsilon^{-1})$ bits to Alice;
		\item Alice always accepts if $P$ accepts the input, and otherwise rejects with probability at least $1 - \epsilon$.
	\end{itemize}
\end{theorem}

Since $A^{\ge k}_{\LCS}(a,b)$ is in $\NL$, the above in particular implies that there is an $\IP$ communication protocol for $A^{\ge k}_{\LCS}$ with $\polylog(D)$ total communication bits ($D$ is length of strings).

\paragraph*{$\IP$ Communication Protocols and Tropical Tensors.} 
The next technical ingredient is the reduction from an $\IP$ communication protocol to a certain Tropical Tensors problem~\cite{AbboudR18}. We use a $3$-round $\IP$ communication protocol as an example to 
illustrate the reduction. Consider the following $3$-round $\IP$ communication protocol $\protocol$:

\begin{itemize}
	\item Alice and Bob hold strings $x$ and $y$, Merlin knows both $x$ and $y$.
	\item Merlin sends Alice and Bob a string $z_1 \in \mathcal{Z}_1$.
	\item Alice sends Merlin a uniform random string $z_2 \in \mathcal{Z}_2$.
	\item Merlin sends Alice and Bob another string $z_3 \in \mathcal{Z}_3$.
	\item Bob sends Alice a string $z_4 \in \mathcal{Z}_4$, and Alice decides whether to accept or reject.
\end{itemize}

The main idea of~\cite{AbboudR18} is that the above $\IP$ protocol can be reduced into a certain Tropical Similarity function. That is, for $x$ and $y$, we build two tensors $u=u(x)$ and $v=v(y)$ of size $|\mathcal{Z}_1| \times |\mathcal{Z}_2| \times |\mathcal{Z}_3| \times |\mathcal{Z}_4|$ as follows: we set $u_{z_1,z_2,z_3,z_4}$ to indicate whether Alice accepts, given the transcript $(z_1,z_2,z_3,z_4)$ and input $x$; we also set $v_{z_1,z_2,z_3,z_4}$ to indicate whether Bob sends the string $z_4$, given the previous transcript $(z_1,z_2,z_3)$ and input $y$. Then, by the definition of $\IP$ protocols, it is not hard to see the acceptance probability when Merlin uses optimal strategy is:
\[
\textsf{acc}(u,v) := \max_{z_1 \in \mathcal{Z}_1} \E_{z_2 \in \mathcal{Z}_2} \max_{(z_3,z_4) \in \mathcal{Z}_3 \times \mathcal{Z}_4} u_{z_1,z_2,z_3,z_4} \cdot v_{z_1,z_2,z_3,z_4}.
\]

In the above equality, the $\max$ operator corresponds to the actions of Merlin, who wishes to maximize the acceptance probability, while the $\E$ operator corresponds to actions of Alice, who sends a uniform random string. 
It can be easily generalized to $\IP$ protocols of any rounds, by replacing $\textsf{acc}$ with a series of $\max$ and $\E$ operators, which is called Tropical Similarity (denoted by $s(u,v)$) in~\cite{AbboudR18} (see also Definition~\ref{def:max-min-tt}). When the number of total communication bits is $d$, both $u$ and $v$ are of size $2^{d}$.

Applying the above to the $\polylog(D)$ bits $\IP$ protocol for $A^{\ge k}_{\LCS}$, it means for two strings $a,b$, we can compute two tensors $u,v$ of length $2^{\polylog(D)} = 2^{(\log N)^{o(1)}}$, such that when $\LCS(a,b) \ge k$, $s(u,v)$ is large, and otherwise $s(u,v)$ is small. 

\paragraph*{Simulating Tropical Tensors by Composing $\max$ and $\Sigma$ Gadgets.} While the above reduction is interesting in its own right, the Tropical Similarity function seems quite artificial.
Another key idea from~\cite{AbboudR18} is that $s(u,v)$ can be simulated by $\LCS$. The reduction works by noting that with $\LCS$, one can implement the $\max$ and $\Sigma$ (which is equivalent to $\E$) gadgets straightforwardly, and composing them recursively leads to gadgets for Tropical Similarity. That is, for tensors $u$ and $v$, one can construct strings $S(u)$ and $T(v)$ of similar sizes, such that $\LCS(S(u),T(v))$ is proportional to $s(u,v)$. 

Putting everything together, for two strings $a,b$, we can compute two other strings $S(a)$ and $T(b)$ of length $2^{\polylog(D)} = 2^{(\log N)^{o(1)}}$, such that $\LCS(a,b) \ge k$, $\LCS(S(a),T(b))$ is large, and otherwise $\LCS(S(a),T(b))$ is small. This completes our reduction.

\subsection*{Our Results In Detail}
\subsection{From Exact to Approximate in the Fine-Grained World}\label{sec:exact-to-approx}

More generally, we consider the following four (general flavor) problems.

\begin{itemize}
\item 
The $\MAXLCSPAIR$ problem.
\item
The \CLOSESTREGEXPSTRPAIR\ problem: Given a set $A$ of $N$ regular expressions of length $2^{(\log N)^{o(1)}}$ and a set $B$ of $N$ strings of length $2^{(\log N)^{o(1)}}$, find $(a,b) \in A \times B$ with maximum Hamming Similarity\footnote{Hamming Similarity between two strings are defined as the fraction of positions that they are equal, while the hamming similarity between a regular expression $a$ and a string $b$ is the maximum of the hamming similarity between $z$ and $b$ where $z$ is in the language of $a$.}.
\item
The \MAXLCSTPAIR\ problem:  Given two sets $A, B$ of $N$ bounded-degree trees with size $2^{(\log N)^{o(1)}}$, find a pair $(a,b) \in A \times B$ such that $a$ and $b$'s have the largest common subtree.
\item
The \MAXTT\ problem:  Given two sets $A, B$ of $N$ binary tensors with size $2^{(\log N)^{o(1)}}$, find the pair with maximum Tropical Similarity\footnote{see Definition~\ref{def:max-min-tt} for a formal definition.}.
\end{itemize}

Our main theorem shows equivalence between exact and approximation versions of the above problems. In fact, we show that all these problems, together with the following two closely related decision problems and a generic satisfying pair problem, are \emph{equivalent} under \emph{near-linear} time reductions. See Theorem~\ref{theo:informal-equiv-class} for a formal statement of the equivalence class.

\begin{itemize}
	\item
	The \REGEXPSTRPAIR problem: Given a set $A$ of $N$ regular expressions of length $2^{(\log N)^{o(1)}}$ and a set $B$ of $N$ strings of length $2^{(\log N)^{o(1)}}$, is there a pair $(a, b) \in A \times B$ such that $b$ matches $a$?
	\item
	The \STIPAIR problem:  Given two sets $A, B$ of $N$ bounded-degree trees with size $2^{(\log N)^{o(1)}}$, is there a pair $(a,b) \in A \times B$ such that $a$ is isomorphic to a subtree of $b$?
	\item 
	The \BPSATPAIR
	problem: Given a branching program\footnote{see Definition~\ref{defi:BP} for a formal definition} $P$ of size $2^{(\log N)^{o(1)}}$ and two sets $A, B$ of $N$ strings, is there a pair $(a,b) \in A \times B$ making $P$ accepts the input $(a,b)$?
\end{itemize}

We will refer to this set of problems as $\equivclass$.

\begin{remark}
	$\STIPAIR$ and $\MAXLCSTPAIR$ may seem artificial, but they are nice intermediate problems for showing hardness of the closely related problems Subtree Isomorphism and Longest Common Subtree, which are extensively studied natural problems (see Section~\ref{sec:equivclass-hard-intro} and Section~\ref{sec:related-works-specific-problems}).
\end{remark}

\subsubsection*{Equivalence in the Data Structure Setting}

These pair problems are interesting as they are natural \emph{off-line} versions of closely related data structure problem, which are highly relevant in the practice~\cite{Tho68,Mye92,BT09,myers1989approximate, wu1995subquadratic, knight1995approximate, myers1998reporting, navarro2004approximate, belazzougui2013approximate}. Therefore, a lower bound on the time complexity of these pair problems directly implies a lower bound for corresponding data structure problems  (see Theorem~\ref{theo:hardness-for-datastructure}). In the next section, we will show under some conjecture which is much more plausible than $\SETH$, these problems requires essentially quadratic time.

Our equivalence continues to hold in the data structure version, in particular, for the following data structure problems, any algorithm for one of them implies an algorithm for all of them with essentially the same preprocessing time/space and query time (up to a factor of $N^{o(1)}$). See Section~\ref{sec:data-structure} for the details.

\begin{itemize}
	\item $\NNS_\LCS$\footnote{already discussed in Section~\ref{sec:first-example}}: Preprocess a database $\mathcal{D}$ of $N$ strings of length $D = 2^{(\log N)^{o(1)}}$, and then for each query string $x$, find $y \in \mathcal{D}$ maximizing $\LCS(x,y)$.
	
	\item Approx. $\NNS_\LCS$: Find $y \in \mathcal{D}$ s.t. $\LCS(x,y)$ is a $2^{(\log D)^{1- \Omega(1)}}$ approximation to the maximum value.
	
	\item Regular Expression Query: Preprocess a database $\mathcal{D}$ of $N$ strings of length $D = 2^{(\log N)^{o(1)}}$, and then for each query regular expression $y$, find an $x \in \mathcal{D}$ matching $y$.
	
	\item Approximate Regular Expression Query: For a query expression $y$, distinguish between\footnote{behavior can be arbitrary when neither of the two cases hold}: (1) there is an $x \in \mathcal{D}$ matching $y$; and (2) for all $x \in \mathcal{D}$, the Hamming distance between $x$ and all $z \in L(y)$ is at least $(1-o(1)) \cdot D$, where $L(y)$ is the set of all strings matched by $y$.
\end{itemize}

That is, a non-trivial data structure for finding \emph{approximate} nearest point with $\LCS$ metric would imply a non-trivial data structure for answering regular expression query! The latter one is supported by most modern database systems such as MySQL, Oracle Database, Microsoft SQL etc., but all of them implement it by simply using full table scan for the most general case.

\paragraph*{$\LCS$ is the Hardest Distance Function for Approximate $\NNS$.} In fact, our results also suggest in a formal sense that $\LCS$ is the \emph{hardest} distance function for approximate Nearest Neighbor Search. We show that for all distance function $\dist$ which is computable in poly-logarithmic \emph{space}\footnote{Which is true for almost all nature distance functions. For example, edit distance and LCS can be computed in $\NL$, thus in $(\log^2 N)$ space by Savitch's Theorem~\cite{Savitch70}.}, exact $\NNS$ for $\dist$ can be reduced to approximate $\NNS_\LCS$.

\begin{theorem}[Informal] For a distance function $\dist$ that is computable in poly-logarithmic space, exact $\NNS$ for $\dist$ can be reduced to $2^{(\log D)^{1- \Omega(1)}}$-approximate $\NNS_\LCS$ in near-linear time.
\end{theorem}

\subsection{Weaker Complexity Assumptions for Approximation Hardness} 

An important goal in the study of fine-grained complexity is to find more plausible conjectures, under which to base hardness. For example, although $\SETH$ is based on the historically unsuccessful attempts on finding better algorithms for $k$-$\SAT$, there is no consensus on its validity (see, e.g.~\cite{Williams16_MA-SETH,williamsLikelihoods}).

This concern has been addressed in various ways. For example, in~\cite{Abboud2015tria}, the authors prove hardness for several problems, basing on \emph{at least one of} the $\SETH$, the APSP conjecture, or the $3$-SUM Conjecture being true.
In~\cite{AHVW16}, Abboud et al. introduce a hierarchy of $\circuitclass$-$\SETH$ assumptions: the $\circuitclass$-$\SETH$ asserts that there is no $2^{(1-\eps) n}$ time satisfiability algorithm for circuits from $\circuitclass$.\footnote{In this way, the original $\SETH$ assert that there is no $2^{(1-\eps)n}$ time algorithm for satisfiability of $\textsf{CNF}$ with arbitrary constant bottom fan-in.} They show that the quadratic time hardness of Edit-Distance, $\LCS$ and other related sequence alignment problems can be based on the much weaker and much more plausible assumption $\NC\textsf{-}\SETH$. However, this has not been shown for approximation version of fine-grained problems.

In this work, we show that all problems in $\equivclass$ require essentially quadratic time under $\NC\textsf{-}\SETH$. Indeed, our hardness results are based on a weaker assumption which we call \textit{$2^{n^{o(1)}}$-size $\BP\textsf{-}\SETH$}:\footnote{Indeed, results in \cite{AHVW16} are also based on a conjecture about branching programs, which can be seen as $O(1)$-width and $2^{o(n)}$-length $\BP\textsf{-}\SETH$ or $2^{o(\sqrt{n})}$-size $\BP\textsf{-}\SETH$ using the terminology in Hypothesis~\ref{hypot:bp-seth}.}
\begin{hypot}[$2^{n^{o(1)}}$-size $\BP\textsf{-}\SETH$] \label{hypot:bp-seth}
The satisfiability of a given $2^{n^{o(1)}}$-size non-deterministic branching program cannot be solved in $O(2^{(1 - \delta)n})$ time for any $\delta > 0$.
\end{hypot}

\begin{theorem}\label{theo:better-foundation-equiv}
All problems in $\equivclass$ require $N^{2 - o(1)}$ time if we assume $2^{n^{o(1)}}$-size $\BP\textsf{-}\SETH$.
\end{theorem}

Note that $2^{n^{o(1)}}$-size $\BP\textsf{-}\SETH$ is even weaker than $n^{o(1)}$-depth circuit $\SETH$: satisfiability for $n^{o(1)}$-depth bounded fan-in circuits cannot be solved in $O(2^{(1 - \delta)n})$ time for any $\delta > 0$. This is because by Barrington's Theorem \cite{Barrington89}, $n^{o(1)}$-depth bounded fan-in circuits can be simulated by branching programs of size $2^{n^{o(1)}}$. 

It is worthwhile to compare with~\cite{ARW17-proceedings}. It is shown in~\cite{ARW17-proceedings} that assuming SETH, a $2^{(\log N)^{1 - o(1)}}$-approximation to \MAXLCSPAIR\ (\CLOSESTREGEXPSTRPAIR) requires $N^{2-o(1)}$ time for $D = N^{o(1)}$. (The $2^{(\log N)^{1 - o(1)}}$ factor is later improved to $N^{o(1)}$ in~\cite{Rub18,Che18}.)

Although our results here are quantitatively worse, it is ``qualitatively'' better in many ways: (1) the results in~\cite{ARW17-proceedings} is based on $\SETH$, while our hardness results are based on the assumptions in Theorem~\ref{theo:better-foundation-equiv}, which are much more plausible than $\SETH$; (2) we in fact have established an equivalence between $\MAXLCSPAIR$ and its approximation version, which seems not possible with the techniques in~\cite{ARW17-proceedings}; (3) our framework allows us to show that even a tiny improvement on the running time would have important algorithmic and circuit lower bound consequences (see Theorem~\ref{theo:shave-logs-equiv}), which again seems not possible with the techniques in~\cite{ARW17-proceedings}.

\subsection{$\equivclass$ Hard Problems}
\label{sec:equivclass-hard-intro}

We also identify a set of other problems which are at least as hard as any problem in \equivclass, but not necessarily in it. We say these problems are \equivclasshard.

\begin{theorem}[$\equivclasshard$ Problems]\label{theo:informal-hard}
There are near-linear time reductions from all the problems in $\equivclass$ to any of the following problems:
\begin{enumerate}
\item (Subtree Isomorphism) Given two trees $a, b$ of size at most $N$, determine whether $a$ is isomorphic to a subtree of $b$ (even if restricted to the case of binary rooted trees);

\item (Largest Common Subtree) Given two trees $a, b$ of size at most $N$, compute the exact value or a $2^{(\log N)^{o(1)}}$-approximation of the size of the largest common subtree of $a$ and $b$ (even if restricted to the case of binary rooted trees);

\item (Regular Expression Membership Testing) Given a regular expression $a$ of length $M$ and a string $b$ of length $N$, determine whether $b$ is in the language of $a$;
\end{enumerate}

\end{theorem}

\begin{cor}\label{cor:better-foundation-hard}
	All the above $\equivclasshard$ problems require $N^{2 - o(1)}$ time (or $(NM)^{1 - o(1)}$ time for Regular Expression Membership Testing) under the same assumption as in Theorem~\ref{theo:better-foundation-equiv}.
\end{cor}

We remark that both Subtree Isomorphism and Largest Common Subtree are studied in~\cite{ABHVZ16}. In particular, they showed that Subtree Isomorphism and Largest Common Subtree require quadratic-time under $\SETH$, even for binary rooted trees. 

Our results improve theirs in many ways: (1) for Subtree Isomorphism, we establish the same quadratic time hardness, with a much safer conjecture; (2) for Largest Common Subtree, we not only put its hardness under a better conjecture, but also show that even a $2^{(\log N)^{o(1)}}$-approximation would be hard; (3) for both of these problems, we demonstrate that even a tiny improvement on the running time would have interesting algorithmic and circuit lower bound consequences (see Theorem~\ref{theo:shave-logs-hard}).

\cite{BGL17} (which builds on~\cite{BackursI16}) classified the running time of constant-depth regular expression membership testing. In particular, they showed a large class of regular expression testing requires quadratic-time, under $\SETH$. Our results are incomparable with theirs, as our hard instances may have unbounded depth regular expressions. On the bright side, our hardness results rely on a much safer conjecture, and we show interesting consequences even for a tiny improvement of the running time.

\subsection{The Consequence of ``shaving-logs'' for Approximation Algorithms} 

There has been a large number of works focusing on ``shaving logs'' of the running time of fundamental problems~\cite{arlazaro1970economical,BansalW12Regularity,Chan15SpeedUp,Yu15Improved} (see also a talk by Chan~\cite{Chan13Art}, named ``The Art of Shaving Logs''). In a recent exciting algorithmic work by Williams~\cite{williams2014faster}, the author shaves ``all the logs'' on the running time of APSP, by getting an $n^{3} / 2^{\Theta(\sqrt{\log n})}$ time algorithm.

However, the best exact algorithms for $\LCS$ and Edit distance~\cite{MasekP80EditDist,Grabowski16EditDist} remain $O(n^{2} / \log^2 n)$, which calls for an explanation. An interesting feature of \cite{AHVW16} is that their results show that even shaving logs on LCS or Edit Distance would be very hard. In particular, they prove that an $n^{2} / \log^{\omega(1)} n$ time algorithm for either of them would imply a $2^{n} / n^{\omega(1)}$ time algorithm for polynomial-size formula satisfiability, which is much better than the current state of arts~\cite{Santhanam10,Tal15}. Such an algorithm would also imply that $\NEXP$ is not contained in non-uniform $\NC^1$, thereby solving a notorious longstanding open question in complexity theory.

The ``shaving logs barrier'' only has been studied for a few problems. It was not clear whether we can get the same barriers for some \emph{approximation problems}.

In this work we show that slightly improved algorithms (such as shaving all the logs) for any \equivclass or \equivclasshard problems, would imply \emph{circuit lower bounds} which are notoriously hard to prove. This extends all the results of~\cite{AHVW16} to approximation problems.

\begin{theorem}\label{theo:shave-logs-equiv}
If there is an $O\left(N^2\poly(D) / 2^{(\log \log N)^3}\right) ~\text{or}~ O\left(N^2 / (\log N)^{\omega(1)}\right)$ time deterministic algorithm for any decision, exact value or $\polylog(D)$-approximation problems in \equivclass, where $D$ is the maximum length (or size) of elements in sets, then the following holds:
\begin{itemize}
	\item $\NTIME[2^{O(n)}]$ is not contained in non-uniform $\NC^1$ and
	\item $\textsf{Formula-SAT}$ with $n^{\omega(1)}$ size can be solved in $2^n / n^{\omega(1)}$ time.
\end{itemize}
\end{theorem}

\begin{theorem}\label{theo:shave-logs-hard}
If there is a deterministic algorithm for any decision, exact value or $\polylog(N)$-approximation problems among \equivclasshard problems listed in Theorem~\ref{theo:informal-hard} running in $O\left(N^2 / 2^{\omega(\log \log N)^3}\right)$ time (or $O\left(NM / 2^{\omega(\log \log(NM))^3}\right)$ time for Regular Expression Membership Testing), then the same consequences in Theorem~\ref{theo:shave-logs-equiv} follows.
\end{theorem}

\subsection{Circuit Lower Bound Consequence for Improving Approximation Algorithms for P Time Problems} 

Finally, we significantly improve the results from~\cite{AbboudR18}, by showing much stronger circuit lower bound consequences for deterministic approximation algorithms to $\LCS$.

\begin{theorem} \label{thm:approx-lcs-cons}
	The following holds for deterministic approximation to $\LCS$:
		
		\begin{enumerate}
			\item A $2^{(\log N)^{1 - \Omega(1)}}$-approximation algorithm in $N^{2 - \delta}$ time for some constant $\delta > 0$ implies that $\textsf{E}^{\NP}$ has no $n^{o(1)}$-depth bounded fan-in circuits;
			
			\item A $2^{o(\log N / (\log \log N)^2)}$-approximation algorithm in $N^{2 - \delta}$ time for some constant $\delta > 0$ implies that $\NTIME[2^{O(n)}]$ is not contained in non-uniform $\NC^1$;
			
			\item An $O(\polylog(N))$-approximation algorithm in $N^2 / 2^{\omega(\log \log N)^3}$ time implies that $\NTIME[2^{O(n)}]$ is not contained in non-uniform $\NC^1$.
		\end{enumerate}
	
\end{theorem}

In comparison with~\cite{AbboudR18}, they show that an $O(N^{2 - \eps})$ time algorithm for constant factor deterministic approximation algorithm to $\LCS$ would imply that $\textsf{E}^\NP$ does not have non-uniform linear-size $\NC^1$ circuits or VSP circuits. Our results here generalize theirs in all aspects: (1) we show that a much stronger lower bound consequence would follow from even a sub-quadratic time $2^{(\log N)^{1 - \Omega(1)}}$-approximation algorithm; (2) we also show that a modestly stronger lower bound would follow even from a quasi-polylogarithmic improvement over the quadratic time, for approximate $\LCS$.

\medskip

More generally, following a similar argument to~\cite{AHVW16}, we can show that truly-subquadratic time algorithms for these $\equivclass$ or $\equivclasshard$ problems would imply strong circuit lower bounds against $\ETIME^{\NP}$.

\begin{cor} \label{cor:cons-bpseth}
If any of the problems listed in Theorem~\ref{theo:informal-equiv-class} and Theorem~\ref{theo:informal-hard} admits an $N^{2 - \eps}$ time algorithm (or $(NM)^{1-\epsilon}$ time algorithm for Regular Expression Membership Testing) for some $\eps > 0$, then $\textsf{E}^{\textsf{NP}}$ does not have:
	\begin{enumerate}
		\item non-uniform $2^{n^{o(1)}}$-size Boolean formulas,
		\item non-uniform $n^{o(1)}$-depth circuits of bounded fan-in, and
		\item non-uniform $2^{n^{o(1)}}$-size nondeterministic branching programs.
	\end{enumerate}
\end{cor}

\subsection{Discussion and Open Problems}

Here we discuss some open problems arising from our work.

\subsection*{Find More Members for $\equivclass$}

One immediate question is to find more natural quadratic-time problems belonging to \equivclass:

\begin{open}
	Find more natural problems which belong to $\equivclass$.
\end{open}

It could be helpful to revisit all $\SETH$-hard problems to see whether they can simulate $\BPSATPAIR$. In particular, one may ask whether the Orthogonal Vectors problem ($\OV$), the most studied problem in fine-grained complexity, belongs to this equivalence class:

\begin{open}
	Does $\OV$ belong to $\equivclass$?
\end{open}

If it does, then it would open up the possibility that perhaps all $\SETH$-hard quadratic-time problems are equivalent. However, some evidence suggests that the answer may be negative, as $\OV$ seems to be much easier than problems in $\equivclass$:

\begin{itemize}
	\item \textbf{The Inner Function in $\OV$ is Much Weaker.} When viewing as a $\textsf{Satisfying-Pair}$ problem, the inner function in $\OV$ is just a simple \emph{Set-Disjointness}, which seems incapable of simulating generic low-space computation.
	 
	\item \textbf{There are Non-trivial Algorithms for $\OV$.} We know that for $\OV$ with $N$ vectors of length $D = c\log n$, there are algorithms with running time $N^{2 - 1/O(\log c)}$~\cite{abboud2015more,ChanW16}. This type of non-trivial speed up seems quite unlikely (or at least much harder to obtain) for problems in \equivclass\ (see Theorem~\ref{theo:shave-logs-equiv}).
\end{itemize}  

It would be interesting to show that $\OV$ and $\BPSATPAIR$ are not equivalent under certain plausible conjectures, perhaps ideas from~\cite{carmosino2016nondeterministic} could help.

\subsection*{Quasi-Polynomial Blow Up of the Dimension}

In the reductions between our $\equivclass$ problems, we get a quasi-polynomial blowup on the dimensions: that is, a problem with element size (vector dimension, string length or tree size) $D$ is transformed into another problem with element size $2^{\polylog(D)}$. This is the main reason that we have to restrict the element size to be small, i.e., $D = 2^{(\log N)^{o(1)}}$.

This technical subtlety arises from the polynomial blow-up in the $\IP = \PSPACE$ proof: given a language in $\NSPACE[S]$, it is first transformed into a $\textsf{TQBF}$ instance of size $O(S^2)$, which is proved by an $\widetilde{O}(S^4)$ time $\IP$ protocols, using arithmetization. 

Applying that into our setting, an $\NL$ computation on two strings of length $D$ (like $\LCS$), is transformed into an $\widetilde{O}(\log^4 D)$ time $\IP$ communication protocol, which is then embedded into an approximate problem with at least $2^{\widetilde{O}(\log^4 D)}$ dimensions (say approximate $\LCS$).

However, if we have an $O(\log D)$ time $\IP$ communication protocol for $\NL$. The new dimension would be $2^{O(\log D)} = D^{O(1)}$, only a polynomial blow up. Which motivates the following interesting question:

\begin{open}\label{open:fast-IP-protocol}
	Is there an $O(\log D)$ time $\IP$ communication protocol for every problems in $\NL$? 
\end{open}

A positive resolution of the above question would also tighten several parameters in many of our results. For example, in Theorem~\ref{theo:better-foundation-equiv}, $n^{o(1)}$-depth circuit $\SETH$ could be replaced by $o(n)$-depth circuit $\SETH$, and Theorem~\ref{theo:shave-logs-equiv}, Theorem~\ref{theo:shave-logs-hard}, Theorem~\ref{theo:informal-hard} and Corollary~\ref{cor:cons-bpseth} would also have improved parameters. 

It is worth noting that $\IP$ communication lower bounds are extremely hard to prove\textemdash proving a non-trivial lower bound for $\AM$ communication protocols is already a long-standing open question~\cite{Lokam01,GoosPW16,goos2016landscape}. Hence, resolving Open Question~\ref{open:fast-IP-protocol} negatively could be hard.

\subsection{Related Works}

\newcommand{\posR}{\mathbb{R}^{+}}

\subsubsection{Hardness of Approximation in $\PTIME$}

In \cite{ARW17-proceedings} the Distributed PCP framework is introduced, which is utilized and generalized by several follow up works. Using Algebraic Geometry codes, in a recent work, \cite{Rub18} obtains a better $\MA$ protocol for Set-Disjointness, improving the efficiency of the distributed PCP construction, and shows quadratic-time hardness for $(1+o(1))$-approximation to Bichromatic Closest Pair and several other related problems. 

Building on the technique of~\cite{Rub18},~\cite{Che18} obtains a characterization of what multiplicative/additive approximation ratios to Maximum Inner Product can be computed in sub-quadratic time. He also shows a connection between $\BQP$ communication protocol for Set-Disjointness and conditional lower bound for Maximum Inner Product with $\{-1,1\}$-valued vectors.

\cite{CLM17} generalizes the distributed PCP framework by considering multi-party communication protocols, and derive inapproximability results for $k$-Dominating Set under various assumptions. In particular, using the techniques of~\cite{Rub18}, they prove that under $\SETH$, $k$-Dominating Set has no $(\log n)^{1/\poly(k,e(\eps))}$ approximation in $n^{k-\eps}$ time\footnote{where $e: \posR \to \mathbb{N}$ is some function}. 

\cite{abboud2017towards} takes a different approach, which makes use of the connection between weak derandomization and circuit lower bound~\cite{Wil13,BV14}. They show that, under a certain plausible complexity assumption, $\textsf{LCS}$ does not have a \emph{deterministic} $(1+o(1))$-approximation in $n^{2 -\eps}$ time. They also establish a connection with circuit lower bounds and prove that such a \emph{deterministic} algorithm implies $\mathsf{E}^{\NP}$ does not have non-uniform linear-size Valiant Series Parallel circuits. In~\cite{AbboudR18}, the circuit lower bound connection is improved to that any constant factor deterministic approximation for $\textsf{LCS}$ in $n^{2 - \eps}$ time implies that $\mathsf{E}^\NP$ does not have non-uniform linear-size $\textsf{NC}^1$ circuits. See~\cite{ARW17-proceedings} for more related results in hardness of approximation in~$\PTIME$.

\subsubsection{Equivalence Classes in $\PTIME$}

A partial list of the APSP equivalence class~\cite{WW10-subcubic,backurs2016tight,abboud2015subcubic,lincoln2018tight} includes: Negative Triangle, Triangle listing, Shortest Cycle, 2nd Shortest Path, Max Subarray, Graph Median, Graph Radius, Wiener Index (see~\cite{VassilevskaW2018survey} for more details). 

In~\cite{GIKW17}, it is shown that ``medium-dimensional'' $\OV$ (i.e., $\OV$ with $n^{o(1)}$ dimensions) is equivalent to High-dimension Sparse $\OV$, High-dimension $2$-Set Cover, and High-dimension Sperner Family. It is also proved that for every $(k+1)$-quantifier first-order property, its model-checking problem can be reduced to Sparse $k$-$\OV$.

In~\cite{CyganMWW17}, the authors introduce an equivalence class for $(\min,+)$-convolution, including some variants of classical knapsack problem and problems related to subadditive sequences.

\subsubsection{Hardness for Shaving Logs}
In~\cite{AHVW16}, it is shown that an $n^{2} / \log^{\omega(1)}$ time algorithm for $\LCS$ would imply that the $\textsf{Formula-SAT}$ have a $2^n / n^{\omega(1)}$ time algorithm. The construction is later tightened in~\cite{abboud2018tighter}, which shows that an $n^2 / \log^{7+\eps} n$ time algorithm for any of $\LCS$, regular expression pattern matching or Fr\'echet distance is already enough to imply new algorithm for $\textsf{Formula-SAT}$. 

\subsubsection{Related Works for Specific Problems}\label{sec:related-works-specific-problems}

\paragraph*{Longest Common Subsequence.} \LCS is a very basic problem in computer science and has been studied for decades \cite{chvatal1972selected,bergroth1998new,de1999extensive,crochemore2001fast}. In recent years, a series of works \cite{BI15,ABV15a,BK15,AHVW16,AbboudR18} have shown different evidences that \LCS may not have truly sub-quadratic time algorithm, even for approximation.

\paragraph*{Subtree Isomophism and Largest Common Subtree.} Subtree Isomorphism has been studied in \cite{matula1968algorithm, matula1978subtree, lingas1983application, reyner1977analysis,chung1987n2, lingas1989subtree,gibbons1990subtree, shamir1999faster}.
Largest Common Subtree is an \NP-hard problem when the number of trees is not fixed, and has been studied in \cite{khanna1995approximation,akutsu2000approximation,akutsu2015complexity}. For (rooted or unrooted) bounded-degree trees, both the two problems can be solved in $O(N^2)$ time, and the fastest algorithm for Subtree Isomorphism runs in $O(N^2 / \log N)$ time~\cite{lingas1983application, shamir1999faster}. In \cite{ABHVZ16}, these two problems are shown to be \SETH-hard.

\paragraph*{Regular Expression.} The found of $O(NM)$ algorithm for regular expression matching and membership testing in \cite{Tho68} is a big sucess in 70s, but after that no algorithm has been found to improve it to truly sub-quadratic time~\cite{Mye92,BT09}. Related works about the hardness of exact regular expression matching or membership testing include~\cite{BackursI16,BGL17}. There are many works on approximate regular expression matching in different formulations \cite{myers1989approximate, wu1995subquadratic, knight1995approximate, myers1998reporting, navarro2004approximate, belazzougui2013approximate}, and its hardness has been analyzed in~\cite{ARW17-proceedings}.

\subsection*{Organization of this Paper}
In Section~\ref{sec:prelim}, we introduce the needed preliminaries, as well as the formal definitions of the problems we studied. In Section~\ref{sec:outline}, we outline the structure of all reductions for our $\equivclass$ and $\equivclasshard$ problems (Theorem~\ref{theo:informal-equiv-class} and Theorem~\ref{theo:informal-hard}). For ease of presentation, these reductions are presented from Section~\ref{sec:to-BPSATPAIR} to Section~\ref{sec:subtree-Iso}. In Section~\ref{sec:algo-to-circuit-lowb}, we show that tiny improvements on the running time of $\equivclass$ or $\equivclasshard$ problems would have important algorithmic and circuit lower bound consequences. In Section~\ref{sec:derand-to-circuit-lowb}, we establish the consequences of faster deterministic approximation algorithms to $\LCS$.

\section{Preliminaries}\label{sec:prelim}

In this section, we define the $\SAT$ problem for branching program (\BPSAT), and then we introduce the formal definitions for each problem in \equivclass and \equivclasshard.

\begin{definition}[Branching Program]\label{defi:BP}
A (nondeterministic) \textit{Branching Program} (BP) on $n$ boolean inputs $x_1, \dots, x_n$ is defined as a layered directed graph with $T$ layers $L_1, \dots, L_T$. Each layer $L_i$ contains $W$ nodes. Except the last layer, every node in $L_i$ ($1 \le i < T$) is asscociated with the same variable $x_{f(i)}$ for some $f(i) \in [n]$. $T$ and $W$ are called the \textit{length} and \textit{width} of the BP. 

For every two adjacent layers $L_i$ and $L_{i+1}$, there are edges between nodes in $L_i$ to nodes in $L_{i+1}$, and each edge is marked with either $0$ or $1$.
The \textit{size} of a BP is defined as the total number of edges $O(W^2T)$.

For $1 \le i \le T, 1 \le j \le W$, the $j$-th node in the $i$-th layer $L_i$ is labeled as $(i, j)$. $u_{\mathrm{start}} = (1, 1)$ is the starting node, and $u_{\mathrm{acc}} = (T, 1)$ is the accepting node. A BP accepts an input $x$ iff there is a path from the starting node to the accepting node consisting of only the edges marked with the value of the variable associated with its starting endpoint.
\end{definition}

\begin{definition}[\BPSAT] Given a branching program $P$ on $n$ boolean inputs, the \textsf{BP-SAT} problem asks whether there is an input making $P$ accept.
\end{definition}

Like the relationship between $k$-SAT and Orthogonal Vectors (OV), we can define \BPSATPAIR problem as the counterpart of \BPSAT in the \PTIME world. \BPSATPAIR can be trivially solved in $O(N^2 \cdot \poly(W, T))$ time, and a faster algorithm for \BPSATPAIR running in $O(N^{2 - \epsilon})$ time implies a faster algorithm for \BPSAT running in $O(2^{(1 - \epsilon/2)n})$ time.
\begin{definition}[\BPSATPAIR] Given a branching program $P$ on $n$ boolean inputs (assume $n$ is even) and two sets of $N$ strings $A, B \subseteq \{0, 1\}^{n/2}$, the \BPSATPAIR problem asks whether there is a pair $a \in A, b \in B$ such that $P$ accepts the concatenation of $a$ and $b$.
\end{definition}

Throughout this paper, unless otherwise stated, we use \BPSATPAIR to denote the \BPSATPAIR problem on branching program of size $2^{(\log N)^{o(1)}}$ for convenience.

\subsection{Satisfying Pair and Best Pair Problems}\label{sec:sat-pair}

\paragraph{Satisfying Pair Problems.} Note that problems like \textit{Orthogonal Vectors} are in the form of deciding whether there is a ``satisfying pair''. In general, we can define the $\mathcal{A}$-\textit{Satisfying-Pair} problem, where $\mathcal{A}$ is an arbitrary decision problem on two input strings $x, y$:
\begin{definition}[\ASATPAIR{$\mathcal{A}$}]
Given two sets $A, B$ of $N$ strings, the $\mathcal{A}$-\textit{Satisfying-Pair} problem asks whether there is a pair of $a \in A, b \in B$ such that $(a, b)$ is an Yes-instance of $\mathcal{A}$.
\end{definition}

In this work, we study a series of \ASATPAIR{$\mathcal{A}$} problems, including
\OAPT, \REGEXPSTRPAIR and \STIPAIR, which will be formally defined in later subsections.

\paragraph{Best Pair Problems.} For an optimization problem $\mathcal{A}$ on two input strings $x, y$, we can define the \MAXAPAIR{$\mathcal{A}$} and the \MINAPAIR{$\mathcal{A}$} problems:
\begin{definition}[\MAXAPAIR{$\mathcal{A}$} / \MINAPAIR{$\mathcal{A}$}]
Given two sets $A, B$ of $N$ strings, the \MAXAPAIR{$\mathcal{A}$} (or \MINAPAIR{$\mathcal{A}$}) problem asks to compute the maximum value (or minimum value) of the result of problem $\mathcal{A}$ on input $(a, b)$.
\end{definition}

In this work, we study a series of \MAXAPAIR{$\mathcal{A}$} / \MINAPAIR{$\mathcal{A}$} problems, including \MAXTT, \MINTT, \MAXLCSPAIR, \MINLCSPAIR, \CLOSESTREGEXPSTRPAIR, \MAXLCSTPAIR and \MINLCSTPAIR, which will be formally defined in later subsections.

Note that both satisfying pair problems and best pair problems contain two sets $A, B$ in the input. Without additional explanation, we use $N$ to denote the set size, and $D$ to denote the maximum element size in sets.

\subsection{Two Tensor Problems} \label{sec:tensor-def}

We introduce two kinds of tensor problems: the \textit{Orthogonal Alternating Product Tensors} problem (\OAPT) and the \textit{Max / Min Tropical Similarity} problem (\MAXTT~/ \MINTT). The former one is an \ASATPAIR{$\mathcal{A}$} problem, which helps us to prove hardness for decision problems; the latter one is a \MAXAPAIR{$\mathcal{A}$} / \MINAPAIR{$\mathcal{A}$} problem, which helps us proving hardness of approximation for optimization problems. 

First we define \OAPT. \OAPT implicitly appears in the reduction from \BPSAT to \LCS in \cite{AHVW16} as an intermediate problem. In our reductions, \OAPT  appears naturally, and we show that \BPSATPAIR and \OAPT of certain size are equivalent under near-linear time reduction in Section \ref{sec:tensor}.

\begin{definition}[\OAPT] \label{def:oapt}
Let $t$ be an even number and $d_1 = d_2 = \cdots = d_t = 2$. The \textit{Alternating Product} $\aprod(u, v)$ of two tensors $u, v \in \{0, 1\}^{d_1 \times \cdots \times d_t}$ is defined as an alternating sequence of logical operators $\land$ and $\lor$ applied to the coordinatewise product of $u$ and $v$:
\begin{equation}
\aprod(u, v) = \bigwedge_{i_1 \in [d_1]} \left[\bigvee_{i_2 \in [d_2]}\left( \bigwedge_{i_3 \in [d_3]} \left[ \cdots \bigvee_{i_t \in [d_t]} \left(u_i \land v_i\right) \cdots \right] \right)\right].
\end{equation}
Given two sets of $N$ tensors $A, B \subseteq \{0, 1\}^{d_1 \times \cdots \times d_t}$, the Orthogonal Alternating Product Tensors (\OAPT) problem asks whether there is a pair $a \in A, b \in B$ such that the Alternating Product $\aprod(a, b) = 0$.
\end{definition}

\OAPTrestri is a restricted version of \OAPT. We mainly use this restricted version in our analysis.
\begin{definition}[\OAPTrestri] \label{def:restriced-oapt}
We say that a tensor $x \in \{0, 1\}^{d_1 \times \cdots \times d_t}$ is $\land$-invariant if the value of $x_{i_1 \cdots i_t}$ does not depend on $i_1, i_3, i_5, \dots, i_{t - 1}$. The \OAPTrestri problem is a restricted version of \OAPT, where one set of $A, B$ contains only $\land$-invariant tensors.
\end{definition}

For proving our hardness of approximation results, we further define the \MAXTT and \MINTT problems. 
The \MAXTT problem was firstly proposed by Abboud and Rubinstein in \cite{AbboudR18} under the name \textit{Tropical Tensors} in their study of \LCS.
\begin{definition}[\MAXTT~/ \MINTT] \label{def:max-min-tt}
Let $t$ be an even number and $d_1 = d_2 = \cdots = d_t = 2$. The \textit{Tropical Similarity} score $s(u, v)$ of two tensors $u, v \in \{0, 1\}^{d_1 \times \cdots \times d_t}$ is defined as an alternating sequence of operators $\E$ and $\max$ applied to the coordinatewise product of $u$ and $v$:
\begin{equation}
s(u, v) = \E_{i_1 \in [d_1]} \left[\max_{i_2 \in [d_2]}\left\{ \E_{i_3 \in [d_3]} \left[ \cdots \max_{i_t \in [d_t]} \left\{u_i \cdot v_i\right\} \cdots \right] \right\}\right].
\end{equation}
Given two sets of $N$ tensors $A, B \subseteq \{0, 1\}^{d_1 \times \cdots \times d_t}$, the \MAXTT problem asks to compute the maximum Tropical Similarity $s(a, b)$ among all pairs of $(a, b) \in A \times B$,
while the \MINTT problem asks to compute the minimum Tropical Similarity $s(a, b)$ among all pairs of $(a, b) \in A \times B$.
\end{definition}

\OAPTrestri is a restricted version of \OAPT. We mainly use this restricted version in our analysis.

\begin{definition}[\MAXTTrestri~/ \MINTTrestri] \label{def:restriced-max-min-tt}
We say that a tensor $x \in \{0, 1\}^{d_1 \times \cdots \times d_t}$ is $\max$-invariant if the value of $x_{i_1 \cdots i_t}$ does not depend on $i_2, i_4, i_6, \dots, i_{t}$. The \MAXTTrestri~/ \MINTTrestri problem is a restricted version of \MAXTT / \MINTT, where one set of $A, B$ contains only $\land$-invariant tensors.
\end{definition}

Like in \cite{AbboudR18}, we also define the following approximation variants of the \MAXTT and \MINTT.
\begin{definition}[\GAPMAXTT{$\epsilon$}] 
Let $t$ be an even number and $d_1 = d_2 = \cdots = d_t = 2$. Given two sets of $N$ tensors $A, B \in \{0, 1\}^{d_1 \times \cdots \times d_t}$ of size $D = 2^t$, distinguish between the following:
\begin{itemize}
\item \textbf{Completeness:} There is a pair of $(a, b) \in A \times B$ with a perfect Tropical Similarity $s(a, b) = 1$;
\item \textbf{Soundness:} Every pair has low Tropical Similarity score, $s(a, b) < \epsilon$.
\end{itemize}
Here $\epsilon$ is a threshold value that may depend on $N$ and $D$. \GAPMAXTTrestri{$\epsilon$} is defined similarly.
\end{definition}

\begin{definition}[\GAPMINTT{$\epsilon$}] 
Let $t$ be an even number and $d_1 = d_2 = \cdots = d_t = 2$. Given two sets of $N$ tensors $A, B \in \{0, 1\}^{d_1 \times \cdots \times d_t}$ of size $D = 2^t$, distinguish between the following:
\begin{itemize}
\item \textbf{Completeness:} There is a pair of $(a, b) \in A \times B$ with a low Tropical Similarity $s(a, b) < \epsilon$;
\item \textbf{Soundness:} Every pair has perfect Tropical Similarity score, $s(a, b) = 1$.
\end{itemize}
Here $\epsilon$ is a threshold value that may depend on $N$ and $D$. \GAPMINTTrestri{$\epsilon$} is defined similarly.
\end{definition}

In this paper we use the proof idea for $\textsf{IP} = \textsf{PSPACE}$ to show that \BPSATPAIR can be reduced to \GAPMAXTT{$\epsilon$} / \GAPMINTT{$\epsilon$} of certain size in near-linear time. The proof is in Section~\ref{sec:tensor}.

\subsection{Longest Common Subsequence}
We study the hardness of \textit{Longest Common Subsequence} (\LCS) and its pair version in this paper.
\begin{definition}[\LCS]
Given two strings $a, b$ of length $N$ over alphabet $\Sigma$, the \textsf{LCS} problem asks to compute the length of the longest sequence that appears in both $a$ and $b$ as a subsequence.
\end{definition}

\begin{definition}[\MAXLCSPAIR~/ \MINLCSPAIR] Given two sets of $N$ strings $A, B$, the \MAXLCSPAIR (or \MINLCSPAIR) problem asks to compute the maximum (or minimum) length of the longest common subsequence among all pairs of $(a, b) \in A \times B$.
\end{definition}

\subsection{Subtree Isomorphism and Largest Common Subtrees}

We study the hardness for the following two problems on trees:
\begin{definition}(Subtree Isomorphism)
Given two trees $G$ and $H$, the \textit{Subtree Isomorphism} problem asks whether $G$ is isomorphic to a subtree of $H$, i.e., can $G$ and $H$ be isomorphic after removing some nodes and edges from~$H$.
\end{definition}
\begin{definition}(Largest Common Subtree)
Given two trees $G$ and $H$, the \textit{Largest Common Subtree} problem asks to compute the size of the largest tree that is isomorphic to both a subtree of~$G$ and a subtree of~$H$.
\end{definition}

In this paper, we focus on the case of unordered trees with bounded degrees. We are interested in both rooted and unrooted trees.
Here ``rooted'' means that the root of $G$ must be mapped to the root of $H$ in the isomorphism.

The pair versions of these two problems are defined as follows:

\begin{definition}[\STIPAIR] Given two sets of $N$ trees $A, B$, the \textsf{Subtree-Isomorphism-Pair} problem asks whether there is a pair of trees $(a, b) \in A \times B$ such that the tree $a$ is isomorphic to a subtree of the tree $b$.
\end{definition}

\begin{definition}[\MAXLCSTPAIR~/ \MINLCSTPAIR] Given two sets of $N$ trees $A, B$, the \MAXLCSTPAIR (or \MINLCSTPAIR) problem asks to compute the maximum (or minimum) size of the largest common subtrees among all pairs of $(a, b) \in A \times B$.
\end{definition}

\subsection{Regular Expression Membership Testing}

We study the hardness of testing membership for regular expression. A regular expression over an alphabet set $\Sigma$ and an operator set $O = \{\,\circ,\,\mid\,, ~^+, ~^*\}$ is defined in a inductive way: 
(1) Every $a \in \Sigma$ is a regular expression;
(2) All of $[R \mid S]$, $R \circ S$, $R$, $[R]^+$ are regular expressions if $R$ and $S$ are regular expressions. 
A regular expression $p$ determines a language $L(p)$ over alphabet $\Sigma$. Specifically, for any regular expressions $R, S$ and any $a \in \Sigma$, we have: $L(a) = \{a\}$; $L([R \mid S]) = L(R) \cup L(S)$; $L(R \circ S) = \{uv \mid u \in L(R), v \in L(S)\}$; $L([R]^+) = \bigcup_{k \ge 1} \{u_1 u_2 \cdots u_k \mid u_1, \dots, u_k \in L(R)\}$; and $L([R]^*) = L(R^+) \cup \{\epsilon\}$, where $\epsilon$ denotes the empty string. The concatenation operator $\circ$ and unnecessary parenthesis is often omitted if the meaning is clear from the context.

In this paper, we study the \textit{Regular Expression Membership Testing} problem, which is defined as follows:
\begin{definition}[Exact Regular Expression Membership Testing]
Given a regular expression $p$ of length $M$ and a string $t$ of length $N$ over alphabet $\Sigma$, the \textit{Exact Regular Expression Membership Testing} problem asks whether $t$ is in the language $L(p)$ of $p$.
\end{definition}
And its pair version is defined as follows:
\begin{definition}[\REGEXPSTRPAIR]
Given a set $A$ of regular expressions of length $O(\poly(D))$and a set $B$ of $N$ strings of length $D$, the \REGEXPSTRPAIR problem asks to determine whether there is a pair $(a, b)$ such that $b$ is in the language $L(a)$ of $a$.
\end{definition}

In \cite{ARW17-proceedings}, Abboud, Rubinstein and Williams studied a problem called \textsf{RegExp Closest Pair} and showed that it is \textsf{SETH}-hard using their distributed PCP framework. In this work, we study a slightly different problem.
\begin{definition}[\CLOSESTREGEXPSTRPAIR]
For two strings $x, y$ of the same length $n$, the Hamming Similarity $\hamsim(x, y)$ between $x$ and $y$ is defined as the fraction of positions for which the corresponding symbols are equal, i.e.,
\[
\hamsim(x, y) = \frac{1}{n}\sum_{i = 1}^{n} \mathbbm{1}_{[x_i = y_i]}.
\]
Given a set $A$ of $N$ regular expressions of length $O(\poly(D))$ and a set of $N$ strings of length $D$, the \CLOSESTREGEXPSTRPAIR problem asks to compute the maximum Hamming Similarity among all pairs of $(x, b)$ satisfying $x \in L(a)$ is a string of length $D$ for some $a \in A$ and $b \in B$.
\end{definition}
\section{$\equivclass$ and an Outline of all Reductions}
\label{sec:outline}

In this section, we first state the equivalence class formally, and outline how it is proved in this paper.

\begin{theorem}[$\equivclass$]\label{theo:informal-equiv-class}
	There are near linear-time\footnote{Throughout this paper, we use near-linear time to denote the running time of $N^{1 + o(1)}$.}
	reductions between all pairs of following problems:
	\begin{enumerate}
		\item (Exact or Approximate \MAXLCSPAIR~/~\MINLCSPAIR) Given two sets $A ,B$ of $N$ strings of length $D = 2^{(\log N)^{o(1)}}$, compute the exact value or a $2^{(\log D)^{1 - \Omega(1)}}$-approximation of the maximum (minimum) $\LCS$ among $(a,b) \in A \times B$;
		
		\item (Approximate \CLOSESTREGEXPSTRPAIR) Given a set $A$ of $N$ regular expressions of length $2^{(\log N)^{o(1)}}$ and a set of $N$ strings of length $D = 2^{(\log N)^{o(1)}}$, distinguish between the case that there is a pair $(a, b) \in A \times B$ such that $b \in L(a)$ (the language of $a$), and the case that every string in $B$ has Hamming Similarity $< 2^{-(\log D)^{1 - \Omega(1)}}$ from every string of length $D$ in $\bigcup_{a \in A} L(a)$.
		
		\item (Exact or Approximate \MAXLCSTPAIR~/~\MINLCSTPAIR) Given two sets $A, B$ of $N$ bounded-degree trees with size at most $D = 2^{(\log N)^{o(1)}}$, compute the exact value or a $2^{(\log D)^{1 - \Omega(1)}}$-approximation of the maximum (minimum) size of largest common subtree among $(a,b) \in A \times B$;
		
		\item (Exact or Approximate \MAXTT~/~\MINTT) Given two sets $A, B$ of $N$ binary tensors with size $D = 2^{(\log N)^{o(1)}}$, compute the exact value or a $2^{(\log D)^{1 - \Omega(1)}}$-approximation of the maximum (minimum) Tropical Similarity among $(a,b) \in A \times B$;		
		
		\item \OAPTfull(\OAPT): Given two sets $A, B$ of $N$ binary tensors with size $2^{(\log N)^{o(1)}}$, is there a pair $(a,b) \in A \times B$ with Alternating Product\footnote{see Definition~\ref{def:oapt} for a formal definition} $0$?
		
		\item \BPSATPAIR;
		\item \REGEXPSTRPAIR;
		\item \STIPAIR;
	\end{enumerate}
\end{theorem}

\begin{remark}
A technical remark is that our reductions actually have a quasi-polynomial blow-up on the string length (tensor / tree size) $D$. That is, the new string length after the transformation would be at most $D' = 2^{\polylog(D)}$, which is still $2^{(\log N)^{o(1)}}$ assuming $D = 2^{(\log N)^{o(1)}}$. That is the reason we set the size parameter to be $2^{(\log N)^{o(1)}}$ in the equivalence class. 
\end{remark}

For the ease of exposition, we break the proofs for Theorem~\ref{theo:informal-equiv-class} and Theorem~\ref{theo:informal-hard} into many sections, each one dealing with one kind of problems. Here we give an outline of the reductions for proving Theorem~\ref{theo:informal-equiv-class} and Theorem~\ref{theo:informal-hard} (Figure~\ref{fig:reduction-graph}).

\begin{figure}[ht]
	\centering
\begin{tikzpicture}[->,>=stealth',shorten >=1pt,auto,
semithick,scale = 2]
\tikzstyle{every state}=[draw=black,text=black,rectangle,font=\small]
\tikzstyle{reduceto}=[thick]

\node [state] (BPSATPAIR) at (-1.5,1) {\BPSATPAIR};
\node [state] (OAPT) at (-1.5,0) {\OAPT};
\node [state,text width=4cm,align=center] (MAXTT) at (2.7,0)
	{\GAPMAXTT{$\epsilon$} \\
	 \GAPMINTT{$\epsilon$}};

\node [state,text width=5cm,align=center] (REGEXPSTRPAIR) at (0,-1)
	{\REGEXPSTRPAIR};

\node [state,text width=5cm,align=center,color=red] (EXACTREGEXP) at (0,-1.7)
	{Regular Expression \\ Membership Testing};

\node [state,text width=5cm,align=center] (STIPAIR) at (0,-2.4)
	{\STIPAIR};

\node [state,text width=5cm,align=center,color=red] (STI) at (0,-3.1)
	{Subtree Isomorphism};

\node [state,text width=5cm,align=center] (REGEXPPAIR) at (3.25,-1)
	{Approx. \\ \CLOSESTREGEXPSTRPAIR};

\node [state,text width=5cm,align=center] (LCSPAIR) at (3.25,-3.1)
	{Approx. \MAXLCSPAIR \\ Approx. \MINLCSPAIR};
\node [state,text width=5cm,align=center] (LCSTPAIR) at (3.25,-1.7)
	{Approx. \MAXLCSTPAIR \\ Approx. \MINLCSTPAIR};
\node [state,text width=5cm,align=center,color=red] (LCST) at (3.25,-2.4)
	{Approx. \\Largest Common Subtree};

\draw[reduceto] (BPSATPAIR.-70) to (OAPT.70);
\draw[reduceto] (OAPT.110) to (BPSATPAIR.-110);

\draw[reduceto] (OAPT) |- (MAXTT);
\draw[reduceto] (OAPT) |- (REGEXPSTRPAIR);
\draw[reduceto] (OAPT) |- (EXACTREGEXP);
\draw[reduceto] (OAPT) |- (STI);
\draw[reduceto] (OAPT) |- (STIPAIR);
\draw[reduceto] (MAXTT.195) |- (LCSPAIR);
\draw[reduceto] (MAXTT.195) |- (LCSTPAIR);
\draw[reduceto] (MAXTT.195) |- (LCST);
\draw[reduceto] (MAXTT.195) |- (REGEXPPAIR);

\draw[reduceto] (REGEXPSTRPAIR) -| (1.5,-4);
\draw[-,thick] (STIPAIR) -| (1.5,-4) -- (4.75,-4);
\draw[-,thick] (REGEXPPAIR) -| (4.75,-4);
\draw[-,thick] (LCSTPAIR) -| (4.75,-4);
\draw[reduceto] (LCSPAIR) -| (4.75,-4);
\draw[-,thick] (MAXTT) -| (4.75,-4) -- (1.5,-4);

\draw[thick] (4.75,-4) -| (BPSATPAIR.210);
\end{tikzpicture}
\caption{A diagram for all our reductions. (red means it is $\equivclasshard$.)}
\label{fig:reduction-graph}
\end{figure}

\begin{itemize}
	\item[Section~\ref{sec:to-BPSATPAIR}] \textbf{$\BPSATPAIR$.} We present a generic reduction from \ASATPAIR{$\mathcal{A}$} and \MAXAPAIR{$\mathcal{A}$} / \MINAPAIR{$\mathcal{A}$} problems to $\BPSATPAIR$ (Theorem~\ref{thm:polylogspace-sat} and Theorem~\ref{thm:polylogspace-min-max}). This implies all problems in $\equivclass$ can be reduced to $\BPSATPAIR$.
	
	\item[Section~\ref{sec:tensor}] \textbf{Tensors Problems.} We show reductions from $\BPSATPAIR$ to tensors problems $\OAPT$, approximate $\MAXTT$ and $\MINTT$ (Theorem~\ref{thm:hard-oapt}, Theorem~\ref{thm:hard-gap-tt} and Theorem~\ref{thm:hard-gap-min-tt}), putting these tensor problems into our $\equivclass$ (Theorem~\ref{thm:OPAT-BPPAIR-equiv} and Theorem~\ref{thm:tt-equiv}).
	
	\item[Section~\ref{sec:LCS}] \textbf{$\LCS$.} We show reductions from approximate $\MAXTT$ ($\MINTT$) to approximate \MAXLCSPAIR\ ($\MINLCSPAIR$) (Theorem~\ref{thm:hard-gap-max-lcs-pair} and Theorem~\ref{thm:hard-gap-min-lcs-pair}), putting these $\textsf{LCS-Pair}$ problems into our $\equivclass$ (Theorem~\ref{thm:lcs-pair-equiv}).
	
	\item[Section~\ref{sec:RegularExp}] \textbf{Regular Expression.} We show reductions from $\OAPT$ to $\REGEXPSTRPAIR$ and from approximate $\MAXTT$ to approximate $\CLOSESTREGEXPSTRPAIR$ (Theorem~\ref{thm:re-str-mem-gadget-oapt} and Corollary~\ref{cor:MAXTT-to-CRSP}), putting these regular expression pair problems into our $\equivclass$ (Theorem~\ref{thm:regexp-string-pair-equiv}). 
	
	We also show a reduction from $\OAPT$ to Regular Expression Membership Testing, showing the latter problem is $\equivclasshard$ (Theorem~\ref{thm:regexp-string-hard}).
	
	\item[Section~\ref{sec:subtree-Iso}] \textbf{Subtree isomorphism.} We show reductions from $\OAPT$ to $\STIPAIR$ and from approximate $\MAXTT$ ($\MINTT$) to approximate $\MAXLCSTPAIR$ ($\MINLCSTPAIR$) (Theorem~\ref{thm:hard-sti-pair} and Theorem~\ref{thm:hard-approx-lcst-pair}), putting these problems related to subtree isomorphism into our $\equivclass$ (Theorem~\ref{thm:bppair-stipair-equiv} and Theorem~\ref{thm:lcst-pair-equiv}).
	
	We also show reductions from $\OAPT$ to Subtree Isomorphism and from approximate $\MAXTT$ to approximate Largest Common Subtree, showing that the latter problems are $\equivclasshard$ (Theorem~\ref{thm:sti-hard} and Theorem~\ref{thm:lcst-hard}).
\end{itemize}

\section{Low-Space Algorithm Implies Reduction to \BPSATPAIR}\label{sec:to-BPSATPAIR}

In this section, we present two important theorems for showing reductions from \ASATPAIR{$\mathcal{A}$}, \MAXAPAIR{$\mathcal{A}$} / \MINAPAIR{$\mathcal{A}$} problems to \BPSATPAIR. 

The key observation is a classic result in space complexity: for $S(n) \ge \log n$, if a decision problem $\mathcal{A}$ is in $\NSPACE[S(n)]$, then there is a BP of length $T = 2^{O(S(n))}$ and width $W = 2^{O(S(n))}$ that decides $\mathcal{A}$ (See, e.g., \cite{AB09-book} for the proof). This means that if $\mathcal{A}$ can be solved in small space, then we can construct a BP of not too large size to represent this algorithm.

Now we introduce our first theorem, which shows that a low-space algorithm for a decision problem $\mathcal{A}$ implies a reduction from \ASATPAIR{$\mathcal{A}$} to \BPSATPAIR:

\begin{theorem} \label{thm:polylogspace-sat}
If the decision problem $\mathcal{A}$ on inputs $a, b$ of length $n$ can be decided in $\NSPACE\left[\polylog(n)\right]$, then \ASATPAIR{$\mathcal{A}$} with two sets of $N$ strings of length $2^{(\log N)^{o(1)}}$ can be reduced to \BPSATPAIR on branching program of size $2^{(\log N)^{o(1)}}$ in near-linear time.
\end{theorem}
\begin{proof}
Let $n = 2^{(\log N)^{o(1)}}$. Since $\mathcal{A}$ is in $\NSPACE\left[\polylog(n)\right]$, we can construct a BP $P$ of size $2^{\polylog(n)} \le 2^{(\log N)^{o(1)}}$ that decides $\mathcal{A}$ on inputs $a, b$ of length $n$. Then to check if there is a pair of $(a, b) \in A \times B$ such that $(a, b)$ is an Yes-instance of $\mathcal{A}$, it is sufficient to check if there is a pair of $a, b$ making $P$ accept.
\end{proof}

Our second theorem is similar. It shows that a low-space algorithm for the decision problem of an optimization problem $\mathcal{A}$ implies a reduction from \MAXAPAIR{$\mathcal{A}$} / \MINAPAIR{$\mathcal{A}$} to \BPSATPAIR:

\begin{theorem} \label{thm:polylogspace-min-max}
Let $\mathcal{A}$ be an optimization problem. If the answer to $\mathcal{A}$ on input $a, b$ of length $n$ is bounded in $[-O(\poly(n)), O(\poly(n))]$, and the decision version of $\mathcal{A}$ (deciding whether the answer is greater than a given number $k$) can be decided in $\NSPACE\left[\polylog(n)\right]$, then \MAXAPAIR{$\mathcal{A}$} (or \MINAPAIR{$\mathcal{A}$}) with two sets of $N$ strings of length $2^{(\log N)^{o(1)}}$ can reduced to $2^{(\log N)^{o(1)}}$ instances of \BPSATPAIR with branching programs of size $2^{(\log N)^{o(1)}}$ in near-linear time.
\end{theorem}
\begin{proof}
For \MAXAPAIR{$\mathcal{A}$}, we enumerate each possible answer $k$, and then we check if there is a pair of $(a, b) \in A \times B$ with answer $> k$ by reducing to \BPSATPAIR via Theorem \ref{thm:polylogspace-sat}. This reduction results in $O(\poly(n)) \le 2^{(\log N)^{o(1)}}$ instances of \BPSATPAIR, and each branching program is of size $2^{(\log N)^{o(1)}}$. Note that $\NSPACE\left[\polylog(n)\right] = \co\NSPACE\left[\polylog(n)\right]$, thus deciding whether the answer of $\mathcal{A}$ is $\le k$ is also in $\NSPACE\left[\polylog(n)\right]$. Using a similar argument we can reduce \MINAPAIR{$\mathcal{A}$} to \BPSATPAIR.
\end{proof}
\section{Tensor Problems} \label{sec:tensor} 

In this section, we show that \BPSATPAIR on branching program of size $2^{(\log N)^{o(1)}}$ is equivalent to \OAPT and (exact or approximate) \MAXTT / \MINTT problems on tensors of size $2^{(\log N)^{o(1)}}$ under near-linear time reductions. 

\subsection{Orthogonal Alternating Product Tensors}

First we show the equivalence between \BPSATPAIR and \OAPT. To start with, we present the reduction from \BPSATPAIR to \OAPT:

\begin{theorem} \label{thm:hard-oapt}
There exists an $O(N \cdot 2^{O(\log W \log T)})$-time reduction from a \BPSATPAIR instance with a branching program of length $T$ and width $W$ and two sets of $N$ strings to an \OAPT problem with two sets of $N$ tensors of size $2^{O(\log W \log T)}$.
\end{theorem}
\begin{proof}
Let $P$ be a branching program of length $T$ and width $W$ on $n$ boolean inputs $x = (x_1, \dots, x_{n})$.
First, we follow the proof for the \PSPACE-completeness of \textsf{TQBF} \cite{stockmeyer1973word} to construct a quantified boolean formula $\phi(x)$, which holds true iff the branching program $P$ accepts $x$.
Then, we construct two sets $A', B'$ of $N$ tensors such that there is a pair $(a, b) \in A \times B$ satisfying $\phi(a, b)$ is true iff there is a pair $(a', b') \in A' \times B'$ with $\aprod(a', b') = 0$.

\paragraph{Construction of Quantified Boolean Formula. } We assume that $n, T - 1, W$ are powers of two without loss of generality. First we construct formulas $\psi_k(x, u, v, i)$ for all $0 \le k \le \log (T - 1), u, v \in [W], i \in [T]$ such that $\psi_k(x, u, v, i)$ holds true iff the node $(i + 2^k, v)$ is reachable from the node $(i, u)$ on the input $x = (x_1, \dots, x_{n})$.

The construction is by induction on $k$.
For $k = 0$, we split the $n$ input variables $x = (x_1, \dots, x_n)$ into two halves: $\xa = (x_1, \dots, x_{n/2})$ and $\xb = (x_{n/2+1}, \dots, x_n)$. We construct two formulas $\alpha(\xa, u, v, i)$ and $\beta(\xb, u, v, i)$. We construct the formula $\alpha$ to be true iff the variable $x_{f(i)}$ associated with the layer $L_i$ is in $\xa$, and there is an edge that goes from the node $(i, u)$ to the node $(i + 1, v)$ and is marked with the value of $x_{f(i)}$. We define the formula $\beta$ similarly for $\xb$. Then we construct $\psi_0(x, u, v, i) = \alpha(\xa, u,v,i) \lor \beta(\xb, u,v,i)$. It is easy to see that $\psi_0(x, u, v, i)$ holds true iff the node $(i + 1, v)$ is reachable from the node $(i, u)$ on the input $x = (x_1, \dots, x_{n})$.
For $k \ge 1$, we construct $\psi_{k}(x, u, v, i)$ as:
\[
(\exists m)(\forall u')(\forall v')(\forall j)[((u', v', j) = (u, m, 0) \lor (u', v', j) = (m, v, 1)) \Rightarrow \psi_{k-1}(x, u', v', i + j \cdot 2^{k-1})] 
\]
where $m, u', v' \in [W]$ and $j \in \{0, 1\}$. It is easy to see that the above formula is equivalent to
\[
(\exists m)[\psi_{k-1}(x, u, m, i) \land \psi_{k-1}(x, m, v, i + 2^{k-1})],
\]
thus it holds true iff the node $(i + 2^k, v)$ is reachable from the node $(i, u)$.

In the end, we construct the formula $\varphi(x) = \psi_{\log(T-1)}(x, 1, 1, 1)$, so $\varphi(x)$ holds true iff the branching program $P$ accepts the input $x = (x_1, \dots, x_{n})$ (Recall that $u_{\mathrm{start}} = (1, 1)$ and $u_{\mathrm{acc}} = (T, 1)$).

We split all the variables $m, u', v', j$ occurred in $\varphi(x)$ into $t$ boolean variables $z_1, \dots, z_t \in \{0, 1\}$ for some $t = O(\log W \log T)$. Without loss of generality we assume $t$ is even. Then we transform $\varphi(x)$ into the following equivalent formula $\phi$:
\[
\phi(x) = (\exists z_1)(\forall z_2)(\exists z_3) \cdots  (\forall z_t) ( f(z) \Rightarrow (g_1(\xa, z) \lor g_2(\xb,z)))
\]
where $f(z)$ is the logical conjunction of all the predicates $((a, b, j) = (u, m, 0) \lor (a, b, j) = (m, v, 1))$ in $\psi_{1}, \dots, \psi_{\log(T - 1)}$, and $g_1(\xa, z), g_2(\xb, z)$ are the innermost $\alpha(\xa, u, v, i), \beta(\xb, u, v, i)$. The quantifiers $\forall$ and $\exists$ appear alternatively.

\paragraph{Converting Quantified Boolean Formula into Tensors.} Let $d_1 = \cdots = d_t = 2$. Now we construct two sets of $N$ tensors $A', B' \subseteq \{0, 1\}^{d_1 \times \cdots \times d_t}$ to be our \textsf{OAPT} instance. For $1 \le k \le t$, we associate the $k$-th dimension of a tensor with the variable $z_k$ and associate each index $p \in [d_1] \times \cdots \times [d_t]$ with an assignment to $z_1, \dots, z_t$.
Note that strings in the set $A$ correspond to assigments to $\xa$, and strings in the set $B$ correspond to assigments to $\xb$. Thus every two strings $(a, b) \in A \times B$ along with an index $p$ specify an assignment to $x$ and $z$.

For each string $a \in A$, we construct a tensor $a' \in \{0, 1\}^{d_1 \times \cdots \times d_t}$ where for every index $p$, $a'_{p}$ is $0$ iff the formula $\lnot f(z) \lor g_1(\xa, z)$ is true with corresponding assignments to $\xa$ and $z$; for each string $b \in B$, we construct a tensor $b' \in \{0, 1\}^{d_1 \times \cdots \times d_t}$ where for every index $p$, $b'_{p}$ is $0$ iff the formula $g_2(\xa, z)$ is true with corresponding assignments to $\xb$ and $z$.

Note that $f(z) \Rightarrow (g_1(\xa, z) \lor g_2(\xb,z))$ is equivalent to $\lnot f(z) \lor g_1(\xa, z) \lor g_2(\xb,z)$, so we have
\[
a'_p \land b'_p = 0 \iff [f(z) \Rightarrow (g_1(\xa, z) \lor g_2(\xb,z))].
\]
Then it is easy to see that $\aprod(a', b') = 0$ iff $\phi(a, b)$ is true, and thus $P$ accepts a pair of $(a, b) \in A \times B$ iff $p_{\mathrm{alt}}(a', b') = 0$ for their corresponding $a', b'$.
Note that $d_1 d_2 \cdots d_t = 2^t = 2^{O(\log W \log T)}$, so each tensor has size $2^{O(\log W \log T)}$.
\end{proof}

Note that in the above construction, tensors in $B$ are all $\land$-invariant, so we have the following corollary for \OAPTrestri:
\begin{cor} \label{thm:hard-oapt-restricted}
There exists an $O(N \cdot 2^{O(\log W \log T)})$-time reduction from a \BPSATPAIR instance with a branching program of length $T$ and width $W$ and two sets of $N$ strings to an \OAPTrestri problem with two sets of $N$ tensors of size $2^{O(\log W \log T)}$.
\end{cor}

For the other direction, by Theorem \ref{thm:polylogspace-sat}, it is sufficient to show that computing Alternating Product can be done in $\SPACE[O(\log n)] \subseteq \NSPACE[\polylog(n)]$.

\begin{lemma}\label{lam:aprod-polylog}
Given two tensors $a, b$ of size $n = 2^t$, their Alternating Product $\aprod(a, b)$ can be computed in $\SPACE[O(\log n)]$.
\end{lemma}
\begin{proof}
We compute the Alternating Product recursively according to the definition. There are $t$ levels of recursion in total. Since $t = \log n$, space $O(\log n)$ is enough for our algorithm.
\end{proof}

Combining Theorem \ref{thm:hard-oapt} and Lemma \ref{lam:aprod-polylog}, we can prove the equivalence between \BPSATPAIR and \OAPT.
\begin{theorem}\label{thm:OPAT-BPPAIR-equiv}
The \textsf{OAPT} problem on tensors of size $2^{(\log N)^{o(1)}}$ is equivalent to \textsf{BP-Satisfying-Pair} on branching program of size $2^{(\log N)^{o(1)}}$ under near-linear time reductions.
\end{theorem}

\subsection{A Communication Protocol for Branching Program}

Before we turn to show the equivalence between \BPSATPAIR and \MAXTT~/ \MINTT, we introduce the following \IP-protocol for branching program. Our reduction from \BPSATPAIR to \MAXTT (or \MINTT) directly follows by simulating the communication protocol using tropical algebra.

\begin{theorem} \label{thm:bp-ip}
Let $P$ be a branching program of length $T$ and width $W$ on $n$ boolean inputs $x_1, \dots, x_{n}$. Suppose Alice holds the input $x_1, \dots, x_{n/2}$ and Bob holds the input $x_{n/2+1}, \dots, x_{n}$. For every $\epsilon > 0$, there exists a computationally efficient \IP-protocol for checking whether $P$ accepts on $x_1, \dots, x_{n}$, in which:
\begin{enumerate}
\item Merlin and Alice exchange $O(\log^2 W \log^2 T \cdot (\log \log W + \log \log T + \log \epsilon^{-1}))$ bits;
\item Alice tosses $O(\log^2 W \log^2 T \cdot (\log \log W + \log \log T + \log \epsilon^{-1}))$ public coins;
\item Bob sends $O(\log \log W + \log \log T + \log \epsilon^{-1})$ bits to Alice;
\item Alice accepts or rejects in the end.
\end{enumerate}
If $P$ accepts on the input $x_1, \dots, x_{n}$, then Alice always accepts; otherwise, Alice rejects with probability at least $1 - \epsilon$.
\end{theorem}
\begin{proof}
Let $\bar{a}$ be the assignment to the input variables held by Alice, and $\bar{b}$ be the assignment to the input variables held by Bob. Recall the construction of the tensors in the proof for Theorem \ref{thm:hard-oapt}. First Alice constructs a tensor $a = G(\bar{a})$, and Bob constructs a tensor $b = H(\bar{b})$. Each tensor here is of shape $d_1 \times d_2  \times \cdots d_t = 2 \times 2 \times \cdots \times 2$ for $t = O(\log T \log W)$. Then the problem reduces to check whether the Alternating Product $\aprod(a, b)$ equals $0$. Now we show that there exists a communication protocol for checking $\aprod(a, b) = 0$, using the idea for proving $\IP = \PSPACE$ \cite{LundFKN92,Shamir92}.

\paragraph{Arithmetization.} First we arithmetize the computation of Alternating Product. Let $q \ge 1$ be a parameter to be specified. Construct a finite field $\F_{2^q}$. Then Alice finds a multilinear extension $\alpha$ over $\F_{2^q}$ for her tensor $a$, i.e., Alice finds a function $\alpha(z_1, \dots, z_t)$ such that $\alpha$ is linear in each of its variables, and $\alpha(z_1, \dots, z_t) = a_i$ for all $i \in [d_1] \times \cdots \times [d_t]$ and $i_k = z_k + 1$ ($1 \le k \le t$). Bob finds a multilinear extension $\beta$ for his tensor $b$ similarly.
Recall that the definition of Alternating Product. $p_{\mathrm{a}}(a, b)$ can be rewritten as
\[
\aprod(a, b) = \bigwedge_{z_1 \in \{0, 1\}} \left[\bigvee_{z_2 \in \{0, 1\}}\left( \bigwedge_{z_3 \in \{0, 1\}} \left[ \cdots \bigvee_{z_t \in \{0, 1\}} \left(\alpha(z_1, \dots, z_t) \cdot \beta(z_1, \dots, z_t) \right) \cdots \right] \right)\right].
\]

To arithmetize $\bigwedge_{z_k \in \{0, 1\}}$ and $\bigvee_{z_k \in \{0, 1\}}$, we define three kinds of operators acting on polynomials:
\begin{enumerate}
\item $\Pi_{z_m}$ operator, which arithmetizes the formula $\bigwedge_{z_m \in \{0, 1\}} F(z_1, \dots, z_{m-1}, z_m)$.
\[\Pi_{z_m} F(z_1, \dots, z_m) = F(z_1, \dots, z_{m-1}, 0) \cdot F(z_1, \dots, z_{m-1}, 1)\]
\item $\Sigma_{z_m}$ operator, which arithmetizes the formula $\bigvee_{z_m \in \{0, 1\}} F(z_1, \dots, z_{m-1}, z_m)$.
 \[\Sigma_{z_m} F(z_1, \dots, z_m) = 1-(1 - F(z_1, \dots, z_{m-1}, 0)) \cdot(1 - F(z_1, \dots, z_{m-1}, 1)).\]
\item $\mathcal{R}_{z_i}$ operator, which is used for the degree reduction. When acting on a polynomial $F(z_1, \dots, z_m)$, it replaces $z_i^k$ for $k \ge 1$ by $z_i$ in all terms. In this way, any polynomial $F(z_1, \dots, z_m)$ can be converted into a multilinear one preserving the values at every $(z_1, \dots, z_m) \in \{0, 1\}^m$. $\mathcal{R}_{z_i}$ operator can be written as
\[
\mathcal{R}_{z_i} F(z_1, \dots, z_m) = F(z_1, \dots, z_{i-1}, 0, z_{i+1}, \dots z_m) + z_i \cdot F(z_1, \dots, z_{i-1}, 1, z_{i+1}, \dots z_m).
\]
\end{enumerate}

Then it is easy to see that
\[
\aprod(a, b) = \Pi_{z_1} \Sigma_{z_2} \Pi_{z_3} \cdots \Sigma_{z_t} (\alpha \cdot \beta).
\]
Note that in the computation of Alternating Product, we only use the function value at Boolean inputs, thus we can insert $i$ operators $R_{z_1} R_{z_2} \cdots R_{z_i}$ right after each $\pi_{z_i}$ or $\Sigma_{z_i}$ without changing the final result:
\[
\aprod(a, b) = \Pi_{z_1} R_{z_1} \Sigma_{z_2} R_{z_1} R_{z_2} \Pi_{z_3} R_{z_1} R_{z_2} R_{z_3} \cdots \Sigma_{z_t} R_{z_1} \cdots R_{z_t} (\alpha \cdot \beta).
\]
In total we use only $M = O(t^2) \le O(\log^2 T \log^2 W)$ operators.

\paragraph{The Protocol.} We introduce our \IP-protocol in an inductive way.
Suppose that we have an \IP-protocol for some polynomial $F(z_1, \dots, z_m)$, in which for any given $(v_1, \dots, v_m) \in \F^m_{2^q}$ and $u = F(v_1, \dots, v_m)$, Merlin can convince Alice and Bob that $F(v_1, \dots, v_m) = u$ with perfect completeness and soundness error $\epsilon_0$. We show that for $G(z_1, \dots, z_{m'}) = \mathcal{O}_{z_i} F(z_1, \dots, z_m)$ and given $v_1, \dots, v_{m'}$ and $u'$ ($\mathcal{O}_{z_i} \in \{ \Sigma_{z_i}, \Pi_{z_i}, \mathcal{R}_{z_i}\}$, $m' = m$ when $\mathcal{O}_{z_i} = \mathcal{R}_{z_i}$ and $m' = m - 1$ otherwise), Merlin can convince Alice and Bob that $G(v_1, \dots, v_{m'}) = u'$ with perfect completeness and soundness error $\epsilon_0 + O(2^{-q})$:
\begin{itemize}
\item First Merlin sends the coefficients of the polynomial $F(v_1, \dots, v_{i-1}, z_i, v_{i+1}, \dots, v_m)$ to Alice (note that it is a univariate polynomial of $z_i$);
\item Alice calculates the value of $G(v_1, \dots, v_{m'})$ using the information sent by Merlin (assuming Merlin is honest), and reject if $G(v_1, \dots, v_{m'}) \neq u'$;
\item Alice randomly draws an element $r \in \F_{2^q}$. Let $v_{i} = r$ (reset $v_i = r$ if $v_i$ already exists);
\item Alice checks if $F(v_1, \dots, v_{i-1}, r, v_{i+1}, \dots, v_m) = u$ via the \IP-protocol for $F$.
\end{itemize}
It is easy to see that the above protocol has perfect completeness. For soundness, notice that $F$ and $G$ are always of $O(1)$ degree because our use of degree reduction operators, thus Alice can find $G(v_1, \dots, v_{m'}) \neq u'$ with probability $O(2^{-q})$ if Merlin lies. By the union bound, the soundness error of the \IP-protocol for $G$ is $\epsilon_0 + O(2^{-q})$.

Our \IP-protocol starts by checking $\aprod(a, b) = 0$. Following the inductive process above, there are $M$ rounds of communication between Merlin and Alice. And after the last round, $\aprod(a, b) = 0$ reduces to check if $\alpha(v_1, \dots, v_t) \cdot \beta(v_1, \dots, v_t) = u$ for given $(v_1, \dots, v_t) \in \F_{2^q}^t, u \in \F_{2^q}$.
Note that all the values of $v_1, \dots, v_t$ can be inferred by the results of public coins Alice tossed. Thus the \IP-protocol for $\alpha \cdot \beta$ is as follows: Bob learns the results of public coins and obtains $v_1, \dots, v_t$. Then Bob sends the value of $\beta(v_1, \dots, v_t)$ to Alice. Finally, Alice accepts iff $\alpha(v_1, \dots, v_t) \cdot \beta(v_1, \dots, v_t) = u$.

By induction, we can show that the whole \IP-protocol has perfect completeness and soundness error $O(M \cdot 2^{-q})$. Setting $2^q = c \cdot M \cdot \epsilon^{-1}$ for large enough constant $c$, we can achieve the soundness error $\epsilon$. And in this case we have
$$
q = \log M + \log \epsilon^{-1} + \log c = O(\log \log W + \log \log T + \log \epsilon^{-1}).
$$

It can be easily seen that Alice tosses 
$$
O(Mq) = O(\log^2 T \log^2 W (\log \log T + \log \log W + \log \epsilon^{-1}))
$$ 
public coins and Bob sends 
$$
O(q) = O(\log \log T + \log \log W + \log \epsilon^{-1})
$$ 
bits to Alice in our communication protocol. 

In each of the $M$ rounds, Merlin sends $O(1)$ elements in $\F_{2^q}$ since $F$ is of at most constant degree.
Thus Merlin sends at most 
$$
O(Mq) = O(\log^2 T \log^2 W (\log \log T + \log \log W + \log \epsilon^{-1}))
$$ bits to Alice.
\end{proof}

\subsection{Tropical Tensors}

Following from \cite{AbboudR18}, we can show a reduction from \BPSATPAIR to \GAPMAXTT{$\epsilon$} based on our \IP-protocol for branching program.

\begin{theorem} \label{thm:hard-gap-tt}

There is a reduction from \BPSATPAIR on branching program of length $T$ and width $W$ and two sets of $N$ strings to \GAPMAXTT{$\epsilon$} on two sets of $N$ tensors of size 
$$
D = 2^{O(\log^2 W \log^2 T (\log \log W + \log \log T + \log \epsilon^{-1}))},
$$
and the reduction runs in $O(N \poly(D))$. Here $\epsilon$ is a threshold value that can depend on $N$.

\end{theorem}
\begin{proof}
For convenience, let 
$$
K = \log^2 W \log^2 T (\log \log W + \log \log T + \log \epsilon^{-1}).
$$
By Theorem \ref{thm:bp-ip}, there is an \textsf{IP}-protocol using $O(K)$ bits for determining whether a branching program accepts when Alice knows the first half and Bob knows the second half, with soundness error $\epsilon$. We can easily modify the communication protocol such that
\begin{itemize}
\item Alice and Merlin interact for $m = O(K)$ rounds, in each round Merlin sends one bit to Alice and Alice tosses one public coin;
\item After the interaction between Alice and Merlin, Bob sends $\ell = O(K)$ bits to Alice, and after Bob sending each bit Merlin sends a dummy bit to Alice;
\item Alice accepts or rejects in the end.
\end{itemize}

Let $t = 2\ell + 2m$ and $d_1 = \cdots = d_t = 2$. Now we construct two sets of $N = 2^{n/2}$ tensors $A, B \in \{0, 1\}^{d_1 \times \cdots \times d_t}$ as our \GAPMAXTT{$\epsilon$} instance. For every $0 \le k < m$, we associate the $(t - 2k)$-th dimension with the result of the public coin Alice tosses at the $(k + 1)$-th round and associate the $(t - 2k - 1)$-th dimension with the bit Merlin sends to Alice at the $(k+1)$-th round. For every $0 \le k < \ell$, we associate the $(2\ell - 2k)$-th dimension of a tensor with the $(k+1)$-th bit sent by Bob and associate the $(2\ell - 2k - 1)$-th dimension with the bit Merlin sent to Alice right after Bob sending the $(k+1)$-th bit to Alice. In this way, every index $p \in [d_1] \times \cdots \times [d_t]$ of a tensor can be seen as a communication transcript.

Let $A, B \subseteq \{0, 1\}^{n/2}$ be the two sets in the \textsf{BP-Satisfying-Pair} instance. For each assignment $a \in A$ to the first half variables $x_1, \dots, x_{n/2}$, we construct a tensor $G(a) \in \{0, 1\}^{d_1 \times \cdots \times d_t}$ where for every index $p$, $G(a)_{p}$ is $1$ iff Alice accepts after seeing the communication transcript $p$ when she holds the assignment $a$ for $x_1, \dots, x_{n/2}$. For each assignment $b \in B$ to the second half variables $x_{n/2+1}, \dots, x_{n}$, we construct a tensor $H(b) \in \{0, 1\}^{d_1 \times \cdots \times d_t}$ where for every index $p$, $H(b)_{p}$ is $1$ iff the bits sent by Bob in the communication transcript $p$ matches what Bob sends when he holds the assignment $b$ for $x_{n/2+1}, \dots, x_n$ and learning the results of Alice's public coins in $p$.

Let $D = 2^t = 2^{O(K)}$
be the size of each tensor. It is easy to see that when Alice holds $a$ and Bob holds $b$, the maximum probability (over all Merlin's actions) that Alice accepts in our communication protocol equals the Tropical Similarity $s(G(a), H(b))$.
\end{proof}

By negating the branching program $P$, we can also show a similar reduction from \BPSATPAIR to \GAPMINTT{$\epsilon$}:
\begin{theorem} \label{thm:hard-gap-min-tt}
There is a reduction from \BPSATPAIR on branching program of length $T$ and width $W$ and two sets of $N$ strings to \GAPMINTT{$\epsilon$} on two sets of $N$ tensors of size 
$$
D = 2^{O(\log^2 W \log^2 T (\log \log W + \log \log T + \log \epsilon^{-1}))},
$$ 
and the reduction runs in $O(N \poly(D))$. Here $\epsilon$ is a threshold value that can depend on $N$.
\end{theorem}
\begin{proof}
The \textsf{IP}-protocol in Theorem \ref{thm:bp-ip} can be easily adapted to check the branching program $P$ does \textit{not} accept, i.e., if $P$ rejects on the input $x_1, \dots, x_{n}$, then Alice always accepts; otherwise, Alice rejects with probability $1 - \epsilon$. To do this, the only thing we need to change is to check whether the Alternating Product is $1$ rather than $0$.
Then, using the same reduction as in Theorem \ref{thm:hard-gap-tt}, we can obtain two sets $A' = \{G(a) \mid a \in A\}, B' = \{H(b) \mid b \in B\}$ of $N$ tensors of size 
$$
D = 2^{O(\log^2 W \log^2 T (\log \log W + \log \log T + \log\epsilon^{-1}))}
$$ 
such that for every pair of strings $(a, b) \in A \times B$, the maximum probability (over all Merlin's actions) that Alice accepts in the \textsf{IP}-protocol equals the Tropical Similarity score of the corresponding tensor gadgets $G(a)$ and $H(b)$. Thus, to decide whether there exists a pair of $(a, b) \in A \times B$ that can make $P$ accept, it is sufficient to distinguish from the case that there is a pair of $G(a), H(b)$ such that the Tropical Similarity score $s(G(a), H(b)) \le \epsilon$ and the case that every pair of $G(a), H(b)$ has perfect Tropical Similarity score $s(G(a), H(b)) = 1$.
\end{proof}

Note that in the above constructions in Theorem~\ref{thm:hard-gap-tt} and \ref{thm:hard-gap-min-tt}, tensors in $B$ are all $\max$-invariant, so we have the following corollary for \OAPTrestri:
\begin{cor} \label{thm:hard-maxmintt-restricted}
There is a reduction from \BPSATPAIR on branching program of length $T$ and width $W$ and two sets of $N$ strings to a \GAPMAXTTrestri{$\epsilon$}~/ \GAPMINTTrestri{$\epsilon$} on two sets of $N$ tensors of size 
$$
D = 2^{O(\log^2 W \log^2 T (\log \log W + \log \log T + \log \epsilon^{-1}))},
$$
and the reduction runs in $O(N \poly(D))$. Here $\epsilon$ is a threshold value that can depend on $N$.
\end{cor}

Theorem \ref{thm:hard-gap-tt} and \ref{thm:hard-gap-min-tt} also imply reductions from \BPSATPAIR to exact \MAXTT or exact \MINTT. For the other direction of reduction, we have the following lemma:

\begin{lemma}\label{lam:tropical-polylog}
Given two tensors $a, b$ of size $n = 2^t$, their Tropical Similarity $s(a, b)$ can be computed in $\SPACE[O(\log^2 n)]$.
\end{lemma}
\begin{proof}
We compute the Tropcial Similarity recursively according to the definition. Note that there are $t$ levels of recursion in total, and $O(t)$-bit precision is sufficient in this computation. Thus this algorithm uses only $O(t^2) \le O(\log^2 n)$ space.
\end{proof}

Combining Theorem \ref{thm:hard-gap-tt}, Theorem \ref{thm:hard-gap-min-tt} and Lemma \ref{lam:tropical-polylog}, we can establish the equivalence between \BPSATPAIR and the exact or approximate Tropical Similarity problems:
\begin{theorem} \label{thm:tt-equiv}
For the case that the maximum element size $D = 2^{(\log N)^{o(1)}}$, there are near-linear time reductions between all pairs of the following problems:
\begin{itemize}
\item \BPSATPAIR; 
\item Exact \MAXTT or \MINTT; 
\item \GAPMAXTT{$2^{-(\log D)^{1 - \Omega(1)}}$} or \GAPMINTT{$2^{-(\log D)^{1 - \Omega(1)}}$};
\item $2^{(\log D)^{1 - \Omega(1)}}$-approximate \MAXTT or \MINTT; 
\end{itemize}
\end{theorem}
\begin{proof}

Let $c > 0$ be a constant. For any instance of \BPSATPAIR on BP of size $S$, by Theorem \ref{thm:hard-gap-tt} with parameter $\epsilon = 2^{-(\log S)^{c}}$, we can reduce it to an \GAPMAXTT{$\epsilon$} instance on tensors of size $D = 2^{\Theta(\log^4 S \, \log \epsilon^{-1})} = 2^{\Theta(\log^{4+c} S)}$ (adding dummy dimensions to tensors if necessary). 

Thus, we have $\epsilon = 2^{-\Theta(\log D)^{c/(4+c)}}$. For any $0 < \delta \le 1$, by choosing an appropriate value for $c$, we can obtain a reduction from \BPSATPAIR on BP of size $S = 2^{(\log N)^{o(1)}}$ to \GAPMAXTT{$2^{-(\log D)^{1 - \delta}}$}.

\GAPMAXTT{$2^{-(\log D)^{1 - \delta}}$} can be trivially reduced to $2^{(\log D)^{1 - \delta}}$-approximate \MAXTT, and $2^{(\log D)^{1 - \delta}}$-approximate \MAXTT can be trivially reduced to \MAXTT. By Lemma \ref{lam:tropical-polylog} and Theorem \ref{thm:polylogspace-min-max}, \MAXTT can be reduced to \BPSATPAIR. 

Therefore under near-linear time reductions \BPSATPAIR, exact \MAXTT, \GAPMAXTT{$2^{-(\log D)^{1 - \Omega(1)}}$} and $2^{(\log D)^{1 - \Omega(1)}}$-approximate \MAXTT are all equivalent. Using a similar argument, we can also prove the same result for \MINTT.
\end{proof}
\section{Longest Common Subsequence}\label{sec:LCS}

In this section, we show near-liear time reductions between \BPSATPAIR and (exact or approximate) \MAXLCSPAIR~/ \MINLCSPAIR.

Our reduction from \BPSATPAIR to \MAXLCSPAIR~/~\MINLCSPAIR relies on the following \LCS gadgets in \cite{AbboudR18}.

\begin{theorem}[\cite{AbboudR18}] \label{thm:lcs-gadget}
Let $t$ be an even number and $d_1 = \cdots = d_t = 2$. Given two sets of $N$ tensors $A, B$ in $\{0, 1\}^{d_1 \times \cdots \times d_t}$, there is a deterministic algorithm running in $O(N \poly(2^t))$ time which outputs two sets $A', B'$ of $N$ strings of length $2^t$ over an $O(2^t)$-size alphabet $\Sigma$, such that each $a \in A$ corresponds to a string $a' \in A'$, each $b \in B$ corresponds to a string $b \in B'$, and $\textsf{LCS}(a', b') = 2^{t/2} \cdot s(a, b)$ for every $a \in A, b \in B$, where $s(a, b)$ stands for the Tropical Similarity score of two tensors $a$ and $b$.
\end{theorem}

Theorem \ref{thm:lcs-gadget} directly leads to the following two theorems:
\begin{theorem} \label{thm:hard-gap-max-lcs-pair}
There exists an $O(N\poly(D))$-time reduction from an \GAPMAXTT{$\epsilon(D)$} instance with two sets of $N$ tensors of size $D$ to an instance of the following approximation variant of \MAXLCSPAIR: Given two sets of $N$ strings of length $D$, distinguish between the following:
\begin{itemize}
\item \textbf{Completeness:} There exists a pair of $a, b$ such that $\textsf{LCS}(a, b) = \sqrt{D}$;
\item \textbf{Soundness:} For every pair $a, b$, $\textsf{LCS}(a, b) < \sqrt{D} \cdot \epsilon(D)$.
\end{itemize}
Thus \GAPMAXTT{$\epsilon(D)$} can be $O(N \poly(D))$-time reduced to $\epsilon(D)^{-1}$-approximate \MAXLCSPAIR.
\end{theorem}
\begin{theorem} \label{thm:hard-gap-min-lcs-pair}
There exists an $O(N\poly(D))$-time reduction from an \GAPMINTT{$\epsilon(D)$} instance with two sets of $N$ tensors of size $D$ to an instance of the following approximation variant of \MINLCSPAIR: Given two sets of $N$ strings of length $D$, distinguish between the following:
\begin{itemize}
\item \textbf{Completeness:} There exists a pair of $a, b$ such that $\textsf{LCS}(a, b) < \sqrt{D} \cdot \epsilon(D)$;
\item \textbf{Soundness:} For every pair $a, b$, $\textsf{LCS}(a, b) = \sqrt{D}$.
\end{itemize}
Thus \GAPMINTT{$\epsilon(D)$} can be $O(N \poly(D))$-time reduced to $\epsilon(D)^{-1}$-approximate \MINLCSPAIR.
\end{theorem}

According to Theorem \ref{thm:polylogspace-min-max}, in order to reduce \MAXLCSPAIR or \MINLCSPAIR problem to \BPSATPAIR, we only need to show that given a number $k$, deciding $\textsf{LCS}(a, b) \ge k$ for two strings $a, b$ of length $n$ is in $\textsf{NSPACE}[\polylog (n)]$:
\begin{lemma}[Folklore] \label{thm:lcs-nl}
Given two strings $a, b$ of length $n$ and a number $k$, deciding whether $\textsf{LCS}(a, b) \ge k$ is in $\NL$.
\end{lemma}
\begin{proof}
The algorithm consists of $k$ stages. Let $c$ be a longest common subsequence of $a$ and $b$. In the $i$-th stage, we nondeterministically guess which positions of $a$ and $b$ are matched with the $i$-th character of $c$, and then we check if the characters on the two positions of $a$ and $b$ are the same. Also, in each stage we store the positions being matched in the last stage, so that we can check if the matched positions in each string are increasing. Finally, we accept if the nondeterministic guesses pass all the checks. The total space for this nondeterministic algorithm is  $O(\log n)$ since we only need to maintain $O(1)$ positions in each stage.
\end{proof}

Combining Theorem \ref{thm:hard-gap-max-lcs-pair}, \ref{thm:hard-gap-min-lcs-pair} and Lemma \ref{thm:lcs-nl} together, we can prove the equivalence between \BPSATPAIR and exact or approximate LCS pair problems:
\begin{theorem} \label{thm:lcs-pair-equiv}
For the case that the maximum element size $D = 2^{(\log N)^{o(1)}}$, there are near-linear time reductions between all pairs of the following problems:
\begin{itemize}
\item \BPSATPAIR;
\item Exact \MAXLCSPAIR or \MINLCSPAIR;
\item $2^{(\log D)^{1 - \Omega(1)}}$-approximate \MAXLCSPAIR or \MINLCSPAIR.
\end{itemize}
\end{theorem}
\begin{proof}
By Theorem \ref{thm:tt-equiv}, the \BPSATPAIR problem on branching program of size $2^{(\log N)^{o(1)}}$ is equivalent to \GAPMAXTT{$2^{(\log D)^{1 - \delta}}$} under near-linear time reductions. By Theorem \ref{thm:hard-gap-max-lcs-pair}, \GAPMAXTT{$2^{(\log D)^{1 - \delta}}$} can be reduced to $2^{(\log D)^{1 - \delta}}$-approximate \MAXLCSPAIR.
$2^{(\log D)^{1 - \delta}}$-approximate \MAXLCSPAIR can be trivially reduced to \MAXLCSPAIR. By Lemma \ref{thm:lcs-nl} and Theorem \ref{thm:polylogspace-min-max}, \MAXLCSPAIR can further be reduced to \BPSATPAIR. Thus \BPSATPAIR, exact \MAXLCSPAIR and $2^{(\log D)^{1 - \delta}}$-approximate \MAXLCSPAIR are equivalent under near-linear time reductions for all $\delta > 0$.

Using a similar argument, we can also prove the same result for exact and approximate \MINLCSPAIR.
\end{proof}

\section{Regular Expression Membership Testing}\label{sec:RegularExp}

In this section, we study the hardness of regular expression problems.
First we prove that \BPSATPAIR, \REGEXPSTRPAIR and \CLOSESTREGEXPSTRPAIR are equivalent under near-linear time reductions, then we show the hardness for the Regular Expression Membership Testing problem.

For simplicity, we denote $\max_{x \in L(a), \lvert x \rvert = \lvert b \rvert} \hamsim(x, b)$ by $\maxsim(a, b)$ for any regular expression $a$ and string $b$. The following theorem gives a construction to implement the Tropical Similarity using $\maxsim$.

\begin{theorem} \label{thm:re-str-mem-gadget}
Let $t$ be an even number and $d_1 = \cdots = d_t = 2$. Given two sets of $N$ tensors $A, B \subseteq \{0, 1\}^{d_1 \times \cdots \times d_t}$ satisfying that all the tensors in $B$ are $\max$-invariant, there is a deterministic algorithm running in $O(N \poly(2^{t}))$ time which outputs a set $A'$ of $N$ regular expressions and a set $B'$ of $N$ strings. Here strings are of length $2^{t}$, regular expressions are of length $\poly(2^t)$, and both of them are over alphabet $\Sigma = \{0, 1, \bot\}$. Each $a \in A$ corresponds to a regular expression $a' \in A'$, each $b \in B$ corresponds to a string $b' \in B'$, and 
\[\maxsim(a', b') = s(a, b).\]
\end{theorem}
\begin{proof}
For each $k$ and each prefix of index $\ik \in  [d_1] \times \cdots \times [d_k]$ , we construct corresponding gadget $a' = G_{\ik}(a)$ and $b' = H_{\ik}(b)$ for each $a \in A$ and $b \in B$ (when $k = 0$, $\ik$ can only be the empty prefix, and we simply use $G(a)$ and $H(b)$ for convenience) inductively, mimicking the evaluation of the Tropical Similarity. For this purpose, we need to construct the following three types of gadgets.

\paragraph{Bit Gadgets.} First we need bit gadgets to simulate the innermost coordinatewise product in the evaluation of Tropical Similarity. For each coordinate $i \in [d_1] \times \cdots \times [d_t]$, for every $a \in A$ and $b \in B$, we construct
\[
G_i(a) = \begin{cases} \bot    &\quad \text{if $a_i = 0$,} \\
					   1             &\quad \text{if $a_i = 1$,}
		 \end{cases}
\quad\text{and}\quad
H_i(b) = \begin{cases} 0 &\quad \text{if $b_i = 0$,} \\
                       1 &\quad \text{if $b_i = 1$.}
         \end{cases}
\]
It is easy to see that $a_i \cdot b_i = \maxsim(G_i(a), H_i(b))$.

Now we combine bit gadgets recursively according to the $\max$ and $\E$ operators in the evaluation for Tropical Similarity. Starting from $k = t - 1$, there are two cases to consider.

\paragraph{Expectation Gadgets.} The first case is when $\E$ operator is applied to $(k+1)$-th dimension. We construct the corresponding gadgets $G_{\ik}(a)$ and $G_{\ik}(b)$ for any $\ik \in [d_1] \times \cdots \times [d_k]$, $a \in A$ and $b \in B$ as follows:
\[
G_{\ik}(a) = G_{\ik, 0}(a) \circ G_{\ik, 1}(a) \quad \text{and} \quad 
H_{\ik}(b) = H_{\ik, 0}(b) \circ H_{\ik, 1}(b).
\]
where $\circ$ stands for concatenation as usual. 
It is easy to see that \[
\maxsim(G_{\ik}(a), H_{\ik}(b)) = \E_{j \in \{0, 1\}}\left[ \maxsim(G_{\ik, j}(a), H_{\ik, j}(b)) \right].
\]

\paragraph{Max Gadgets.} The second case is when $\max$ operator is applied to $(k+1)$-th dimension.  For $\ik \in [d_1] \times \cdots \times [d_k]$ and $a \in A$, $G_{\ik}(a)$ is constructed as follows:
\[
G_{\ik}(a) = \left[~G_{\ik,0}(a) ~\middle|~  G_{\ik,1}(a)~\right],
\]
and we construct $H_{\ik}(b)$ for $b \in B$ to be
\[
H_{\ik}(b) = H_{\ik,j}(b)
\]
for all $j \in \{0, 1\}$, which is well-defined since $b$ is $\max$-invariant. It is easy to see that \[
\maxsim(G_{\ik}(a), H_{\ik}(b)) = \max_{j \in \{0, 1\}}\left[ \maxsim(G_{\ik, j}(a), H_{\ik, j}(b)) \right].
\]
Finally, we can obtain tensor gadgets $G(a), H(b)$ for each $a \in A$ and $b \in B$.
\end{proof}

From Theorem \ref{thm:re-str-mem-gadget}, we have the following reduction:
\begin{cor}\label{cor:MAXTT-to-CRSP}
Let \GAPCLOSESTREGEXPSTRPAIR{$\epsilon$} be the approximation variant of \CLOSESTREGEXPSTRPAIR: Given a set of $N$ regular expressions of length $O(\poly(D))$ and a set of $N$ strings of length $D$, distinguish between the following:
\begin{itemize}
\item \textbf{Completeness:} There exists a pair of $a, b$ such that $\maxsim(a, b) = 1$ (i.e., $b \in L(a)$);
\item \textbf{Soundness:} For every pair $a, b$, $\maxsim(a, b) < \epsilon$.
\end{itemize}
There exists an $O(N\poly(D))$-time reduction from a \GAPMAXTTrestri{$\epsilon(D)$} instance on two sets $A, B$ of $N$ tensors of size $D$ to an instance of \GAPCLOSESTREGEXPSTRPAIR{$\epsilon(D)$} on a set of $N$ regular expressions of length $O(\poly(D))$ and a set of $N$ strings of length $D$.
\end{cor}
This corollary also follows that there is a reduction from \GAPMAXTT{$\epsilon$} to \REGEXPSTRPAIR, which is enough to show that \REGEXPSTRPAIR is no easier than \BPSATPAIR. But actually it is possible to show a direct reduction from \OAPTrestri to \REGEXPSTRPAIR, without using the reduction from \OAPTrestri to \GAPMAXTTrestri{$\epsilon$}:
\begin{theorem} \label{thm:re-str-mem-gadget-oapt}
There exists an $O(N\poly(D))$-time reduction from a \OAPTrestri instance with two sets $A, B$ of $N$ tensors of size $D$ to an \REGEXPSTRPAIR instance with a set of $N$ regular expressions of length $O(\poly(D))$ and a set of $N$ strings of length $D$.
\end{theorem}
\begin{proof}
We use nearly the same reduction as in Theorem \ref{thm:re-str-mem-gadget}. The bit gadgets are constructed as follows:
\[
G_i(a) = \begin{cases} [0 \mid 1]    &\quad \text{if $a_i = 0$,} \\
					   0             &\quad \text{if $a_i = 1$,}
		 \end{cases}
\quad\text{and}\quad
H_i(b) = \begin{cases} 0 &\quad \text{if $b_i = 0$,} \\
                       1 &\quad \text{if $b_i = 1$.}
         \end{cases}
\]
for any $a \in A$, $b \in B$. And we construct $\land$ gadgets in the same way as the $\max$ gadgets, $\lor$ gadgets in the same way as the $\E$~gadgets. By De Morgan's laws, we can show that $H(b) \in L(G(a))$ iff $\aprod(a, b) = 0$.
\end{proof}

For the other direction, we note that the following theorem gives a low-space algorithm for exact and approximate regular expression membership testing, then we can obtain a reduction by Theorem~\ref{thm:polylogspace-sat}. The following theorem is noted in \cite{JR91}:
\begin{theorem}[\cite{JR91}] \label{thm:regexp-polylog}
Given a regular expression $a$ and a string $b$, deciding whether $b \in L(a)$ is in $\NL$.
\end{theorem}

Combining all the above reductions together, we can show the equivalence between all pairs of \BPSATPAIR, \REGEXPSTRPAIR and \GAPCLOSESTREGEXPSTRPAIR{$\epsilon$}.
\begin{theorem} \label{thm:regexp-string-pair-equiv}
For the case that the maximum element size $D = 2^{(\log N)^{o(1)}}$, there are near-linear time reductions between all pairs of the following problems:
\begin{itemize}
\item \BPSATPAIR;
\item \REGEXPSTRPAIR;
\item \GAPCLOSESTREGEXPSTRPAIR{$2^{(\log D)^{1 - \Omega(1)}}$}.
\end{itemize}
\end{theorem}
\begin{proof}
By Theorem \ref{thm:re-str-mem-gadget} and \ref{thm:regexp-polylog}, we can show cyclic reductions between these three problems in a similar way as in the proof of Theorem \ref{thm:lcs-pair-equiv}.
\end{proof}

We can also show a reduction from \OAPTrestri to Regular Expression Membership Testing on two strings using the same gadgets in Theorem \ref{thm:re-str-mem-gadget-oapt}.
\begin{theorem} \label{thm:regexp-string-hard}
There exists an $O(N \poly(D))$-time reduction from a \OAPTrestri instance with two sets $A, B$ of $N$ tensors of size $D$ to an instance of Regular Expression Membership Testing on a regular expression $R$ of length $O(N \poly(D))$ and a string $S$ of length $O(N \poly(D))$.
\end{theorem}
\begin{proof}
We construct the two sets $A', B'$ as in Theorem \ref{thm:re-str-mem-gadget}. For the construction for the regular expression, let $w$ be the concatenation of all $a' \in A$ separated by ``~$|$~''. Then we construct the regular expression $R$ to be
\[
R = \left[\bigcirc_{i=1}^{D} [0 \mid 1] \right]^{*} \circ [w] \circ \left[\bigcirc_{i=1}^{D} [0 \mid 1] \right]^{*}
\]
and we construct the string $S$ by concatenating all  $b' \in B$ directly. It is easy to see that there exists a pair $(a, b) \in A \times B$ with $\aprod(a, b) = 0$ iff $S \in L(R)$, by noticing that all the strings $x \in L(w), b' \in B$ are of the same length $D$.
\end{proof}
\section{Subtree Isomorphism and Largest Common Subtree} \label{sec:subtree-Iso}

In this section, we study the hardness of Subtree Isomorphism and Largest Common Subtree. Our reductions here are inspired by \cite{ABHVZ16,AbboudR18}. We begin with some notations to ease our construction of trees. Recall that all trees considered in this paper are bounded-degree and unordered. We are interested in both rooted and unrooted trees.
Here ``rooted'' means that the root of $G$ must be mapped to the root of $H$ in the isomorphism.

We use $\mathcal{T}_2$ to denote the tree with exactly two nodes. Let $\mathcal{T}^{0}_3$ be the $3$-node tree with root degree $1$, and let $\mathcal{T}^{1}_3$ be the $3$-node tree with root degree $2$. For a tree $T$, let $\mathcal{P}^k(T)$ be the tree constructed by joining a path of $k$ nodes and the tree $T$: one end of the path is regarded as the root, the other end of the path is linked to the root of $T$ by an edge. For two trees $\mathcal{T}_a$ and $\mathcal{T}_b$, we use $(\mathcal{T}_a \circ \mathcal{T}_b)$ to denote the tree whose root has two children $\mathcal{T}_a$ and $\mathcal{T}_b$.

\subsection{Subtree Isomorphism}

In this subsection, first we prove that $\BPSATPAIR$ and $\STIPAIR$ are equivalent under near-linear time reductions, then we show the hardness for Subtree Isomorphism on two trees.

For two trees $T_a, T_b$, we use $\Sub(T_a, T_b)$ to indicate whether $T_a$ is isomorphic to a subtree of $T_b$ when $T_a, T_b$ are seen as unrooted trees. Also, we use $\RSub(a, b)$ to indicate whether $T_a$ is isomorphic to a subtree of $T_b$ when $T_a, T_b$ are seen as rooted trees.

\begin{theorem} \label{thm:hard-sti-pair}
Let $t$ be an even number and $d_1 = \cdots = d_t = 2$. Given two sets of $N$ tensors $A, B$ in $\{0, 1\}^{d_1 \times \cdots \times d_t}$ satisfying that all the tensors in $A$ are $\land$-invariant, there is a deterministic algorithm running in $O(N \poly(2^t))$ time which outputs two sets $A', B'$ of $N$ binary trees of size $O(2^t)$ and depth $O(t)$, such that each $a \in A$ corresponds to a tree $a' \in A'$, each $b \in B$ corresponds to a tree $b \in B'$, and
\[\overline{\aprod(a, b)} = \RSub(a',b') = \Sub(a', b'),\]
where $\overline{\aprod(a, b)}$ is the negation of the Alternating Product of $a$ and $b$.
\end{theorem}

\begin{proof} 
For each $k$ and each prefix of index $\ik \in  [d_1] \times \cdots \times [d_k]$ , we construct corresponding tree gadgets $G_{\ik}(a)$ and $H_{\ik}(b)$ for each $a \in A$ and $b \in B$ (when $k = 0$, $\ik$ can only be the empty prefix, and we simply use $G(a)$ and $H(b)$ for convenience) inductively, mimicking the evaluation of the alternating product. 
Our gadgets satisfy that
\[
\RSub(G_{\ik}(a),H_{\ik}(b)) = \overline{\aprod(a_{\ik}, b_{\ik})}
\]
for any subtensors $a_{\ik}, b_{\ik}$.
For this purpose, we need to construct the following three types of tree gadgets.

\paragraph{Bit Gadgets.} First we need bit gadgets to simulate the innermost coordinatewise product in the alternating product. For each coordinate $i \in [d_1] \times \cdots \times [d_t]$, for every $a \in A$ and $b \in B$, we construct 

\[
G_i(a) = \begin{cases} \mathcal{T}_2   &\quad \text{if $a_i = 0$,} \\
					   \mathcal{T}^1_3 &\quad \text{if $a_i = 1$,}
		 \end{cases}
\quad\text{and}\quad
H_i(b) = \begin{cases} \mathcal{T}^1_3 &\quad \text{if $b_i = 0$,} \\
                       \mathcal{T}_2 &\quad \text{if $b_i = 1$.}
         \end{cases}
\]
It is easy to see that $\RSub(G_i(a),H_i(b))$ iff $a_i \land b_i = 0$.

Now we combine bit gadgets recursively according to the $\land$ and $\lor$ operators in the Alternating Product.
Starting from $k = t - 1$, there are two cases to consider.

\paragraph{AND Gadgets.} The first case is when $\land$ operator is applied to $(k+1)$-th dimension, then by De~Morgan's laws, we need to construct our gadgets such that for all $\ik \in [d_1] \times \cdots \times [d_k]$
\[
\RSub(G_{\ik}(a), H_{\ik}(b)) = \RSub(G_{\ik,0}(a), H_{\ik,0}(b)) \lor \RSub(G_{\ik,1}(a), H_{\ik,1}(b)).
\]

To do so, for $a \in A$, we construct $G_{\ik}(a)$ to be
\[
	G_{\ik}(a) = \mathcal{P}^{1}(G_{\ik, 0}(a)) = \mathcal{P}^{1}(G_{\ik, 1}(a)),
\]
which is well-defined since $a$ is $\land$-invariant.
And we construct $H_{\ik}(b)$ to be
\[
	H_{\ik}(b) = \left(H_{\ik, 0}(b) \circ H_{\ik, 1}(b)\right).
\]
In any subtree isomorphism, it is easy to see that $G_{\ik, 0}(a)$ (or $G_{\ik, 1}(a)$) can only be mapped to either $H_{\ik, 0}(b)$ or $H_{\ik, 1}(b)$, so $G_{\ik}(a), H_{\ik}(b)$ implement an $\land$ operator.

\paragraph{OR Gadgets.} The second case is when $\lor$ operator is applied to $(k+1)$-th dimension, then by De~Morgan's laws, we need to construct our gadgets such that for all $\ik \in [d_1] \times \cdots \times [d_k]$,
\[
\RSub(G_{\ik}(a), H_{\ik}(b)) = \RSub(G_{\ik,0}(a), H_{\ik,0}(b)) \land \RSub(G_{\ik,1}(a), H_{\ik,1}(b)).
\]

First for any tree $T$, we define two auxiliary trees $\mathcal{U}_0(T), \mathcal{U}_1(T)$  to ease our construction:
\[
\mathcal{U}_0(T) = \left(\mathcal{P}^3(T) \circ \mathcal{T}^0_3 \right)
\quad\text{and}\quad
\mathcal{U}_1(T) = \left(\mathcal{P}^3(T) \circ \mathcal{T}^1_3 \right).
\]
It is easy to verify that for any two trees $T_1, T_2$,  $\RSub(\mathcal{U}_0(T_1),\mathcal{U}_1(T_2)) = \RSub(\mathcal{U}_1(T_1),\mathcal{U}_0(T_2)) = 0$ and $\RSub(\mathcal{U}_0(T_1),\mathcal{U}_0(T_2)) = \RSub(\mathcal{U}_1(T_1),\mathcal{U}_1(T_2)) = \RSub(T_1, T_2)$.

We construct the corresponding tensor gadgets $G_{\ik}(a)$ and $H_{\ik}(b)$ for $a \in A$ and $b \in B$ as follows:
\[
G_{\ik}(a) = \left( \mathcal{U}_0(G_{\ik, 0}(a)) \circ \mathcal{U}_1(G_{\ik, 1}(a)) \right),
\]
and
\[
H_{\ik}(b) = \left( \mathcal{U}_0(H_{\ik, 0}(b)) \circ \mathcal{U}_1(H_{\ik, 1}(b))  \right).
\]
In any subtree isomorphism, it is easy to see that $\mathcal{U}_0(G_{\ik, 0}(a))$ can only be mapped to $\mathcal{U}_0(H_{\ik, 0}(b))$, and $\mathcal{U}_1(G_{\ik, 1}(a))$ can only be mapped to $\mathcal{U}_1(H_{\ik, 1}(b))$, so $G_{\ik}(a), H_{\ik}(b)$ implement an $\lor$ operator. 

\paragraph{Correctness.} It is not hard to verify that $\RSub(G(a), H(b)) = \overline{\aprod(a, b)}$ by De~Morgan's laws. To show $\RSub(G(a), H(b)) = \Sub(G(a), H(b))$, we focus on the case that $t > 0$ since the case that $t = 0$ is obvious. Let the root of $G_A$ be $r$ and the height of $G_A$ be $h$. The outermost operator in an Alternating Product is $\land$, so $r$ has only one child which has two subtrees of equal height $h - 2$. It is easy to see that the height of $H_B$ is also $h$. Suppose that $r$ is mapped to a node $r'$ in $H_B$ and $c$ is mapped to  $c'$. If we regard $c'$ as the root of $H_B$, then after deleting $c'$, $H_B$ should be split into two subtrees of height $\ge h - 2$ and a single node $r'$. The only possible case is that $c'$ is of depth $1$ w.r.t. the original root of $H_B$ (the depth of a root is $0$) and $r'$ is the original root of $H_B$.
\end{proof}

\begin{theorem} \label{thm:sti-polylog}
Given two bounded-degree unrooted trees $T_A$ and $T_B$, it can be decided in $\NSPACE[(\log n)^2]$ that whether $T_A$ is isomorphic to a subtree of $T_B$.
\end{theorem}
\begin{proof}
This algorithm works by divide and conquer on trees. At each recursion, we have two trees $S_A$ and $S_B$ (implicit representation) as well as a set of node pairs $M = \{(a_1, b_1), \dots, (a_k, b_k)\}$ (initially, $S_A = T_A, S_B = T_B$ and $M = \varnothing$). We need to decide whether there is an isomorphism from $S_A$ to some subtree of $S_B$ satisfying $a_i$ in $S_A$ is mapped to $b_i$ in $S_B$ for all $1 \le i \le k$.

First we find a centroid $c$ of $S_A$, i.e., a node of $S_A$ that decomposes $S_A$ into subtrees of size at most $\lceil \lvert {S_A} \rvert / 2 \rceil$ when the node is deleted. Then we nondeterministically guess a node $c'$ in $S_B$ to be the node that mapped by $c$ in the isomorphism. 
If $c = a_i$ for some $i$ but $c' \neq b_i$, then we reject; otherwise, we guess an injective mapping from the neighbors of $c$ in $S_A$ to the neighbors of $c'$ in $S_B$.

For each neighbor $v$ of $c$, let $v'$ be the neighbor of $c'$ mapped by $v$, $S^v_A$ be the subtree of $S_A$ containing $v$ when the edge between $v$ and $c$ is deleted, $S^{v'}_B$ be the subtree of $S_B$ containing $v'$ when the edge between $v'$ and $c'$ is deleted. We create a new set of node pairs $M' = \{(a_i, b_i) \in M \mid a_i \in S^v_A \}$. If $b_i \notin S^{v'}_B$ for some pair $(a_i, b_i) \in M'$, then we reject; otherwise, we recursively checking if there is an isomorphism from $S^v_A$ to some subtree of $S^{v'}_B$ satisfying $a_i$ is mapped to $b_i$ for all $(a_i, b_i) \in M'$ and $v$ is mapped to $v'$.

This algorithm terminates when $S_A$ is a single node. There are at most $O(\log n)$ levels of recursion by the property of centroid.
At each level, we use only $O(\log n)$ space for $c, c'$ and their neighbors (note that $S_A, S_B$ are bounded-degree trees), and $S_A, S_B$ can always be accessed according to the information stored at the upper levels of recursion. Thus this algorithm runs in $\NSPACE[(\log n)^2]$.
\end{proof}

Combining the above reductions together, we can show the equivalence between \BPSATPAIR and \STIPAIR.
\begin{theorem}\label{thm:bppair-stipair-equiv}
\BPSATPAIR on branching program of size $2^{(\log N)^{o(1)}}$ and \STIPAIR on (rooted or unrooted) trees of size $2^{(\log N)^{o(1)}}$ are equivalent under near-linear time reductions.
\end{theorem}

We can also show a reduction from \OAPT to Subtree Isomorphism on two trees using the same gadgets in Theorem \ref{thm:hard-sti-pair}.

\begin{theorem} \label{thm:sti-hard}
There exists an $O(N \poly(D))$-time reduction from 
a \OAPTrestri instance with two sets $A, B$ of $N$ tensors of size $D$
to an instance of Subtree Isomorphism on two (rooted or unrooted) binary trees of size $O(N \poly(D))$ and depth $2 \log N + O(\log D)$.
\end{theorem}

\begin{proof} 
Using the recursive construction in Theorem \ref{thm:hard-sti-pair} we can obtain tensor gadgets $G(a), H(b)$
for each $a \in A$ and $b \in B$, such that $\overline{\aprod(a,b)} = \RSub(G(a),H(b)) = \Sub(G(a),H(b))$.

We can assume the set size $N$ is a power of $2$ by adding dummy vectors into each set. Now we combine the tensor gadgets in each set respectively to construct two trees $G_A, H_B$ as our instance for Subtree Isomorphism:
\begin{enumerate}[a)]
\item To construct $G_A$ for set $A$:
\begin{itemize}
\item Initialize $G_A$ by a complete binary tree of $N$ leaves;
\item Associate each leaf with a tensor $a \in A$;
\item For all $a \in A$, construct $\mathcal{P}^{\log N}(G(a))$ and link an edge from its root to the corresponding leaf of $a$.
\end{itemize}
\item To construct $H_B$ for set $B$:
\begin{itemize}
\item Initialize $H_B$ by a complete binary tree of $N$ leaves;
\item Select one leaf node $v_{\ell}$;
\item For every unselected leaf, construct $\mathcal{P}^{\log N}(H(\mathbf{0}))$ and link an edge from its root to the leaf.
\item Construct a complete binary tree of $N$ leaves rooted at $v_{\ell}$;
\item Associate each leaf of the tree rooted at $v_{\ell}$ with a tensor $b \in B$;
\item For all $b \in B$, construct $\mathcal{P}^{\log N}(H(b))$ and link an edge from its root to the corresponding leaf of $b$.
\end{itemize}
\end{enumerate}
\paragraph{Correctness} For any subtree isomorphism, one $G(a)$ can be mapped to any $H(b)$ or $H(\mathbf{0})$. Since there are only $N - 1$ gadgets of $H(\mathbf{0})$, there must be some $G(a)$ mapped to some $H(b)$. Thus $\RSub(G_A, H_B)$ iff there exists a pair of $(a, b) \in A \times B$ with $\aprod(a, b) = 0$. It is not hard to see that the root of $G_A$ can only be mapped to the root of $H_B$ by arguing about the tree height (similar as Theorem \ref{thm:hard-sti-pair}), so $\RSub(G_A, H_B) = \Sub(G_A, H_B)$.
\end{proof}

\subsection{Largest Common Subtree}

In this subsection, first we prove that under near-linear time reductions between $\BPSATPAIR$ and (exact or approximate) $\MAXLCSTPAIR$ / $\MINLCSTPAIR$ are equivalent, then we show the hardness for Largest Common Subtree on two trees.

For two trees $a, b$, define $\LCST(a, b)$ to be the size of the largest common subtree of $a$ and $b$ when $a, b$ are seen as unrooted trees. Also, we define $\RLCST(a, b)$ to be the size of the largest common subtree of $a$ and $b$ when $a, b$ are seen as rooted trees.

Now we establish a connection between \MAXTTrestri and \MAXLCSTPAIR:

\begin{theorem} \label{thm:hard-approx-lcst-pair}
Let $t$ be an even number and $d_1 = \cdots = d_t = 2$. Given two sets of $N$ tensors $A, B$ in $\{0, 1\}^{d_1 \times \cdots \times d_t}$ satisfying that all the tensors in $A$ are $\max$-invariant, for any $L \ge 2^{t}$, there is a deterministic algorithm running in $O(N\cdot \poly(2^{t}) \cdot L)$ time which outputs two sets $A', B'$ of $N$ binary trees of size $O(\poly(2^{t}) \cdot L )$ and depth $O( (2^{t/2}+ \log L) \cdot t)$, such that each $a \in A$ corresponds to a tree $a' \in A'$, each $b \in B$ corresponds to a tree $b \in B'$, and
\begin{align*}
\RLCST(a', b') &= (2^{t/2} s(a, b) + O(1)) L \\
\LCST(a', b')  &= (2^{t/2} s(a, b) + O(1)) L
\end{align*}
where $s(a, b)$ is the Tropical Similarity score of $a$ and $b$. In particular, if $s(a, b) = 1$, then $a', b'$ satisfy $\RLCST(a', b') = \LCST(a', b') = \lvert a' \rvert$.
\end{theorem}
\begin{proof}
For each $k$ and each prefix of index $\ik \in  [d_1] \times \cdots \times [d_k]$ , we construct corresponding gadget $a' = G_{\ik}(a)$ and $b' = H_{\ik}(b)$ for each $a \in A$ and $b \in B$ (when $k = 0$, $\ik$ can only be the empty prefix, and we simply use $G(a)$ and $H(b)$ for convenience) inductively, mimicking the evaluation of the Tropical Similarity. For this purpose, we need to construct the following three types of gadgets.

\paragraph{Bit Gadgets.} For each coordinate $i \in [d_1] \times \cdots \times [d_t]$, let $\mathcal{C}_i$ be the tree constructed by join a path of length $2^{t/2}$ and a complete binary tree of $L$ nodes: one end of the path is regarded as the root, and we link an edge between the node of depth $\binodd(i)$ and 
the root of the complete binary tree, where $\binodd(i) \in [0, 2^{t/2})$ is the number whose binary representation is $i_1i_3 \cdots i_{t-1}$.

For every $a \in A$ and for each coordinate $i \in [d_1] \times \cdots \times [d_t]$, we construct $G_i(a) = \mathcal{C}_i$ if $a_i = 1$, or simply a path of length $2^{t/2}$ if $a_i = 0$. Similarly, for every $b \in B$ and for each coordinate $i \in [d_1] \times \cdots \times [d_t]$, we construct $H_i(b) = \mathcal{C}_i$ if $b_i = 1$, or simply a path of length $2^{t/2}$ if $a_i = 0$.

If $a_i \cdot b_i = 0$, then $\textsf{RLCST}(G_i(a), H_i(b)) = 2^{t/2}$; otherwise $\RLCST(G_i(a), H_i(b)) = 2^{t/2} + L$. Furthermore, for any two coordinates $i, j$ with $\binodd(i) \ne \binodd(j)$, $\textsf{RLCST}(G_i(a), H_j(b)) = 2^{t/2}$.

Let $K = 2^{t/2} + \lceil \log L \rceil + 1$ be the maximum possible height of a bit gadget. Now we combine bit gadgets recursively according to the $\E$ and $\max$ operators in the evaluation of Tropical Similarity score. Starting from $k = t - 1$, there are two cases to consider.

\paragraph{Expectation Gadgets.} The first case is when $\E$ operator is applied to $(k+1)$-th dimension. For $\ik \in [d_1] \times \cdots \times [d_k]$ and $a \in A$, we construct $G_{\ik}(a), H_{\ik}(b)$ as follows:
\[
G_{\ik}(a) = \left( \mathcal{P}^{K-1}(G_{\ik,0}(a)) \circ \mathcal{P}^{K-1}(G_{\ik,1}(a)) \right)
\]
and
\[
H_{\ik}(a) = \left( \mathcal{P}^{K-1}(H_{\ik,0}(b)) \circ \mathcal{P}^{K-1}(H_{\ik,1}(b)) \right)
\]

\paragraph{Max Gadgets.} The second case is when $\max$ operator is applied to $(k+1)$-th dimension.  For $\ik \in [d_1] \times \cdots \times [d_k]$ and $a \in A$, we construct $G_{\ik}(a), H_{\ik}(b)$ as follows:
\[
G_{\ik}(a) = \mathcal{P}^{K}(G_{\ik,0}(a)) = \mathcal{P}^{K}(G_{\ik,1}(a))
\]
and
\[
H_{\ik}(a) = \left( \mathcal{P}^{K-1}(H_{\ik,0}(b)) \circ \mathcal{P}^{K-1}(H_{\ik,1}(b)) \right).
\]
Note that $G_{\ik}(a)$ is well-defined since $a$ is $\max$-invariant.

Finally we obtain tensor gadgets $G(a), H(b)$ for every $a \in A, b \in B$. It is easy to see that the depth of trees is $O(tK) \le O((2^{t/2} + \log L)\cdot t)$, and the size of trees is $O(L \cdot 2^t + K \cdot 2^t) = O(\poly(2^t) \cdot L)$.

\paragraph{Correctness for $\RLCST$.} First we show that $\RLCST(G(a), H(b)) = (2^{t/2} s(a, b) + O(1)) L$. We fix two tensors $a, b$. For every dimension $k$, let $U_k$ be a set of gadget pairs:
\[
U_k = \{(G_{\ik}(a), H_{j_{(k)}}(b)) \mid \ik, j_{(k)} \in [d_1] \times \cdots \times [d_k], i_p \ne j_p \text{ for some odd } p\}.
\]
Let $\epsilon_k$ be the maximum $\RLCST$ among the pairs in $U_k$.

For every $\ik \in [d_1] \times \cdots \times [d_t]$, let $f(\ik) = \textsf{RLCST}(G_{\ik}(a), H_{\ik}(b))$. For the last dimension $t$, it is easy to see that $f(i) = 2^{t/2} + (a_i \cdot b_i) \cdot L$ and $\epsilon_t = 2^{t/2}$. Now we prove by induction that $f(i_{(k)}) \ge \epsilon_k$ holds for every $0 \le k \le t$ and $\ik \in [d_1] \times \cdots \times [d_k]$.

On the one hand, if an $\E$ operator is applied to the $(k+1)$-th dimension, by induction hypothesis we have
$f(\ik,0) + f(\ik,1) \ge 2\epsilon_{k+1}$, so $f(\ik) = f(\ik,0) + f(\ik,1) + 2(K - 1) + 1$. Note that $\epsilon_k \le 2\epsilon_{k+1} + 2(K - 1) + 1$, so $f(i_{(k)}) \ge \epsilon_k$ holds.
On the other hand, If the $\max$ operator is applied to the $(k+1)$-th dimension, then we have $f(\ik)  = \max\{ f(\ik,0), f(\ik,1) \} + K$ and $\epsilon_k \le \epsilon_{k-1} + K$, so $f(i_{(k)}) \ge \epsilon_k$ holds as well.

Expanding the above recurrence relation of $f(\ik)$, we have
\begin{align*}
\RLCST(G(a), H(b)) &= 2^{t/2} s(a, b) L + O(K) \cdot 2^{t/2} \\
& = 2^{t/2} s(a, b) L + O(2^{t/2} + \log L) \cdot 2^{t/2} \\
& = (2^{t/2} s(a, b) + O(1)) L.
\end{align*}

\paragraph{Correctness for $\LCST$.} Now we show that $\LCST(G(a), H(b)) = \RLCST(G(a), H(b)) + O(L)$. Fix a pair of $(a, b) \in A \times B$.
If a node of $G(a)$ or $H(b)$ is in a bit gadget, then we call it \textit{bit node}. If a node of $G(a)$ or $H(b)$ is not in any bit gadget, then we call it \textit{operator node}.

Let $I_G(a)$ and $I_H(b)$ be the largest isomorphic subtrees in $G(a)$ and $H(b)$.
Let $r_a$ be the root of $I_G(a)$, i.e., the lowest node when the tree is directed with respect to the root of $G(a)$, and let $r_b$ be the root of $I_H(b)$.
Let $r'_b$ be the node in $G(a)$ that is mapped to $r_b$, and $r'_a$ be the node in $H(b)$ that is mapped from $r_a$.

\begin{figure}
\centering
\begin{tikzpicture}[
	level/.style={sibling distance=95},
	level distance = 45,scale=0.7,transform shape]

\tikzset{every node/.style={minimum width=20}}
\tikzset{itria/.style={
  draw,shape border uses incircle,
  isosceles triangle,shape border rotate=90,yshift=-50},
}

\node [circle,fill=gray!40,draw] (z){$r_a$}
	child {node [circle,fill=gray!40,draw] (a) {}
		child {node [circle,fill=white,draw] (u2) {$u_2$}}
    child {node [circle,fill=gray!40,draw] (g) {$$}
    		child {node [circle,fill=white,draw] (u3) {$u_3$}}
		child[missing]
    }
  }
  child[sibling distance=120] {
	node [circle,fill=white,draw] (u1) {$u_1$}
};

\node [circle,fill=gray!40,draw,level 1/.style={sibling distance=150}] (h) at ($(g) + (1.6, -3.5)$) {}
	child { node [circle,fill=white,draw] (uq1) {$u_{q-1}$} }
	child { node [circle,fill=gray!40,draw] (i) {$r'_b$}
		child { node [circle,fill=white,draw] (uq) {$u_{q}$} }
	}
;

\node [circle,fill=gray!40,draw] (zk) at ($(z) + (10, 0)$) {$r_b$}	
  child[sibling distance=120] {
	node [circle,fill=white,draw] (u1k) {$u'_{q}$}
  }
  child {node [circle,fill=gray!40,draw] (ak) {}
		child {node [circle,fill=white,draw] (u2k) {$u'_{q-1}$}}
    child {node [circle,fill=gray!40,draw] (gk) {$$}
		child[missing]
		child {node [circle,fill=white,draw] (u3k) {$u'_{q-2}$}}
    }
  };

\node [circle,fill=gray!40,draw,level 1/.style={sibling distance=150}] (hk) at ($(gk) + (-1.6, -3.5)$) {}
	child { node [circle,fill=white,draw] (uq1k) {$u'_{2}$} }
	child { node [circle,fill=gray!40,draw] (ik) {$r'_a$}
		child { node [circle,fill=white,draw] (uqk) {$u'_{1}$} }
	}
;

\node at ($(z)+(0,1.2)$) {$I_G(a)$};
\node at ($(zk)+(0,1.2)$) {$I_H(b)$};

\draw[loosely dotted,line cap=round,thick] (g) to (h);
\draw[loosely dotted,line cap=round,thick] (gk) to (hk);

\begin{scope} [on background layer]
\node[itria] (v1) at ($(u1)$) {$\widetilde{T}_{u_1}$};
\node[itria] (v2) at (u2) {$\widetilde{T}_{u_2}$};
\node[itria] (v3) at (u3) {$\widetilde{T}_{u_3}$};
\node[itria] (vq1) at (uq1) {$\widetilde{T}_{u_{q-1}}$};
\node[itria] (vq) at (uq) {$\widetilde{T}_{u_{q}}$};

\node[itria] (v1k) at (u1k) {$\widetilde{T}_{u'_{q}}$};
\node[itria] (v2k) at (u2k) {$\widetilde{T}_{u'_{q-1}}$};
\node[itria] (v3k) at (u3k) {$\widetilde{T}_{u'_{q-2}}$};
\node[itria] (vq1k) at (uq1k) {$\widetilde{T}_{u'_{2}}$};
\node[itria] (vqk) at (uqk) {$\widetilde{T}_{u'_{1}}$};
\end{scope}

\end{tikzpicture}
\caption{An illustration of $I_G(a)$ and $I_H(b)$.}
\end{figure}

If $r_a = r'_b$, then let $q = 1$ and $u_1 = r_a, u'_1 = r_b$. Otherwise,
let $u_1, \dots, u_q$ be the list of nodes that are in $I_G(a)$ and are adjacent to some node on the path from $r_a$ to $r'_b$. Assume $u_1, \dots, u_q$ is in depth-increasing order (it is easy to see that no two such nodes are of same depth). Let $u'_1, \dots, u'_q$ be the nodes in $I_H(b)$ that are mapped by $u_1, \dots, u_q$, respectively. Each node in $u'_1, \dots, u'_q$ should be adjacent to some node on the path from $r'_a$ to $r_b$ in $I_H(b)$.

For a node $u_i$, we denote the whole subtree of $u_i$ in $G(a)$ as $T_{u_i}$, and we define $T'_{u'_i}$ similarly. We can decompose the subtree $I_G(a)$ into two parts: the first part is the path from $r_a$ to $r'_b$, and the second part is the $q$ rooted subtrees $\widetilde{T}_{u_1} = T_{u_1} \cap I_G(a), \dots, \widetilde{T}_{u_q} = T_{u_q} \cap I_G(a)$. Similarly, we can decompose $I_H(b)$ into the path from $r_b$ to $r'_a$ and $q$ rooted subtrees $\widetilde{T}_{u'_1} = T_{u'_1} \cap I_H(b), \dots, \widetilde{T}_{u'_q} = T_{u'_q} \cap I_H(b)$. Thus we have
\[
\LCST(G(a), H(b)) = q + \sum_{i = 1}^{q} \left\lvert \widetilde{T}_{u_i} \right\rvert \le  O(tK) + \sum_{i = 1}^{q} \RLCST(T_{u_i}, T'_{u'_i})
\]

It is sufficient to obtain a bound for the sum of $\RLCST$ of $T_{u_i}, T'_{u'_i}$. 
If some $u_j$ is a bit node, then $u_i$ is also a bit node for all $i > j$, and all of them are in the same bit gadget, so $\sum_{i = j}^{q} \RLCST(T_{u_i}, T'_{u'_i}) \le O(2^{t/2} + L)$. Similarly, some $u'_j$ is a bit node, then $u'_i$ is also a bit node for all $i < j$, and all of them are in the same bit gadget, so $\sum_{i = 1}^{j} \RLCST(T_{u_i}, T'_{u'_i}) \le O(2^{t/2} + L)$.

Now we consider the following three cases when both $u_i$ and $u'_i$ are operator nodes: ($\depth(u_i)$ stands for the depth of $u_i$ in $G(a)$, $\depth(u'_i)$ stands for the depth of $u'_i$ in $H(b)$)
\begin{enumerate}[\bf C{a}se 1.]
\item $\depth(u_i) \not\equiv \depth(u'_i) \pmod{K}$, then it is impossible to map some operator node with two children in $T_{u_i}$ to an operator node with two children in $T_{u'_i}$. If $\depth(u_i) > \depth(u'_i)$, then at most one bit gadget in $T_{u_i}$ can have nodes being mapped to $T_{u'_i}$, and the case that $\depth(u_i) < \depth(u'_i)$ is similar. Thus $\RLCST(T_{u_i}, T'_{u'_i}) \le O(tK) + (2^{t/2} + L) = O(tK + L)$.

Note that the depth of parent nodes of $u_i$ and $u'_i$ should also be different modulo $K$,
so either $u_i$ has no parent in $I_G(a)$ (this is the case when $r_a = r'_b$) or the parent of $u_i$ has only one child in $I_G(a)$, and either case implies $i = q$.

\item If $\depth(u_i) \equiv \depth(u'_i) \pmod{K}$ but $\depth(u_i) > \depth(u'_i)$, then since $K$ is an upper bound of the height of any bit gadget, all the nodes in $T_{u_i}$ can only be mapped to the operator nodes in $T_{u'_i}$, which implies $\RLCST(T_{u_i}, T'_{u'_i})$ is no more than the number of operator nodes in $T'_{u'_i}$. The case that $\depth(u_i) \equiv \depth(u'_i) \pmod{K}$ but $\depth(u_i) < \depth(u'_i)$ is similar.

All the trees $T'_{u'_i}$ are disjoint. Thus the sum $\sum_{i} \RLCST(T_{u_i}, T'_{u'_i})$ over all $i$ in this case can be upper-bounded by the total number of operator nodes in $G(a)$, which is $O(K \cdot 2^{t/2})$.
\item If $\depth(u_i) = \depth(u'_i)$, then there exists two prefixes of index $i_{(k)}, j_{(k)}$, such that
\[
\RLCST(T_{u_i}, T_{u'_i}) = \RLCST(G_{i_{(k)}}(a), H_{j_{(k)}}(b)) + O(K)
\]
Note that $u_1, \dots, u_q$ are in depth-increasing order, and $u'_q, \dots, u'_1$ are in depth-decreasing order, so this case can only happen for at most one pair of nodes.
\end{enumerate}
Summing up all the above cases, we have
\[
\LCST(G(a), H(b)) \le O(tK) + O(2^{t/2} + L) + O(tK + L) + O(K \cdot 2^{t/2}) +  \RLCST(G_{i_{(k)}}(a), H_{j_{(k)}}(b))
\]
for any $i_{(k)}, j_{(k)}$. Thus $\LCST(G(a), H(b)) \le \RLCST(G(a), H(b)) + O(L)$.
\end{proof}

By Theorem \ref{thm:hard-approx-lcst-pair} with $L = 2^{t}$, we can easily have the following reductions from \GAPMAXTT{$\epsilon(D)$} and \GAPMINTT{$\epsilon(D)$} to approximation variants of \MAXLCSTPAIR and \MINLCSTPAIR. Here we focus on the case that $\epsilon(D) = \Omega(D^{-1/2})$. It is reasonable since $\sqrt{D} \cdot s(a, b)$ is always an integer, and \GAPMAXTT{$o(D^{-1/2})$} is essentially equivalent to \MAXTT. This argument also holds for \GAPMINTT{$o(D^{-1/2})$}.

\begin{theorem} \label{thm:hard-gap-max-lcst-pair}
For any function $\epsilon(D) = \Omega(D^{-1/2})$,
there exists an $O(N \poly(D))$-time reduction from a \GAPMAXTTrestri{$\epsilon(D)$} instance with two sets of $N$ tensors of size $D$ to an instance of the following approximation variant of $\MAXLCSTPAIR$: Given two sets $A, B$ of $N$ trees of size $\poly(D)$ and a set $B$ of $N$ strings of length $D$ over a constant-size alphabet, distinguish between the following:
\begin{itemize}
\item \textbf{Completeness:} There exists a pair of $a, b$ such that $\LCST(a, b) = \lvert a \rvert =  (1 + o(1))D^{3/2}$;
\item \textbf{Soundness:} For every pair $a, b$, $\LCST(a, b) \le O(\epsilon(D) D^{3/2})$.
\end{itemize}
And this conclusion also holds for $\RLCST$.
\end{theorem}

\begin{remark}
Note that Theorem~\ref{thm:hard-gap-max-lcst-pair} also implies a reduction from \BPSATPAIR to \STIPAIR, but the trees constructed by the reduction in Theorem \ref{thm:hard-sti-pair} have a smaller size and a lower depth.
\end{remark}

\begin{theorem} \label{thm:hard-gap-min-lcst-pair}
For any function $\epsilon(D) = \Omega(D^{-1/2})$,
there exists an $O(N \poly(D))$-time reduction from 
a \GAPMINTTrestri{$\epsilon(D)$} instance with two sets of $N$ tensors of size $D$
to an instance of the following approximation variant of $\MINLCSTPAIR$: Given two sets $A, B$ of $N$ trees of size $\poly(D)$ and a set $B$ of $N$ strings of length $D$ over a constant-size alphabet, distinguish between the following:
\begin{itemize}
\item \textbf{Completeness:} There exists a pair of $a, b$ such that $\LCST(a, b) \le O(\epsilon(D) D^{3/2})$;
\item \textbf{Soundness:} For every pair $a, b$, $\LCST(a, b) = \lvert a \rvert = (1 + o(1))D^{3/2}$.
\end{itemize}
And this conclusion also holds for $\RLCST$.
\end{theorem}

By Theorem \ref{thm:polylogspace-min-max}, we show reductions from \MAXLCSTPAIR (or \MINLCSTPAIR) to \BPSATPAIR via the following theorem:
\begin{theorem} \label{thm:lcst-polylog}
Given two bounded-degree unrooted trees $T_A$ and $T_B$ and a number $q$, it can be decided in $\NSPACE[(\log n)^2]$ that whether there is an isomorphism between a subtree of $T_A$ and a subtree of $T_B$ of size $q$.
\end{theorem}
\begin{proof}
The algorithm in Theorem \ref{thm:sti-polylog} suffices to fulfill the requirement if modified slightly. At each level of recursion, we have two trees $S_A$ and $S_B$, a number $q$, and a set of node pairs $M = \{(a_1, b_1), \dots, (a_k, b_k)\}$. We need to decide whether there is an isomorphism from a subtree of $S_A$ to a subtree of $S_B$ satisfying it is of size $q$ and $a_i$ in $S_A$ is mapped to $b_i$ in $S_B$ for all $1 \le i \le k$.

First we find a centroid $c$ of $S_A$, then we guess if there is a subtree of size $q$ that contains $c$ and is isomorphic to a subtree of $S_B$. If not, then we delete $c$ to decompose $S_A$ into subtrees, guess which subtree contains a subtree that is isomorphic to a subtree of $S_B$ of size $q$, and runs our algorithm to check recursively; If it is, then follow the same routine as in Theorem \ref{thm:sti-polylog}: we guess a node $c'$ in $S_B$ to be the node that mapped by $c$ in the isomorphism and a bijective mapping from some of the neighbors of $c$ in $S_A$ to some of the neighbors of $c'$ in $S_B$. Additionally, we guess a number $q_v$ for each neighbor $v$ of $c$ and ensure the sum of $q_v$ over all neighbors equals to $q - 1$. We then check recursively if there is an isomorphism from a subtree of $S^{v}_A$ to a subtree of $S^{v'}_B$ of size $q$ subject to the constraint that some set of node pairs are matched. It is clear that this algorithm runs in $\NSPACE[(\log n)^2]$.
\end{proof}

\begin{theorem} \label{thm:lcst-pair-equiv}
For the case that the maximum element size $D = 2^{(\log N)^{o(1)}}$, there are near-linear time reductions between all pairs of the following problems:
\begin{itemize}
\item \BPSATPAIR;
\item \MAXLCSTPAIR or \MINLCSTPAIR on (rooted or unrooted) trees;
\item $2^{(\log D)^{1 - \Omega(1)}}$-approximate \MAXLCSTPAIR or \MINLCSTPAIR on (rooted or unrooted) trees;
\end{itemize}
\end{theorem}
\begin{proof}
By Theorem \ref{thm:hard-gap-max-lcst-pair}, \ref{thm:hard-gap-min-lcst-pair} and \ref{thm:lcst-polylog}, we can show cyclic reductions between these three problems in a similar way as in the proof of Theorem \ref{thm:lcs-pair-equiv}.
\end{proof}

We can also show a reduction from \GAPMAXTT{$\epsilon(N)$} to Largest Common Subtree on two large trees using the same gadgets in Theorem \ref{thm:hard-approx-lcst-pair}.

\begin{theorem} \label{thm:lcst-gadget}
Let $t$ be an even number and $d_1 = \cdots = d_t = 2$. Given two sets of $N$ tensors $A, B$ in $\{0, 1\}^{d_1 \times \cdots \times d_t}$, for $L = \Theta(2^{t} \log^2 N)$, there is a deterministic algorithm running in $O(N \log^2 N 2^{O(t)})$ time which outputs two binary trees $A', B'$ of size $O(N \log^2 N 2^{O(t)})$ and depth $O(\log^2 N  \poly(t) 2^{t/2})$, such that
\begin{align*}
\RLCST(A', B') &= (2^{t/2} s_{\max} + O(1)) L \\
\LCST(A', B')  &= (2^{t/2} s_{\max} + O(1)) L
\end{align*}
where $s_{\max}$ is the maximum Tropical Similarity among all pairs of $(a, b) \in A \times B$.
\end{theorem}
\begin{proof}
Using the recursive construction in Theorem \ref{thm:hard-approx-lcst-pair}, we can obtain tensor gadgets $G(a), H(b)$ for each $a \in A$ and $b \in B$. Let $m$ be the smallest number such that $N \le 2^m$. Now we combine the tensor gadgets in each set respectively to construct two trees $G_A, H_B$ as our instance for Largest Common Subtree.

For any number $K$, we define $K$-zoomed complete binary tree $\mathcal{Z}_K$ of $2^m$ as follows: first we construct a complete binary tree of $2^m$ leaves, then we insert $K-1$ internal nodes between every pair of adjacent nodes (so $\mathcal{Z}_K$ is of height $mK + 1$).

Let $K_D = O((2^{t/2}+ \log L)t)$ be the maximum diameter of any tensor gadget ($G(a)$ or $H(b)$). Let $K_G = 2(m+1)K_D$.

We construct $G_A=\mathcal{P}^{m K_G+1}(T_A)$, where $T_A$ is the following auxiliary tree:
\begin{itemize}
\item Initialize $T_A = \mathcal{Z}_{K_D}$, and arbitrarily select $N$ leaves;
\item For each selected leaf, associate it with a tensor $a \in A$;
\item Construct $G(a)$ for every $a \in A$, and link an edge from its root to the corresponding leaf of $a$.
\end{itemize}
And the tree $H_B$ for set $B$ is constructed as follows:
\begin{itemize}
\item Initialize $H_B = \mathcal{Z}_{K_G}$ and arbitrarily select $N$ leaves;
\item For each selected leaf, associate it with a tensor $b \in B$;
\item Construct $\mathcal{P}^{mK_D+1}(b)$ for every $b \in B$, and link an edge from its root to the corresponding leaf of $b$.
\end{itemize}

\paragraph{Proof for $\RLCST$.} Note that all the tensor gadgets are of same depth, and only one gadget $G(a)$ in $G_A$ can be mapped to a gadget $H(b)$ in $H_B$.

For $L = \Theta(2^{t} \log^2 N)$, using the fact that $\RLCST(G(a), H(b)) =  (2^{t/2} s(a, b) + O(1))L$, one can easily show that
\[
\RLCST(G_A, H_B) = (2^{t/2} s_{\max} + O(1))L + O(mK_G + mK_D) = (2^{t/2} s_{\max} + O(1))L.
\]
\paragraph{Proof for $\LCST$.} Now we show that $\LCST(G_A, H_B) \le \LCST(G(a), H(b)) + O(L)$ for some $(a, b) \in A \times B$.

If a node of $G_A$ or $H_B$ is in a tensor gadget, then we call it \textit{tensor node}. If a node of $G(a)$ or $H(b)$ is not in any tensor gadget, then we call it \textit{assembly node}. Let $G'_A$ be the tree $G_A$ with all tensor nodes removed, and we define $H'_B$ respectively.

We consider the following three cases:
\begin{enumerate}[\bf C{a}se 1.]
\item If none of tensor node of $G_A$ is in the $\LCST$, then
\[
\LCST(G_A, H_B) = \LCST(G'_A, H_B) \le \LCST(P_A, H_B) + \LCST(\mathcal{Z}_{K_D}, H_B),
\]
where $P_A$ is the path of length $mK_G + 1$ linked with the root of $T_A$. It is easy to see that $\LCST(P_A, H_B) \le mK_G + 1$.
Note that every pair of two tensor node from different tensor gadgets has distance at least $2K_G$, which is greater than the diameter of $T_A$, so the isomorphic subtree of $T_A$ in $H_B$ cannot contain nodes from more than one tensor gadgets. 
By noticing that $\LCST(\mathcal{Z}_{K_D}, H(b)) = O(K_D)$ for all $b \in B$ and $\LCST(\mathcal{Z}_{K_D},\mathcal{Z}_{K_G}) = O(K_G)$, we have $\LCST(\mathcal{Z}_{K_D}, H_B) = O(K_G)$. Thus $\LCST(G_A, H_B) = mK_G + O(K_G) = O(mK_G)$.

\item If $\LCST$ contains some tensor nodes of $G_A$, and all such tensor nodes are mapped to assembly nodes of $H_B$, then $\LCST(G_A, H_B)$ is no more than $\LCST(G'_A, H_B)$ plus the number of tensor nodes in the $\LCST$. By Case 1, $\LCST(G'_A, H_B) = O(m K_G)$.
Note that every two tensor nodes in $G_A$ has distance at most $K_G$ and any set of nodes of diameter $O(K_G)$ in $H_B$ can only have size $O(K_G)$, so the number of tensor nodes in $\LCST$ is at most $O(K_G)$. Thus $\LCST(G_A, H_B) = O(m K_G) + O(K_G) = O(m K_G)$.

\item If some tensor node in $G_A$ is mapped to a tensor node in $H_B$, then all the tensor nodes of $G_A$ in the $\LCST$ are in the same tensor gadget, and it also holds for $G_B$. This can be shown as follows: Let $u_1, u_2$ be two tensor node in $G_A$ that are mapped to tensor nodes $u'_1, u'_2$ in $H_B$, then
\begin{itemize}
\item $u_1, u_2$ are in the same tensor gadget. This is because that the minimum distance between two tensor nodes from different tensor gadgets in $G_A$ is at least $2K_D$ and at most $K_G$, but the distance between any two tensor nodes in $H_B$ is either $\le K_D$ or $\ge 2K_G$.
\item $u'_1, u'_2$ are in the same tensor gadget. This is because that the minimum distance between two tensor nodes from different tensor gadgets in $H_B$ is at least $2K_G$, but the distance between any two tensor nodes in $G_A$ is at most $K_G$.
\end{itemize}
Let $G(a)$ be the unique tensor gadget in $G_A$ that has nodes in the $\LCST$, and let $H(b)$ be the unique tensor gadget in $H_B$ that has nodes in the $\LCST$.  By Case 1, $\LCST(G'_A, H_B) = O(m K_G)$. Thus we have
\begin{align*}
\LCST(G_A, H_B) &\le \LCST(G_A \setminus G(a), H_B) + \LCST(G(a), H(b)) \\
&= O(mK_G) + \LCST(G(a), H(b)).
\end{align*}
\end{enumerate}

In any case, we can show that $\LCST(G_A, H_B) \le \LCST(G(a), H(b)) + O(L)$ for some $(a, b) \in A \times B$, which completes the proof.
\end{proof}

\begin{theorem} \label{thm:lcst-hard}
There exists an $O(N \poly(D))$-time reduction from an \GAPMAXTT{$\epsilon(N)$} instance with two sets of $N$ tensors of size $D$ to an instance of $o(\epsilon(N)^{-1})$-approximate Largest Common Subtrees on two (rooted or unrooted) binary trees of size $O(N \poly(D, \epsilon(N)^{-1}))$ and depth $O(\log^2 N \poly(D, \epsilon(N)^{-1}))$.
\end{theorem}
\begin{proof}
First we add dummy dimensions to each tensor such that the new size $C$ of every tensor is at least $\Omega(\epsilon^{-2}(N))$. We construct the two trees $A', B'$ as in Theorem \ref{thm:lcst-gadget}. By setting $L = C \log^2 N$, we have
\[
\LCST(A', B') = (\sqrt{C} \cdot s_{\max} + O(1)) \cdot C \log^2 N.
\]
Thus it reduces to distinguish $\LCST(A', B')$ from being $\ge C^{3/2} \log^2 N$ and $\le O(\epsilon(N) C^{3/2} \log^2 N)$, which can be solved by an $o(\epsilon(N)^{-1})$-approximation algorithm for $\LCST$. This conclusion also holds for $\RLCST$.
\end{proof}

\section{Equivalence in the Data Structure Setting}\label{sec:data-structure}

In this section, we establish the equivalence between $\equivclass$ problems in the data structure setting.

\begin{theorem}\label{theo:data-structure-formal}
	For the following data structure problems, if any of the following problems admits an algorithm with preprocessing time $T(N)$, space $S(N)$ and query time $Q(N)$, then all other problems admits a similar algorithm with preprocessing time $T(N) \cdot N^{o(1)}$, space $S(N) \cdot N^{o(1)}$ and query time $Q(N) \cdot N^{o(1)}$.
	
	\begin{itemize}
		\item $\NNS_\LCS$: Preprocess a database $\mathcal{D}$ of $N$ strings of length $D = 2^{(\log N)^{o(1)}}$, and then for each query string $x$, find $y \in \mathcal{D}$ maximizing $\LCS(x,y)$.
		
		\item Approx. $\NNS_\LCS$: Find $y \in \mathcal{D}$ s.t. $\LCS(x,y)$ is a $2^{(\log D)^{1- \Omega(1)}}$ approximation to the maximum value.
		
		\item Regular Expression Query: Preprocess a database $\mathcal{D}$ of $N$ strings of length $D = 2^{(\log N)^{o(1)}}$, and then for each query regular expression $y$, find an $x \in \mathcal{D}$ matching $y$.
		
		\item Approximate Regular Expression Query: for a query expression $y$, distinguish between: (1) there is an $x \in \mathcal{D}$ matching $y$; and (2) for all $x \in \mathcal{D}$, the hamming distance between $x$ and all $z \in L(y)$ is at least $(1-o(1)) \cdot D$.
	\end{itemize}
\end{theorem}
\begin{proof}
We show a reduction from $\NNS_\LCS$ to Approx. $\NNS_\LCS$ for illustration, and the proofs for the other pairs of problems are essentially the same.

By Lemma \ref{thm:lcs-nl}, there is a BP of poly-logarithmic size that accepts $(a, b, k)$ iff $\LCS(a, b) \ge k$. Then using a similar argument as in Theorem \ref{thm:lcs-pair-equiv}, we can show a reduction from an instance $\phi = (A, B)$ of the following problem to an Approx.~$\MAXLCSPAIR$ instance $\phi' = (A', B')$: given a set of strings $A$ and a set of string-integer pairs, determine whether there are $a \in A$ and $(b, k) \in B$ such that $\LCS(a, b) \ge k$.
Moreover, this reduction maps each element separetely, i.e., there exist two maps $f, g$ such that $A' = \{f(a) \mid a \in A\}, B' = \{g(b,k) \mid (b,k) \in B\}$, and both $f$ and $g$ are computable in $O(2^{\polylog(D)}) = 2^{(\log N)^{o(1)}}$ time and space for each string of length $D$.

Now suppose there is a data structure for Approx. $\NNS_\LCS$ with preprocessing time $T(N)$, space $S(N)$ and query time $Q(N)$. For a set of string $A$, we construct a data structure for $\NNS_\LCS$ as follows. In the preprocessing stage, 
we map all the strings in $A$ via $f$ and store them in a data structure $\mathcal{D}$ for Approx.~$\NNS_\LCS$.
For each query string, we do a binary search for the maximum LCS. For every length $k$ encountered, we first map the query string and the length $k$ via $g$ and then query it in the data structure $\mathcal{D}$. The time cost and space usage of the new data structure can be easily analyzed.
\end{proof}

A direct generalization of the above proof is that $\NNS_\LCS$ is actually the hardest $\NNS$ problem among all distance that can be computed in small space.
\begin{cor}
For every distance function $\dist$ that can be computed in poly-logarithmic space, the exact $\NNS$ problem with respect to $\dist$ ($\NNS_\dist$) can be reduced to $2^{(\log D)^{1- \Omega(1)}}$-approximate $\NNS_\LCS$ in near-linear time.
\end{cor}
\begin{remark}
	Here we assume that when the size parameter for $\NNS_\dist$ is $N$, the inputs to $\dist$ takes $2^{(\log N)^{o(1)}}$ bits to describe, and the output of $\dist$ takes $\polylog(N)$ bits to describe.
\end{remark}

Furthermore, basing on the hardness of solving $\BPSAT$, we can show that there may not be an efficient data structure for $\NNS_\LCS$ in the following sense:
\begin{theorem}\label{theo:hardness-for-datastructure}
Assuming the satisfiability for branching programs of size $2^{n^{o(1)}}$ cannot be decided in $O(2^{(1-\delta)n})$ for some $\delta > 0$, there is no data structure for Approx.~$\NNS_\LCS$ (or other data structure problems listed in Theorem~\ref{theo:data-structure-formal}) with preprocessing time $O(N^c)$ and query time $N^{1-\epsilon}$ for some $c, \epsilon > 0$.
\end{theorem}
\begin{proof}
Our proof closely follows the proof for Corollary 1.3 in \cite{Rub18}. We only prove that it is true for Approx.~$\NNS_\LCS$, the case for other data structures are similar. Assume such a data structure for Approx.~$\NNS_\LCS$ does exist. Now we show that there is an algorithm for approximate \MAXLCSPAIR that runs in $O(N^{2 - \delta'})$ time for some $\delta' > 0$.

Let $(A, B)$ be an instance of approximate \MAXLCSPAIR. Let $\gamma = 1/(2c)$. We partition the set $A$ into $M = O(N^{1-\gamma})$ subsets $A_1, \dots, A_M$, each of size $O(N^{\gamma})$. Now for each subset $A_i$, we build a data structure for Approx.~$\NNS_\LCS$, and query each $b \in B$ in the data structure. Finally, we take the maximum among all the query results. The total time for preprocessing is $O(M \cdot (N^{\gamma})^c) = O(N^{3/2})$ and that for query is $O(N \cdot M \cdot (N^\gamma)^{1-\epsilon}) = O(N^{2 - \epsilon \gamma})$.
\end{proof}
\section{Faster BP-SAT Implies Circuit Lower Bounds}\label{sec:algo-to-circuit-lowb}

In \cite{AHVW16}, Abboud et al. showed that faster exact algorithms for Edit Distance or LCS imply faster \BPSAT, and it leads to circuit lower bound consequences that are far stronger than any state of art.
Using a similar argument, strong circuit lower bounds can also be shown if any of \equivclass or \equivclasshard problems has faster algorithms, even for shaving a quasipolylog factor.

We apply the following results from \cite{AHVW16} to show the circuit lower bound consequences, which are direct corollaries from \cite{Wil13,Wil14ACC}:
\begin{theorem}[\cite{Wil13,Wil14ACC}]
\label{thm:truly-faster-implies}
Let $n \le S(n) \le 2^{o(n)}$ be a time constructible and monotone non-decreasing function. Let $\mathcal{C}$ be a class of circuits. If the satisfiability of a function of the form
\begin{itemize}
\item AND of fan-in in $O(S(n))$ of
\item arbitrary functions of fan-in $3$ of
\item $O(S(n))$-size circuits from $\mathcal{C}$
\end{itemize}
can be decided in $\DTIME[O(2^n/n^{10})]$ time, then $\ETIME^{\NP}$ does not have $S(n)$-size $\mathcal{C}$-circuits.
\end{theorem}
\begin{theorem}[\cite{Wil13}] \label{thm:implies-ntime-nc}
For the complexity class $\NTIME[2^{O(n)}]$, we have
\begin{enumerate}
\item If for every constant $k > 0$,
there is a satisfiability algorithm for bounded fan-in formulas of size $n^k$ running in $\DTIME[O(2^n/n^k)]$ time,
then $\NTIME[2^{O(n)}]$ is not contained in non-uniform $\NC^1$;
\item If for every constant $k > 0$,
there is a satisfiability algorithm for \NC-circuits of size $n^k$ running in $\DTIME[O(2^n/n^k)]$ time,
then $\NTIME[2^{O(n)}]$ is not contained in non-uniform $\NC$.
\end{enumerate}
\end{theorem}

First we show the circuit lower bound consequences if truly-subquadratic algorithm exists:

\begin{reminder}{of Corollary \ref{cor:cons-bpseth}}
If any of the \equivclass or \equivclasshard problems admits an $N^{2 - \eps}$ time deterministic algorithm (or $(NM)^{1-\epsilon}$ time algorithm for regular expression membership testing) for some $\eps > 0$, then $\textsf{E}^{\textsf{NP}}$ does not have:
	\begin{enumerate}
		\item non-uniform $2^{n^{o(1)}}$-size Boolean formulas,
		\item non-uniform $n^{o(1)}$-depth circuits of bounded fan-in, and
		\item non-uniform $2^{n^{o(1)}}$-size nondeterministic branching programs.
	\end{enumerate}
	Furthermore, $\NTIME[2^{O(n)}]$ is not in non-uniform \NC.
\end{reminder}
\begin{proof}
A truly-subquadratic time algorithm for \equivclass or \equivclasshard problems implies a $2^{(1 - \Omega(1))n}$-time algorithm for \BPSAT on branching program of size $2^{n^{o(1)}}$. Let $S(n) = 2^{n^{o(1)}}$. $O(S(n))$-size boolean formulas, $O(\log S(n))$-depth circuits, $2^{n^{o(1)}}$-size nondeterministic branching programs are all closed under AND, OR and NOT gates proscribed in Theorem~\ref{thm:truly-faster-implies}. Note that any formula of size $2^{n^{o(1)}}$ can be transformed into an equivalent $n^{o(1)}$-depth circuit~\cite{Spi71}, and any $n^{o(1)}$-depth circuit can be transformed into $2^{n^{o(1)}}$-size branching program by Barrington's Theorem~\cite{Barrington89}. Then all the consequences in Item 1, 2, 3 follow from Theorem~\ref{thm:truly-faster-implies}. Combining Item~2 and Theorem~\ref{thm:implies-ntime-nc}, we can obtain the consequence that $\NTIME[2^{O(n)}]$ is not in non-uniform \NC.
\end{proof}

We can also obtain results showing that even shaving a quasipolylog factor $2^{(\log \log N)^3}$ for problems in \equivclass and \equivclasshard can imply new circuit lower bound. First, it is easy to see that shaving a $(\log N)^{\omega(1)}$ factor can lead to new circuit lower bound by Theorem~\ref{thm:implies-ntime-nc}.

\begin{theorem}\label{thm:bp-pair-shaving}
If there is a deterministic algorithm for \BPSATPAIR on BP of size $S = 2^{(\log N)^{o(1)}}$ running in $O(N^2 \poly(S) / (\log N)^{\omega(1)})$ time, then the following holds:
\begin{enumerate}
\item For any constant $k > 0$, $\SAT$ on bounded fan-in formula of size $n^{k}$ can be solved in $O(2^n / n^{\omega(1)})$ deterministic time;
\item $\NTIME[2^{O(n)}]$ is not contained in non-uniform $\NC^1$.
\end{enumerate}
\end{theorem}
\begin{proof}
By Theorem~\ref{thm:implies-ntime-nc}, Item 1 implies Item 2, so we only need to show the conclusion in Item 1.

If \BPSATPAIR can be solved in $O(N^2 \poly(S)/(\log N)^{\omega(1)})$ time, then \BPSAT on BP of size $O(\poly(n))$ can be solved in
\[
O(2^{n} \poly(n)/n^{\omega(1)}) = O(2^{n} / n^{\omega(1)}).
\]
Note that any formula of size $n^k$ can be transformed into an equivalent BP of width $W = 5$ and length $T = O(n^{8k})$ (by rebalancing into a formula of depth $4k \log n$~\cite{Spi71} and using Barrington's Theorem~\cite{Barrington89}). Thus $\SAT$ on bounded fan-in formulas of size $n^k$ can also be solved in $O(2^{n} / n^{\omega(1)})$.
\end{proof}

A part of our reductions from \BPSATPAIR to problems in \equivclass can be summerized below.
In the rest of this section, for each problem in \equivclass (but except \BPSATPAIR), we use the variable $N$ to denote the number of elements in each set, and $D$ to denote the maximum length (or size) of each element. We exclude \BPSATPAIR here because the size $S$ of BP is more important than $D$ in \BPSATPAIR.

\begin{cor} \label{cor:bp-to-equiv}
For every problem $\mathcal{P}$ in \equivclass except \BPSATPAIR:
\begin{itemize}
\item If $\mathcal{P}$ is a decision problem, then there is an $O(N \poly(D))$-time reduction from \OAPTrestri to $\mathcal{P}$;
\item If $\mathcal{P}$ is an approximate problem, then for every $\epsilon(D) = \Omega(D^{-1/2})$, there is an $O(N \poly(D))$-time reduction from \GAPMAXTTrestri{$\epsilon(D)$} or \GAPMAXTTrestri{$\epsilon(D)$} to $\mathcal{P}$ with approximation ratio $o(\epsilon(D)^{-1})$, and each element has size $O(\poly(D))$.
\end{itemize}
And any reduction here preserves the value of $N$.
\end{cor}

\begin{reminder}{Theorem \ref{theo:shave-logs-equiv}}
For $D = 2^{(\log N)^{o(1)}}$, if there is an 
\[
O\left(N^2\poly(D) / 2^{(\log \log N)^3}\right) ~\text{or}~ O\left(N^2 / (\log N)^{\omega(1)}\right)
\]
time deterministic algorithm for the decision, exact value, or $O(\polylog(D))$-approximation problems in \equivclass, then the same consequences in Theorem \ref{thm:bp-pair-shaving} follows.
\end{reminder}
\begin{proof}
By Corollary \ref{cor:bp-to-equiv} and the fact that exact value problem can be trivially reduced to its approximation version, we only need to show that this statement is true for \OAPT and \GAPMAXTT{$(\log D)^c$} for every $c > 0$ (the proof for \GAPMINTT{$(\log D)^c$} should be similar).

Note that all our reductions here preserve the value of $N$. If there is an $O(N^2 / (\log N)^{\omega(1)})$-time algorithm, then \BPSATPAIR can also be solved in $O(N^2 / (\log N)^{\omega(1)})$-time and the consequences in Theorem \ref{thm:bp-pair-shaving} follows.

Now consider the case that a $O(N^2\poly(D) / 2^{(\log \log N)^3})$-time algorithm exists.
Recall that the hard instances of \BPSATPAIR we constructed in the proof of Theorem \ref{thm:bp-pair-shaving} is on BP of width $W = O(1)$ and length $T = O(\poly(n)) = O(\polylog(N))$. By Theorem \ref{thm:hard-oapt}, we know that this instance can be near-linear time reduced to an \OAPT instance with
\[
D = 2^{O(\log W \log T)} = 2^{O(\log \log N)} = \polylog(N).
\]
Thus shaving an $O(2^{(\log \log N)^3})$ factor to \OAPT implies an $O(N^2 / (\log N)^{\omega(1)})$-time algorithm for \BPSATPAIR.

By Theorem \ref{thm:hard-gap-tt}, for $\epsilon = \log^{-3c}(T)$, we know that a hard instance of \BPSATPAIR can also be near-linear time reduced to an \GAPMAXTT{$\epsilon$} instance with (adding dummy dimensions if necessary)
\[
D = 2^{\Theta(\log^2 W \log^2 T(\log \log W + \log \log T + \log\epsilon^{-1}))} = 2^{\Theta(\log^2 T \log \log T)}.
\]
Then we have $(\log D)^c = o(\epsilon^{-1})$, and thus shaving an $O(2^{(\log \log N)^3})$ factor to \GAPMAXTT{$(\log D)^c$} implies an algorithm for the hard instances of \BPSATPAIR running in the following time:
\begin{align*}
O(N^2\poly(D) / 2^{(\log \log N)^3}) &= O(N^2 \cdot 2^{\Theta(\log^2 T \log \log T)} / 2^{(\log \log N)^3}) \\
&= O(N^2 / (\log N)^{\omega(1)}).
\end{align*}
\end{proof}

For \equivclasshard problems, recall that part of our reduction can be summerized below:
\begin{cor} \label{cor:bp-to-hard}
For every problem $\mathcal{P}$ in \equivclasshard:
\begin{itemize}
\item If $\mathcal{P}$ is a decision problem, then there is an $O(N \poly(D))$-time reduction from \OAPTrestri to $\mathcal{P}$ on input of length $O(N \poly(D))$;
\item If $\mathcal{P}$ is an approximate problem, then for every $\epsilon(N)$, there is an $O(N \poly(D, \epsilon(N)^{-1}))$-time reduction from \GAPMAXTTrestri{$\epsilon(N)$} or \GAPMAXTTrestri{$\epsilon(N)$} to $\mathcal{P}$ with approximation ratio $o(\epsilon(N)^{-1})$ on input of length $O(N \poly(D, \epsilon(N)^{-1}))$.
\end{itemize}
\end{cor}

Then we can obtain the following result:

\begin{reminder}{Theorem \ref{theo:shave-logs-hard}}
If there is an deterministic algorithm for any decision, exact value or $\polylog(N)$-approximation problems among \equivclasshard problems listed in Theorem~\ref{theo:informal-hard} running in running in
\[
O\left(N^2 / 2^{\omega(\log \log N)^3}\right)
\]
time (or $O\left(NM / 2^{\omega(\log \log(NM))^3}\right)$ time for Regular Expression Membership Testing), then the same consequences in Theorem~\ref{theo:shave-logs-equiv} follows.
\end{reminder}
\begin{proof}
By Corollary \ref{cor:bp-to-hard} and the fact that exact value problem can be trivially reduced to its approximation version, we only need to show that this statement is true for \OAPT and \GAPMAXTT{$(\log N)^c$} for every $c > 0$. 

The proof for \OAPT is similar as in Theorem \ref{theo:shave-logs-equiv}. For \GAPMAXTT{$(\log N)^c$}, we know that the hard instances of \BPSATPAIR in Theorem \ref{thm:bp-pair-shaving} can be reduced to a \GAPMAXTT{$\epsilon$} instance with
\[
D = 2^{O(\log^2 W \log^2 T(\log \log W + \log \log T + \log\epsilon^{-1}))} = 2^{O(\log \log N)^3}
\]
for $\epsilon = (\log N)^c$. Thus shaving an $O(2^{\omega(\log \log N)^3})$ factor to \GAPMAXTT{$(\log N)^c$} implies an algorithm for the hard instances of \BPSATPAIR running in $O(N^2 / (\log N)^{\omega(1)})$ time.
\end{proof}

\section{Derandomization Implies Circuit Lower Bounds}\label{sec:derand-to-circuit-lowb}

For the some problems $\mathcal{A}$ like \textit{Longest Common Subsequence}, despite its approximating for the pair version of $\mathcal{A}$ (Approximate \MAXAPAIR{$\mathcal{A}$}) is subquadratically equivalent to \MAXTT, 
it is still hard to find a reduction from approximating $\mathcal{A}$. The main barrier is when trying to construct gadgets to reduce Approximate \MAXAPAIR{$\mathcal{A}$} to Approximate $\mathcal{A}$, the contribution to the final result for just one pair is too small to make a large approximating gap.

To overcome this barrier, we follows from \cite{AbboudR18} to define \SUPERGAPMAXTT{$\epsilon(N)$}, which is a variant of \GAPMAXTT{$\epsilon(N)$} with a large fraction of pairs having perfect Tropical Similarities:
\begin{definition}[\SUPERGAPMAXTT{$\epsilon$}] \label{def:att}
Let $t$ be an even number and $d_1 = d_2 = \cdots = d_t = 2$. Given two sets of tensors $A, B \in \{0, 1\}^{d_1 \times \cdots \times d_t}$ of size $D = 2^t$, distinguish between the following:
\begin{itemize}
\item \textbf{Completeness:} A $(1 - 1/\log^{10} N)$-fraction of the pairs of $a \in A, b \in B$ have a perfect Tropical Similarity, $s(a, b) = 1$;
\item \textbf{Soundness:} Every pair has low Tropical Similarity score, $s(a, b) < \epsilon$.
\end{itemize}
where $\epsilon$ is a threshold that can depend on $N$ and $D$.
\end{definition}

In \cite{AbboudR18}, Abboud and Rubinstein has shown that \SUPERGAPMAXTT{$o(1)$} can be reduced to $O(1)$-approximate \LCS. Using the same reduction, we have the following corollary for arbitrary approximation ratio:
\begin{theorem}[\cite{AbboudR18}] \label{thm:super-gap-lcs-red}
Given an \SUPERGAPMAXTT{$\epsilon(N)$} instance on $N$ tensors of size $D$, we can construct two strings $x, y$ of length $ND$ in $O(N \poly(D))$ deterministic time such that:
\begin{itemize}
\item If $(1 - 1/\log^{10} N)$-fraction of the pairs have a perfect Tropical Similarity, then $\LCS(x, y) > (1/3)N \sqrt{D}$;
\item If every pair has low Tropical Similarity score, then $\LCS(x, y) < 2\epsilon(N) N\sqrt{D}$.
\end{itemize}
Thus, if there is an $\epsilon(N)^{-1}$-approximation algorithm for such kind of $(x, y)$ pairs, then there is a faster algorithm for \SUPERGAPMAXTT{$(\epsilon(N)/6)$}.
\end{theorem}
\begin{proof}
We construct strings $G(a), H(b)$ as tensor gadgets for each tensor $a \in A, b \in B$ as in the reduction in \cite{AbboudR18} (stated in Theorem \ref{thm:lcs-gadget}). Then we construct the final strings $x, y$ by concatenating all the tensor gadgets. Using a similar argument as in \cite{AbboudR18}, we can show that if there are at least $(1 - 1/\log^{10} N) N^2$ pairs of tensors with perfect Tropcial Similarities, then $\LCS(x, y) > (1 - 1/\log^{10}N) \cdot N \sqrt{D} / 2$; if every pair has low Tropical Similarity score, then $\LCS(x, y) < \epsilon(N) \cdot 2N\sqrt{D}$.
\end{proof}

There is no obvious reduction from \BPSATPAIR to \SUPERGAPMAXTT{$\epsilon(N)$}, and a randomized algorithm can even solve \SUPERGAPMAXTT{$\epsilon(N)$} in nearly linear time. But finding a \textit{deterministic} algorithm for \SUPERGAPMAXTT{$\epsilon(N)$} is still hard: as noted by Abboud and Rubinstein in \cite{AbboudR18}, a truly-subquadratic time deterministic algorithm for \SUPERGAPMAXTT{$\epsilon(N)$} can imply some circuit lower bound for $\ETIME^\NP$. Combining their ideas with the connection between Tropical Tensors and \BPSAT we established, we can show even stronger circuit lower bounds if such algorithm exists.

We base our proof on the following results in the literature:
\begin{theorem}[\cite{BV14}] \label{thm:ac-prob-enp}
Let $F_n$ be a set of function from $\{0, 1\}^n$ to $\{0, 1\}$ that are efficiently closed under projections. If the acceptance probability of a function of the form
\begin{itemize}
\item AND of fan-in in $n^{O(1)}$ of
\item OR's of fan-in $3$ of
\item functions from $F_{n + O(\log n)}$
\end{itemize}
can be distinguished from being $= 1$ or $\le 1/n^{10}$ in $\DTIME[2^n / n^{\omega(1)}]$, then there is a function $f \in \ETIME^{\NP}$ on $n$ variables and $f \notin F_n$.
\end{theorem}

\begin{theorem}[\cite{Wil13, BV14}] \label{thm:ac-prob-ntime}
	If the acceptance probability of a function from $\NC^1$ can be distinguished from being $ = 1$ or $\le 1/n^{10}$ in $\DTIME[2^n / n^{\omega(1)}]$, then $\NTIME[2^{O(n)}]$ is not contained in $\NC^1$.
\end{theorem}

\begin{theorem} \label{thm:hard-super-gap-tt}
Let \ACBPSAT be the following problem: given a branching program $P$ of length $T$ and width $W$ on $n$ inputs, distinguish the acceptance probability of $P$ from being $ = 1$ or $\le 1/n^{10}$. 

There is a reduction from \ACBPSAT to  \SUPERGAPMAXTT{$\epsilon$} on two sets of $N = 2^{n/2}$ tensors of size
$D = 2^{O(\log^2 W \log^2 T (\log \log W + \log \log T + \log \epsilon^{-1}))}$,
and the reduction runs in $O(N\poly(D))$. Here $\epsilon$ is a threshold value that can depend on $N$ (but cannot depend on $D$).
\end{theorem}

\begin{reminder}{Theorem \ref{thm:approx-lcs-cons}}
	The following holds for deterministic approximation to $\LCS$:
	
	\begin{enumerate}
		\item A $2^{(\log N)^{1 - \Omega(1)}}$-approximation algorithm in $N^{2 - \delta}$ time for some constant $\delta > 0$ implies that $\textsf{E}^{\NP}$ has no $n^{o(1)}$-depth bounded fan-in circuits;
		
		\item A $2^{o(\log N / (\log \log N)^2)}$-approximation algorithm in $N^{2 - \delta}$ time for some constant $\delta > 0$ implies that $\NTIME[2^{O(n)}]$ is not contained in non-uniform $\NC^1$;
		
		\item An $O(\polylog(N))$-approximation algorithm in $N^2 / 2^{\omega(\log \log N)^3}$ time implies that $\NTIME[2^{O(n)}]$ is not contained in non-uniform $\NC^1$.
	\end{enumerate}
	
\end{reminder}
\begin{proof} By Theorem \ref{thm:ac-prob-enp} and Theorem~\ref{thm:ac-prob-ntime}, for Item 1, it is sufficient to show \ACBPSAT on BP of length $2^{n^{o(1)}}$ and width $O(1)$ on $n$ inputs can be solved in $2^{(1 - \Omega(1))n}$ time; for Item 2 and 3, it is sufficient to show \ACBPSAT on BP of length $O(\poly(n))$ and width $O(1)$ on $n$ inputs can be solved in $2^{n}/n^{\omega(1)}$ time (the former scale of BP is able to simulate $n^{o(1)}$-depth circuit, while the later one is able to simulate $\NC^1$ by Barrington's Theorem \cite{Barrington89}).

\item
\paragraph*{Item 1.}
Assume there exists a $2^{(\log N)^{1 - c}}$-approximation algorithm for \LCS in $N^{2 - \delta}$ time for some $c > 0$ and $\delta > 0$.
By Theorem~\ref{thm:hard-super-gap-tt}, \ACBPSAT on BP of length $T = 2^{n^{o(1)}}$ and width $W = O(1)$ on $n$ inputs can be reduced to \SUPERGAPMAXTT{$(2^{-(\log K)^{1 - c}}/6)$} on $K = 2^{n/2}$ tensors of size 
\[
	D =  2^{O(n^{o(1)} \cdot (o(\log n) + (\log K)^{1 - c}))} = 2^{n^{1 - c + o(1)}}.
\] 
Then by Theorem~\ref{thm:super-gap-lcs-red}, \SUPERGAPMAXTT{$(2^{-(\log K)^{1 - c}}/6)$} can be reduced to $2^{(\log K)^{1 - c}}$-approximate LCS for strings of length 
\[
N = KD = 2^{n/2 + n^{1-c+o(1)}} = 2^{(1/2 + o(1))n}.
\] 
By our assumption, the last problem can be solved in $N^{2-\delta}$ time, so \ACBPSAT on branching program of length $2^{n^{o(1)}}$ and width $O(1)$ on $n$ inputs can be solved in 
$2^{(1 - \delta/2 + o(1))n}$
time. Applying Theorem~\ref{thm:ac-prob-enp} completes the proof.

\paragraph*{Item 2.} Assume there exists a $2^{f(\log N)}$-approximation algorithm for \LCS in $N^{2 - \delta}$ time for some constant $\delta > 0$ and some function $f(k) = o(k/\log^2 k)$. 
Let $g(k) = 2f(k) + \log k$. Then we have $g(\log N) = o(\log N/(\log \log N)^2)$ and
\[
	2^{f((1+o(1))\log K)} \le 2^{(1+o(1)) f(\log K)} \le 2^{2 f(\log K)} = o(2^{g(\log K)}).
\]
By Theorem~\ref{thm:hard-super-gap-tt}, \ACBPSAT on BP of length $T = O(\poly(n))$ and width $W = O(1)$ on $n$ inputs can be reduced to \SUPERGAPMAXTT{$2^{-g(\log K)}$} on $K = 2^{n/2}$ tensors of size 
\[
	D =  2^{O(\log^2 n \cdot (\log \log n + g(\log K)))} = 2^{O(\log^2 n \cdot o(\log K / (\log \log K)^2))} = 2^{o(n)}.
\] 
Then by Theorem~\ref{thm:super-gap-lcs-red}, \SUPERGAPMAXTT{$2^{-g(\log K)}$} can be reduced to $o(2^{g(\log K)})$-approximate LCS for strings of length $N = KD =  2^{(1/2 + o(1))n}$.
Note that $2^{f(\log(K^{1+o(1)}))} = o(2^{g(\log K)})$. Thus by our assumption, the last problem can be solved in $N^{2-\delta}$ time, which means \ACBPSAT on branching program of length $2^{n^{o(1)}}$ and width $O(1)$ on $n$ inputs can be solved in $2^{(1 - \delta/2 + o(1))n}$ time. Applying Theorem~\ref{thm:ac-prob-ntime} completes the proof.

\paragraph*{Item 3.} Assume there exists a $\log^c(N)$-approximation algorithm for \LCS in $N^2 / 2^{\omega(\log \log N)^3}$ time for some $c > 0$. Using a similar calculation as in Theorem \ref{theo:shave-logs-hard}, we know that 
\ACBPSAT on branching program of length $O(\poly(n))$ and width $O(1)$ on $n$ inputs can be solved in 
\[
N^2/2^{\omega(\log \log N)^3} \le 2^{n + O(\log^3 n)} / 2^{\omega(\log^{3} n)} \le 2^n / n^{\omega(1)}
\] 
time. Applying Theorem~\ref{thm:ac-prob-ntime} completes the proof.
\end{proof}

\section*{Acknowledgment}
The authors would like to thank Ryan Williams for helpful discussion. The authors thank Ray Li for pointing out a mistake in Theorem 11.2 in an earlier version of this paper.



\bibliographystyle{alpha}
\bibliography{main}

\end{document}